\documentclass[12pt]{article}
\usepackage{amsmath}
\usepackage{graphicx,psfrag,epsf}
\usepackage{enumerate}
\usepackage{natbib}
\usepackage{url} 

\newcommand{\blind}{0}

\RequirePackage{amsthm,amsfonts}
\RequirePackage[colorlinks,citecolor=blue,urlcolor=blue]{hyperref}

\usepackage{appendix}

\usepackage[noend]{algpseudocode}
\usepackage{algorithm}
\usepackage{mathtools}
\usepackage{colortbl}
\usepackage{comment}
\usepackage{multirow}

\usepackage{tikz}
\usepackage{nicematrix}
\usepackage{booktabs}

\usepackage{bbm}

\usepackage{wrapfig}

\usepackage{placeins}

\newtheorem{theorem}{Theorem}
\newtheorem{assumption}{Assumption}
\newtheorem{lemma}{Lemma}

\DeclareMathOperator*{\argmax}{arg\,max}
\DeclareMathOperator*{\argmin}{arg\,min}
\newcommand{\code}[1]{\texttt{#1}}
\newcommand{\1}{\mathbf 1}
\newcommand{\C}{\mathcal C}

\newcommand{\R}{\mathbb R}
\newcommand{\E}{{\bf E}}

\renewcommand{\O}{\mathcal O}
\renewcommand{\S}{\mathbb S}
\renewcommand{\l}{{\ell}}
\newcommand{\M}{{\mathcal M}}
\newcommand{\K}{{\Tilde{K}}}

\renewcommand{\d}{\Tilde{d}}

\renewcommand{\P}{{\boldsymbol P}}
\renewcommand{\H}{\mathbb H}
\newcommand{\I}{\mathbf I}
\newcommand{\U}{\mathcal U}

\newcommand{\logodds}{\mbox{logodds}}

\newcommand{\ping}[1]{{\color{magenta}[addressed]}}

\definecolor{lightblue}{RGB}{232, 244, 255}
\definecolor{lightgray}{RGB}{245, 246, 250}

\allowdisplaybreaks

\begin{document}

\def\spacingset#1{\renewcommand{\baselinestretch}%
{#1}\small\normalsize} \spacingset{1}


\if0\blind
{
  \title{\bf Network Cross-Validation and Model Selection via Subsampling}
  \author{Sayan Chakrabarty\\
    Department of Statistics, University of Michigan, USA\\
    Srijan Sengupta\\
    Department of Statistics, North Carolina State University, USA\\
    Yuguo Chen \\
    Department of Statistics, University of Illinois Urbana-Champaign, USA}
    \date{}
  \maketitle
} \fi

\if1\blind
{
  \bigskip
  \bigskip
  \bigskip
  \begin{center}
    {\LARGE\bf Network Cross-Validation and Model Selection via Subsampling}
\end{center}
  \medskip
} \fi

\bigskip
\begin{abstract}
Complex and larger networks are becoming increasingly prevalent in scientific applications in various domains. Although a number of models and methods exist for such networks, cross-validation on networks remains challenging due to the unique structure of network data. 
In this paper, we propose a general cross-validation procedure called \code{NETCROP} (\textbf{NET}work \textbf{CR}oss-Validation using \textbf{O}verlapping \textbf{P}artitions --- 
\href{https://cran.r-project.org/web/packages/netOP/}{R Package}
). 
The key idea is to divide the original network into multiple subnetworks with a shared overlap part, producing training sets consisting of the subnetworks and a test set with the node pairs between the subnetworks. This train-test split provides the basis for a network cross-validation procedure that can be applied on a wide range of model selection and parameter tuning problems for networks. The method is computationally efficient for large networks as it uses smaller subnetworks for the training step. We provide methodological details and theoretical guarantees for several model selection and parameter tuning tasks using \code{NETCROP}.
Numerical results demonstrate that \code{NETCROP} performs accurate cross-validation on a diverse set of network model selection and parameter tuning problems.
The results also indicate that \code{NETCROP} is computationally much faster while being often more accurate than the existing methods for network cross-validation.
\end{abstract}


\spacingset{1.45} 

\section{Introduction}\label{sec:intro}
Most modern systems involve interactions between multiple agents, often modelled using networks. 
Network data arise in various scientific fields such as social networks \citep{fb-network}, biomedical networks \citep{komolafe2022scalable}, epidemiological networks \citep{dasgupta2022scalable}, etc. 
Numerous network models, such as stochastic blockmodels (SBMs, \citet{HOLLAND1983109}), degree corrected blockmodels (DCBMs, \citet{KarrerDCBM}), random dot product graphs (RDPGs, \citet{RDPG_survey}) and latent space models \citep{Hoff2002LatentSA} among others, are utilized in the analysis of network data. SBM and DCBM categorize nodes into communities, whereas RDPG and latent space models associate nodes with their positions in a latent space. 
Estimation techniques for network models encompass a spectrum of approaches, including spectral methods like spectral clustering and adjacency spectral embedding \citep{rohe, lei2015, bickel_sc, sengupta2015spectral}, likelihood-based methods \citep{likeli1, likelihood,senguptapabm}, and Markov chain Monte Carlo methods 
\citep{Hoff2002LatentSA, Handcock2007}, which are applied for parameter estimation in these models.

Given an observed network, selecting an appropriate network model and the corresponding methodology is an indispensable task.
Cross-validation is one of the most commonly used statistical tools for such model selection and parameter tuning problems for traditional (non-network) data.
In a standard cross-validation, the data is randomly split into a training set  and a test set. The candidate models (or tuning parameter values) are trained on the training set and evaluated on the test set to select the most appropriate one.
The effectiveness of standard cross-validation follows from its simple and flexible design that covers a variety of settings. For a modern take on cross-validation, and its advantages and disadvantages, see \citet{standard-cv} and the citations therein. However, cross-validation for model selection and parameter tuning in network data is very little explored. The major challenge for a network cross-validation method stems from the complex nature of the data. Usually, only one instance of the network is observed, which makes defining data points for cross-validation difficult. 
\cite{cv_comm} developed an algorithm called Network Cross-Validation (\code{NCV}), where the nodes are 
split into multiple folds.
For each fold, the rectangular submatrix of the network adjacency matrix corresponding to the nodes outside the fold is used as the training set, and the square submatrix corresponding to the nodes in the fold is used as the test set. \code{NCV} is designed specifically for selecting the number of communities in blockmodels with consistency results only for SBM.
\cite{cv_edge} considered the node pairs as the data points in their network cross-validation algorithm, called Edge Cross-Validation (\code{ECV}), which was shown to consistently recover the number of communities in SBM and the latent space dimension in RDPG along with some additional numerical examples.

Both \code{NCV} and \code{ECV} suffer from certain drawbacks, leaving a gap in the network cross-validation literature. \code{NCV} takes each row of the adjacency matrix as a sampling unit and subsamples a proportion of the rows in each fold. 
For a network with $n$ nodes,
it still requires working on rectangular matrices of order $np \times n$ for some fraction $p$, making the training sets very large and the training steps very slow. It is also specific to estimating the number of communities in blockmodels.
In contrast, as long as the network model satisfies a low rank mean assumption, \code{ECV} is applicable to a wider class of network cross-validation problems by employing node pair sampling to extract a subset of the adjacency matrix as the training data and using matrix completion algorithms to impute the missing entries. The candidate models are then estimated from the imputed matrix, and the held-out entries of the adjacency matrix serve as test data. However, if the original network is binary, the imputed values in \code{ECV} are very likely to be non-binary, making the method unsuitable for estimators that require a binary adjacency matrix as input, such as Bernoulli likelihood-based methods \citep{likeli2, Hoff2002LatentSA}. Instead, it is primarily applicable to methods that allow non-binary adjacency matrices, such as those based on spectral decomposition \citep{rohe, ms1}.
Furthermore, for the adjacency matrix estimation step to be accurate, it requires a large subsample of node pairs, usually around $90\%$ of all the node pairs in the network. 
This can potentially cause overfitting of the network and slow computations for large networks. The matrix completion algorithm used in \code{ECV} affects its theoretical and computational performance as well. 
The authors of both algorithms recommended repeating their methods $20$ times and choosing the most frequently occurring outcome model to stabilize the results. This stabilization slows down the algorithms even further.

In this paper, we develop a computationally efficient subsampling based network cross-validation procedure. The method is called \textbf{NET}work \textbf{CR}oss-Validation using \textbf{O}verlapping \textbf{P}artitions (\code{NETCROP}). Given an observed network and a candidate set of models, \code{NETCROP} splits the network into suitably chosen training and test sets, fits all the candidate models on the training sets, and tests them on the test set using a suitable loss function. \code{NETCROP} first randomly subsamples a set of \textsl{overlap nodes} and partitions the remaining nodes into $s$ \textsl{non-overlap partitions}. Then it forms $s$ subsets of nodes, where each subset consists of the overlap nodes and the nodes from the corresponding non-overlap part. The $s$ subnetworks induced by the nodes in each of the $s$ subsets of nodes are the training subnetworks.
These subsamples are much smaller than the original samples, which makes the method scalable \citep{politis1999subsampling,sengupta2016subsampled,ganguly2023scalable}.
Every candidate model is fitted to each subnetwork to estimate the corresponding model parameters. The parameters are then matched using the overlap nodes.
For example, the community labels in a blockmodel are identifiable only up to permutation of the labels, and the latent positions in RDPG or latent space models are identifiable up to their orthogonal rotations and translations. After matching the parameter estimates using the overlap part, we combine the estimates from the subnetworks. Then, the estimated parameters are used to predict the edge probabilities for the node pairs in the test set. This test set consists of the node pairs between the non-overlap parts. Finally, for each candidate model, the test set prediction error is computed using a suitable loss function, and the model with the lowest test set loss is chosen as the outcome.

Since \code{NETCROP} splits the observed network into a training set and a test set, it inherits the flexible nature of standard cross-validation, making it adaptable to a wide range of network cross-validation problems including model selection and parameter tuning.
The unique design of \code{NETCROP} stabilizes the estimates of the trained model parameters by aggregating estimates from multiple subnetworks \citep{breiman1996bagging, Soloffbagging2024}, thus requiring a much smaller number of repetitions than \code{NCV} and \code{ECV} for stability. 
Recall from the last paragraph that some models can have unidentifiable parameters. The inclusion of the overlap nodes makes it possible to match the estimates of such unidentifiable parameters.
These factors, combined, help \code{NETCROP} avoid the problem of overfitting, while making it significantly faster and more accurate than \code{NCV} and \code{ECV}. \code{NETCROP} is also naturally parallelizable as the computations on different subnetworks can be done on different processors in a multiprocessor system.

This paper proposes \code{NETCROP} as a versatile network cross-validation tool for model selection and parameter tuning problems. We establish the theoretical consistency of \code{NETCROP} for selecting the number of communities in SBM, achieving slightly stronger bounds under milder assumptions compared to \code{NCV}. Additionally, we provide, to the best of our knowledge, the first theoretical consistency result for cross-validation in DCBM. 
Furthermore, we prove the theoretical consistency of \code{NETCROP} for selecting the latent space dimension in RDPG.
Numerical examples include selecting the degree heterogeneity and the number of communities in blockmodels, the latent space dimensions in RDPG and latent space models, and the tuning parameter for regularized spectral clustering. Although several model-based methods exist for estimating the number of communities and assessing the goodness of fit for SBM and DCBM \citep{lei2016goodness, ma2021determining, cerqueira2023consistent, deng2023subsamplingbased},
there is a scarcity of general and computationally efficient cross-validation techniques for network model selection and parameter tuning in the current literature. With this in consideration, we compare \code{NETCROP} only with the other network cross-validation algorithms (\code{NCV} and \code{ECV}) throughout the paper. 
All codes for \code{NETCROP} are available in an R Package (\url{https://cran.r-project.org/web/packages/netOP/}) and the scripts to replicate the numerical experiments in this paper are publicly available in this GitHub Repository: \url{https://github.com/sayan-ch/NETCROP.git}.

The paper is organized as follows. Section \ref{sec:method} contains an overview of \code{NETCROP}. In Section \ref{sec:example}, we discuss applications of \code{NETCROP} for four model selection and tuning parameter problems.
We present theoretical results for the SBM, DCBM, and RDPG cases in Section \ref{sec:theory}. The numerical performance of \code{NETCROP} in comparison with other network cross-validation methods for simulated and real networks is in Section \ref{sec:numerical}. 
Section \ref{sec:discussion} contains discussion and concluding remarks on this paper.

\section{Methodology}\label{sec:method}
First, we introduce some notations that are used throughout the article. 
The term \textsl{simple network} indicates an undirected, unweighted graph without any self-loops or multiple edges between nodes. Although \code{NETCROP} is described and the theoretical results are derived for simple networks, they can be extended mutatis mutandis for directed or weighted networks. Unless otherwise specified, the word \textsl{network} refers to a simple network in this paper.

Define $[n] = \{1, \ldots, n\}$ and $[n_1 : n_2] = \{n_1, \ldots, n_2\}$ for any natural numbers $n$ and $n_1 \leq n_2$. $I_n$, $\mathbf{1}_n$, and $J_n$ denote the $n \times n$ identity matrix, the $n \times 1$ vector of all $1$, and the $n \times n$ matrix of all $1$, respectively. The $n$ in the subscript is dropped whenever it is evident from the context. $\I(\cdot)$ denotes the indicator function of the event inside it. For $a \in \R^d$, $\lVert a \rVert$ is defined as the $\ell_2$ norm. For a matrix $M \in \R^{n \times d}$, the spectral norm, the Frobenius norm, matrix $(2,1)$ and $(2,\infty)$ norms are respectively denoted by $\lVert M \rVert$, $\lVert M \rVert_F$, $\lVert M \rVert_{2,1}$ and $\lVert M \rVert_{2, \infty}$.
For any countable set $S$, $\lvert S \rvert$ denotes its cardinality.
For any two subsets $I \subset [n]$ and $J \subset [d]$, $M_{IJ}, M_{I\cdot}$ and $M_{\cdot J}$ denote the submatrices of $M$ with the rows in $I$ and the columns in $J$, the rows in $I$ and all the columns, and the columns in $J$ and all the rows, respectively. 
$M_{i \cdot}$ indicates the $i$th row of $M$, represented as a $d \times 1$ vector. 
We define the squared error loss as $\ell_2(a, p) = (a-p)^2$ for $a, p \in \R$.
The space of all permutation matrices of order $k \times k$ is denoted by $\H_k$. The space of all orthogonal matrices of order $d\times d$ is denoted by $\O_d$. 
For any two real sequences $\{a_n\}$ and $\{b_n\}$, $a_n = O(b_n)$ if $\limsup\limits_{n \to \infty} \lvert a_n/b_n \rvert < \infty$ and $a_n = \Omega(b_n)$ if $\liminf\limits_{n \to \infty} \lvert a_n/b_n \rvert > 0$. 
Define $b_n = \omega(a_n)$ as $\limsup\limits_{n \to \infty} \lvert a_n/b_n \rvert = 0$ and $a_n \asymp b_n$ as $a_n = O(b_n)$ and $b_n = O(a_n)$. 
For two sequences of random variables $\{X_n\}$ and $\{Y_n\}$, $X_n = O_p(Y_n)$ or $Y_n = \Omega_p(X_n)$ if for every $\epsilon > 0$, there exist $M_{\epsilon}$ and $N_{\epsilon}$ such that $P\left( \lvert X_n/Y_n \rvert \leq M_{\epsilon} \right) > 1 - \epsilon$ for $n \geq N_{\epsilon}$.

Given a network of size $n$, $A$ represents the corresponding $n \times n$ adjacency matrix. For any $(i,j) \in [n] \times [n]$, $A_{ij} = 1$ if there is an edge between nodes $i$ and $j$ and $0$ otherwise. The set of nodes in the network is denoted by $S = [n]$. A subnetwork induced by a subset of nodes $S_0 \subset [n]$ is represented by the subadjacency matrix $A_{S_0} \coloneqq A_{S_0S_0}$. Throughout the paper, we assume that $A_{ij}$'s are independent (conditional on relevant model parameters) Bernoulli random variables with the probability of observing an edge between two nodes $i$ and $j$ being $E(A_{ij}) = P_{ij}$. The probability matrix, also called the expected adjacency matrix, is represented by $P = (P_{ij})_{(i,j) \in [n] \times [n]} = E(A)$. 
We assume that the link probability matrix $P \coloneqq P_n = \rho_n P^0_n \in [0, 1]^{n \times n}$ for some sparsity parameter $\rho_n > 0$ and a sequence of matrices $P^0_n \in [0, \infty)^{n \times n}$ with growing dimensions such that $P^0_{n}$ is a submatrix of $P^0_{n+1}$. 
Thus, the elements of $P = {P_n}$, a sequence of matrices with growing dimensions, depend on $n$ only through the sparsity parameter $\rho_n$. 
For all the theoretical results and derivations, the number of communities $K$ in SBM and DCBM, and the latent space dimension $d$ in RDPG are assumed to be fixed.

For the proposed method \code{NETCROP} (to be defined in the following subsection), $s$ denotes the number of subnetworks, $o$ denotes the overlap size and $m$ denotes the size of a non-overlap partition throughout the paper. These parameters are connected by the relation $m = (n-o)/s$.

\subsection{\code{NETCROP} for Cross-Validation}\label{sec:method:gen_algo}
\code{NETCROP} (outlined in Algorithm \ref{algo:gen}) aims to select the best model that fits the network or the best tuning parameter for model fitting from a set of candidate models or parameters.
In this section, we describe \code{NETCROP} in the context of model selection as the same ideas can easily be modified for a parameter tuning problem.
With a set of $\kappa$ candidate models $\M = \{\M_1, \ldots, \M_{\kappa}\}$ and a loss function $\ell$ as inputs, \code{NETCROP} has four major steps - \textsl{division}, \textsl{model fitting}, \textsl{stitching}, and \textsl{loss computation}. 
The first three steps constitute the model training phase of cross-validation, while the fourth step corresponds to the model testing phase.
The division step divides the network into $s$ subnetworks that share a set of overlap nodes. 
The model fitting step fits each candidate model on each subnetwork. The stitching step combines the model parameter estimates from the subnetworks using the overlap part. Finally, for each candidate model, the stitched model parameter estimates are used to predict the edge probabilities of the node pairs in the test set, and the prediction error is computed as the sum of losses between the predicted edge probability and the observed adjacency matrix. The steps are described in the following paragraphs.

\begin{figure}[!tb]
    \centering


    \includegraphics[width=9cm,keepaspectratio]{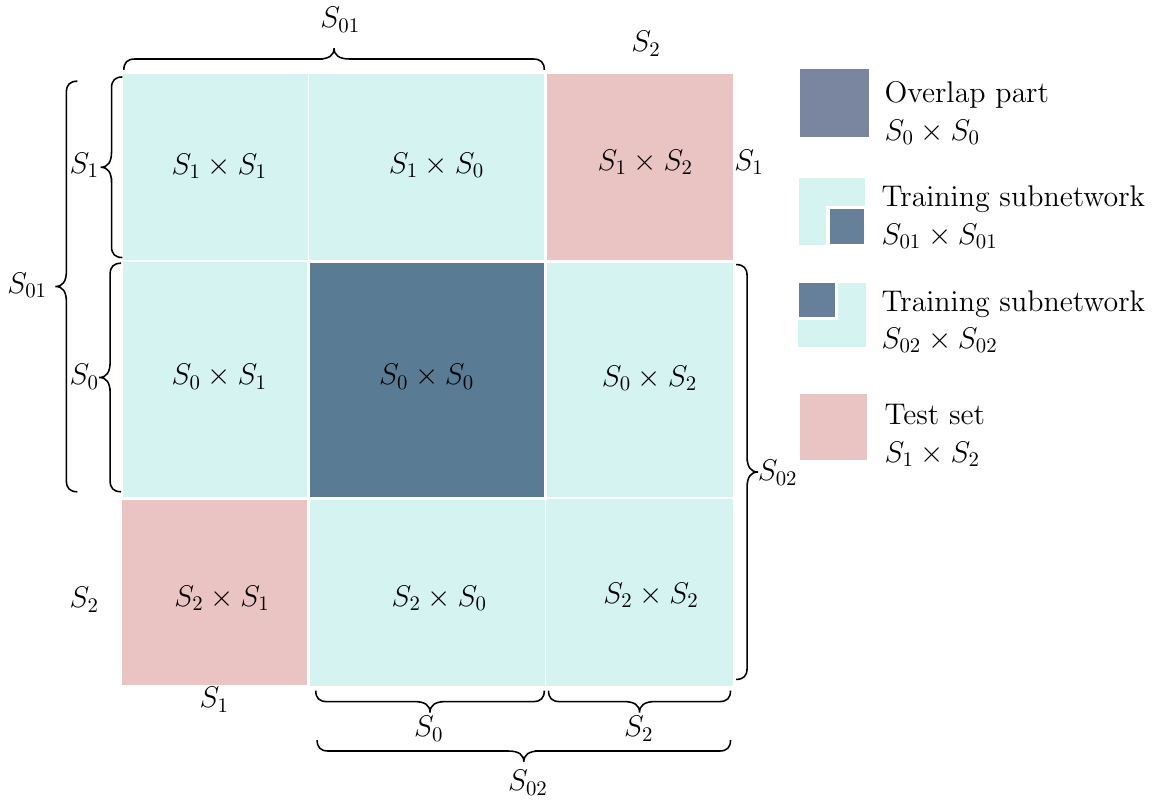} 

    \caption{\normalsize{Training subnetworks ($S_{01}$ and $S_{02}$) and testing set ($\S^c = S_1 \times S_2$) of \code{NETCROP} with overlap part $S_0$ and $s = 2$ non-overlap parts $S_1$ and $S_2$ ($s = 2$ subnetworks are used for illustration, while \code{NETCROP} can be used with $s \geq 2$).}}
    \label{fig:sonnet_descr}

    \par\smallskip
    {\raggedright
    \noindent\textbf{Alt text:} 
    Adjacency matrix with the rows and columns ordered by node subsets S one, S zero and S two for NETCROP with s equal 2. Training subnetworks occupy the top left and the bottom right blocks and share the central S zero by S zero block. The top right S one by S two block is the test set.
    \par}  
\end{figure}

In the division step, the overlap part $S_0$ is formed by randomly subsampling $o$ nodes from the set of nodes $S$. Then, the remaining nodes in $S \setminus S_0$ are partitioned into $s$ equally sized non-overlap parts $S_1, \ldots, S_s$. The subsets of nodes are formed as $S_{01} = S_0 \cup S_1, \ldots, S_{0s} = S_0 \cup S_s$. Note that the subnetworks are of size $(o+m)$, where $m = (n-o)/s$ is the size of each non-overlap part. \code{NETCROP} uses the subnetworks given by the sub-adjacency matrices $A_{S_{01}}, \ldots, A_{S_{0s}}$ as the training sets and the node pairs between the non-overlap parts as the test set. They are defined as
\begin{align*}
    \text{Train:} \ \S \coloneqq \bigcup\limits_{q =1}^s (S_{0q} \times S_{0q}),\ \quad \text{Test:} \ \S^c = (S \times S) \setminus \S = \mathop{\bigcup}\limits_{1 \leq p < q \leq s} (S_p \times S_q).
\end{align*}
%
%
The training and test set splitting of \code{NETCROP} with $s=2$ subnetworks is presented in Figure \ref{fig:sonnet_descr}.
The entire square in Figure \ref{fig:sonnet_descr} represents the $n \times n$ network adjacency matrix, where the nodes are rearranged so that the overlapping nodes in $S_0$ are in the middle of the rows and the columns, and the nodes in the non-overlap partitions $S_1$ and $S_2$ are on either side of $S_0$. 
Then, the two training subnetworks are in the two big squared blocks, indicated by $S_{01} \times S_{01} = (S_{0} \cup S_{1}) \times (S_{0} \cup S_{1})$ (the four blocks in the upper-left corner) and 
$S_{02} \times S_{02} = (S_{0} \cup S_{2}) \times (S_{0} \cup S_{2})$ (the four blocks in the lower-right corner). 
Note that the two blocks formed by the training subnetworks overlap on the central squared block indicated by $S_0 \times S_0$. 
The upper-right block outside of the training subnetworks in the top right corner of the adjacency matrix contains the test set $S_1 \times S_2$ (for directed networks, $S_2 \times S_1$ in the lower-left corner is also included in the test set).

In the model fitting step, each candidate model is fitted on each subnetwork $A_{S_{01}}, \ldots, A_{S_{0s}}$ to obtain their corresponding parameter estimates. 
The estimates of the unidentifiable parameters are matched using the overlap nodes by a suitable method. This eliminates any issues due to the unidentifiability of the parameters and ensures that the parameter estimates from the subnetworks are at par.
The matching method depends on the model structure. For example, permutation based label matching is needed to match the estimated community labels in a blockmodel. After the parameter matching step, the parameter estimates from each subnetwork are combined for each candidate model. The matching and the combining of the parameter estimates are called stitching. Let the stitched model estimates be denoted by $\hat{\M}_1, \ldots, \hat{\M}_{\kappa}$.

In the final loss computation step, the estimated probabilities for each node pair in the test set are computed for each candidate model. 
For $(i,j) \in \S^c$, let $\hat{P}_{ij}^{(k)}$ be the predicted edge probability between nodes $i$ and $j$ from the stitched $k$-th candidate model $\hat{\M}_k$, $k \in [\kappa]$. Let $\hat{P}^{(k)}_{\S^c} = \left( \hat{P}_{ij}^{(k)} \right)_{(i,j) \in \S^c}$. Then the computed losses are obtained as the sum of losses between the entries of the observed adjacency matrix and the predicted edge probabilities from each candidate model over the node pairs in the test set as
\begin{equation}
    \hat{L}^{(k)} = L\left( A_{\S^c}, \hat{P}^{(k)}_{\S^c} \right) \coloneqq \sum\limits_{(i,j) \in \S^c} \l\left( A_{ij}, \hat{P}^{(k)}_{ij} \right) = \sum\limits_{1 \leq p < q \leq s} \ \sum\limits_{(i,j) \in S_p \times S_q} \l\left( A_{ij}, \hat{P}^{(k)}_{ij} \right),\  k \in [\kappa].
\end{equation}
The outcome model $\hat{\M}$ is the model that has the smallest computed loss on the test set, i.e., 
$\hat{k} = \argmin_{k \in \kappa} \hat{L}^{(k)}$ and $\hat{\M} = \M_{\hat{k}}$.

\code{NETCROP} can be repeated $R$ times with different randomly subsampled overlap nodes and partitions.
Majority voting on the $R$ outcome models from the repetitions produce the final output. 
This optional step is called the \textsl{repetition step}.
Repetitions are expected to increase the stability of the model selection procedure by eliminating random noise as new random splits of the data are used for each repetition. A similar use of repetition can be seen in stability selection procedure by \citet{stability} and \citet{Soloffbagging2024}. 
Algorithm \ref{algo:gen} outlines \code{NETCROP} for general model selection. 

\begin{algorithm}[!tb]
    \caption{\code{NETCROP} for Model Selection}
    \label{algo:gen}
    \begin{algorithmic}
      \State \textbf{Input} An $n \times n$ network adjacency matrix $A$ with nodes indexed by $S = [n]$, a set of candidate models $\M = \{\M_1, \ldots, \M_\kappa\}$, number of subnetworks $s$, overlap size $o$ (non-overlap size $m = (n-o)/s$), number of repetitions $R$, and a loss function $\l: \R^2 \to [0, \infty)$.
    \Procedure{\code{NETCROP}}{$A, \M, s, o, R, \l$}
    \State \textbf{1.} \textbf{for}{ ($r$ from $1$ to $R$)}\ \{ 
        \State \textbf{1.1.} Randomly select $o$ nodes from $S$ to form the overlap part $S_0$.
        \State \textbf{1.2.} Split $S\setminus S_0$ randomly into $s$ exhaustive parts of equal size $m = (n-o)/s$ as $S_{1}, \ldots, S_{s}$.
        \State \textbf{1.3.} Form the subset of nodes $S_{0q} = S_0 \cup S_q$, $q \in [s]$.
        \State \textbf{1.4.} For each $k \in [\kappa]$ and $q \in [s]$, fit $\M_k$ on $A_{S_{0q}}$ to obtain the model estimates $\hat{\M}_k^{(q)}$.
        \State \textbf{1.5.} For each $k \in [\kappa]$, match the model estimates $\hat{\M}_k^{(2)}, \ldots, \hat{\M}_k^{(s)}$ with $\hat{\M}_k^{(1)}$ using the overlap nodes and update the model estimates.
        \State \textbf{1.6.} Combine the matched model estimates $\hat{\M}_k^{(1)}, \ldots, \hat{\M}_k^{(s)}$ to obtain the final estimate $\hat{\M}_k$ for each $k \in [\kappa]$.
        \State \textbf{1.7.} For each $k \in [\kappa]$, use the stitched model estimate $\hat{\M}_{k}$ to predict the edge probabilities of node pairs in the test set as $\hat{P}_{ij}^{(k)}$, for each $i \in S_p$ and $j \in S_q$, $1 \leq p < q \leq s$.
        \State \textbf{1.8.} Compute
        $\hat{L}^{(k)}_r = L\left(A_{\S^c}, \hat{P}^{(k)}_{\S^c}\right) = \sum\limits_{1 \leq p < q \leq s} \ \sum\limits_{(i,j) \in S_p \times S_q} \l\left( A_{ij}, \hat{P}^{(k)}_{ij} \right)$ for each $k \in [\kappa]$.\ \}
    \State \textbf{2.} Compute $\hat{k}_r = \argmin\limits_{k \in [\kappa]} \hat{L}^{(k)}_r$, for $r \in [R]$. Select $\hat{k}$ as the mode of $\hat{k}_1, \ldots, \hat{k}_R$.
    \State \textbf{3.} Return $\M_{\hat{k}}$ as the outcome model.
    \EndProcedure
    \end{algorithmic}
\end{algorithm}

\subsection{Remarks on the Design of \code{NETCROP}}
Each training subnetwork in \code{NETCROP}, being a subsample from the entire network, can be considered as an observed network of size $o + m = o + (n-o)/s$ from the true model. Thus, any consistency results as $n \to \infty$, which are applicable on the model estimates from the entire network, should be valid for the estimates from each training subnetwork as long as $o+m \to \infty$. Fixing the number of subnetworks $s$ with respect to $n$ and assuming the overlap size $o \to \infty$, we get $ m = s^{-1} (n - o) \asymp n$ and $o+m \asymp n$. 
These choices of $o$ and $m$ lead to a balance between the training and test subsets, which are both of size $\asymp n^2$.
This ensures that the parameter estimates for the true model, obtained from the training subnetworks, are consistent, which in turn lowers the computed loss for the true model. 
Additionally, a large enough test set is usually a better representative of the network and thus helps in distinguishing the wrong models from the correct one.
In Section \ref{sec:method:partune}, we provide a selection strategy for the parameters $o$ and $s$, which helps achieve a balance between the training and test splits that avoids overfitting and thereby preserves cross-validation accuracy.

\code{NETCROP} requires a choice of the loss function $\l$. 
From our observation, squared error loss performs well with \code{NETCROP}, producing reliable and accurate results in most situations. It is also very fast and easy to compute. We derived all the theoretical results using this loss function due to its algebraic tractability.
For unweighted networks with binary adjacency matrices, binomial deviance loss or negative of the area under the receiver operating characteristic curve (AUROC or AUC) can also be used with \code{NETCROP}. We noticed that, numerically, the binomial deviance loss produces similar results as the squared error loss, albeit it is a bit slower than the squared error loss. Computation of AUC is very slow for large networks and may offset the computational benefits of \code{NETCROP}. 

\subsection{Parameter Selection}\label{sec:method:partune}
To implement \code{NETCROP}, a user needs to choose its parameters --- the number of subgraphs $s$ and the overlap size $o$.
A simple way to choose the parameters is by setting the test set proportion to a user-specified fraction $p_{test}$.

Let $p_o = o/n$ be the overlap proportion. Then the proportion of node pairs in the test set is given by
$s(s-1)m^2 / n^2 = s(s-1) {(n-o)^2}/{s^2 n^2} = \left(1 - s^{-1}\right) (1 - p_o)^2.$
Setting this proportion as $p_{test}$, we get
\begin{align}
    s = \frac{(1 - p_o)^2}{(1-p_o)^2 - p_{test}}. \label{eq:param-select-s}    
\end{align}

The solution in \eqref{eq:param-select-s} is valid only when $p_o < 1 - \sqrt{p_{test}}$. Additionally, the solution must also satisfy the condition $s \geq 2$, which reduces to $p_o \geq 1 - \sqrt{2 p_{test}}$. 
This parameter selection strategy provides an interval $[1 - \sqrt{2p_{test}},\ 1 - \sqrt{p_{test}})$ for the overlap proportion $p_o$. 
We recommend using $p_o = 1 - \sqrt{(2 - \epsilon_o) p_{test}}$ with $\epsilon_o \in [0, 1)$.
As $p_o$ approaches the upper bound $1 - \sqrt{p_{test}}$, the denominator of \eqref{eq:param-select-s} approaches $0$, creating a sharp asymptote in the number of subnetworks $s$.
From our observations, for a given $p_{test}$, the value of $\epsilon_o$ does not affect the performance of \code{NETCROP} much as long as $\epsilon_o \leq 0.95$.
Empirically, choosing $p_{test} \leq 0.1$ and $\epsilon_o = 0$, with $s$ as the smallest integer greater than the right-hand side of \eqref{eq:param-select-s}, produces good choices of \code{NETCROP} parameters that keep the computational burden small for most cross-validation problems while maintaining a balance of large enough training sets that are not too large to cause overfitting.
In this paper, we used $p_{test} = 0.02$ for larger networks with $n \geq 1000$ and $p_{test} = 0.1$ for smaller networks.
For a given $p_{test}$, we choose the smallest integer $o$ greater than $n(1-\sqrt{2p_{test}})$ such that $n-o$ is divisible by the corresponding $s$.

\subsection{Computation Complexity}\label{sec:method:gen_comp}
The computation complexity of \code{NETCROP} for model selection depends on the sparsity of the network and the complexities of the model fitting, stitching and the loss computation steps. For a probability matrix $P = P_n = \rho_n P^0_n$ for some sparsity parameter $\rho_n$ and an $n \times n$ matrix $P^0_n$ with positive entries, the complexities of most model fitting procedures are usually polynomial in terms of $n, \log{n}$ and $\rho_n$. For the sake of simplicity, we assume that the complexity of fitting the candidate model with the slowest estimation procedure is $O(n^\theta)$ for a network of size $n$ for some $\theta > 0$.
Then the complexity of fitting $\kappa$ candidate models on $s$ training subnetworks, each of size $(o+m)$, is $O\left(\kappa s (o+m)^\theta\right)$, where $m = (n-o)/s$ is the size of each non-overlap part.

In the stitching step, the model parameters are matched only using the overlap nodes. Thus, complexity of parameter matching will depend on the overlap size $o$, the number of matchings $(s-1)$, the size $\kappa$ of the candidate set and the dimension of the parameter space. If the complexity of parameter matching based on the overlap nodes for the candidate models is $O(o^\beta)$ for some $\beta \geq 0$, then the overall complexity of the stitching step is $O\left(\kappa (s-1) o^\beta\right)$.

The loss computation step requires as many computations of the loss function as the size of the test set. The test set consists of the node pairs $(i,j) \in (S_p \times S_q)$ for $1 \leq p < q \leq s$. There are $s(s-1)/2$ pairs of non-overlap parts, where each pair of non-overlap parts contains $m^2$ node pairs. Thus, the size of the test set is $\lvert \S^c \rvert = s(s-1)m^2/2$. The complexity of this step is $O(\kappa s(s-1)m^2)$.

The complexity of \code{NETCROP} is proportional to the number of repetitions $R$.
Combining the complexities of the individual steps, the complexity of \code{NETCROP} with overlap size $o$, $s$ subnetworks and $R$ repetitions is $O\left(R \kappa s(o+m)^\theta \right) + O\left(R \kappa (s-1)o^\beta \right) + O\left(R \kappa s(s-1) m^2 \right)$.

In most cases, the model fitting step dominates the parameter matching step, i.e., $\theta > \beta$. The model fitting, stitching and loss computation steps can be parallelized on a multi-processor computer. Barring some overhead due to parallelization, complexity of \code{NETCROP} is approximately inversely proportional to the number of processors. If \code{NETCROP} is applied in parallel over $\zeta$ processors, the complexity of \code{NETCROP} is $O\left( \zeta^{-1} R\kappa \left(s(o+m)^\theta + s^2m^2\right)  \right) \approx O\left( \zeta^{-1} R\kappa s^2 (o+m)^\theta \right)$, where the last approximation is based on the fact that $\theta \geq 2$ for most network model fitting tasks.

With a similar setting, the complexity of \code{NCV} with $V$ folds and $R$ repetitions is $O\left(\zeta^{-1} R \kappa V n^\theta  \right)$ as the most computationally dominating step involves $V$ community detections on rectangular submatrices of order $n (1 - 1/V) \times n$, which is usually of order $O(n^\theta)$, depending on the algorithm used for the singular value decomposition within community detection. 
The complexity of \code{ECV} with $V$ cross validation folds under the same setting is $O\left(\zeta^{-1} V(n^{\theta_0} + \kappa n^\theta)\right)$, where $n^{\theta_0}$ is the complexity of matrix completion using singular value decomposition (SVD) on a matrix of size $n$. Depending on the SVD algorithm, usually $\theta_0 \in [2, 3]$. 

\subsection{Memory Complexity}\label{sec:method:memory}
\code{NETCROP} is very memory-efficient as it only uses a part of the network at a time, thus needing to load only a small part of the network to the volatile memory.
Deriving the exact memory usage of the computations is beyond the scope of this paper as it depends on many variables, including the model selection problem, the fitting algorithms, the programming language, the operating system and the type of computer hardware. Here, we provide a comparative analysis of storage memory use of \code{NETCROP} and the competing algorithms \code{NCV} and \code{ECV}.

The amount of memory required to store a network of size $n$ and sparsity (or edge density) regulated by $\rho_n$ is on the order of $n^2 \rho_n$. Since the training step of \code{NETCROP} splits the nodes into $s$ overlapping partitions, each of size $o+m$, per repetition, and then accesses the subnetworks corresponding to those partitions to fit each candidate model, the total storage cost of the training step of \code{NETCROP} boils down to $R s (o+m)^2 \rho_n$. However, the processors only need to access one training block at a time for the computations, thus reducing the storage memory required for the training step by one processor to $(o+m)^2 \rho_n$. 
The testing step of \code{NETCROP} performs link predictions and loss computations on $s(s-1)$ submatrices each of size $m^2$. The total storage for the test step is $s(s-1)m^2$, while the storage requirement per processor is just $m^2$.

On the other hand, the training step of \code{NCV} fits the candidate models on $V$ submatrices, each of size $n (1 - 1/V) \times n$ (with sparsity $\rho_n$), while testing on $V$ blocks, each of size $n/V \times n/V$. Hence, the total and per-processor storage memory requirements of \code{NCV} are $V n^2 (1 - 1/V) \rho_n$ and $n^2(1 - 1/V)\rho_n$ for training, and $n^2/V$ and $n^2/V^2$ for testing, respectively. \code{ECV} has a higher memory requirement than both \code{NETCROP} and \code{NCV}. \code{ECV} with $V$ folds and training proportion $(1-p)$ needs to store $Vn^2(1-p)\rho_n$ node pairs for matrix completion, $Vn^2 \rho_n$ node pairs for fitting the models and $Vn^2(1-p)$ node pairs for testing. The per-processor storage memory costs are then $n^2(1-p)\rho_n$, $n^2 \rho_n$ and $n^2(1-p)$ for matrix completion, fitting and testing, respectively.

In Section \ref{app:numerical:small} of the Appendix, we have included a numerical example comparing the memory usage of \code{NETCROP} with \code{NCV} and \code{ECV}. We observed that \code{NETCROP} uses significantly less memory than both \code{NCV} and \code{ECV} and the difference in memory use grows with the network size.

\section{{\code{NETCROP} for Model Selection and Parameter Tuning}}\label{sec:example}
In this section, we present the details for applying \code{NETCROP} to four cross-validation tasks:
(1) detecting the number of communities and degree heterogeneity in blockmodels, 
(2) estimating the latent space dimension in RDPG models,
(3) estimating the latent space dimension in latent space models, and 
(4) tuning the regularizing parameter for regularized spectral clustering.

\subsection{Model Selection for Blockmodels}\label{sec:example:SBM_DCBM_algo} 
Here, we discuss how \code{NETCROP} (Algorithm \ref{algo:gen}) is used for selecting the number of communities and choosing the correct model simultaneously from a candidate set consisting of SBM and DCBM.
An SBM with $n$ nodes and $K$ communities is parameterized by $(g, K, B)$, where $g \in [K]^n$ and $B \in [0,1]^{K \times K}$.
An edge between nodes $i$ and $j$ is given by $A_{ij} \sim \operatorname{Bernoulli}(P_{ij} = B_{g_i g_j})$ independently for all $i < j$, given $g$. Define $G_k = \{i \in [n] : g_i = k \}$, $n_k = \lvert G_k \rvert$ for $k \in [K]$ and $C \in [0,1]^{n \times K}$ such that $C_{ig_i} = 1$ and $0$ otherwise. Then $P = CBC^\top$. We denote this model as $A \sim SBM(n, K, B)$.
A DCBM adds a layer of complexity to SBM with a degree heterogeneity parameter $\psi_i > 0$ for the $i$th node. 
Define $\Psi \in \R^{n \times K}$ satisfying $\Psi_{i g_i} = \psi_i $ and $0$ otherwise. Given the parameters $(g, K, B, \Psi)$, $P = \Psi B \Psi^\top$ and $P_{ij} = B_{g_i g_j} \psi_i \psi_j$. We denote this model as $A \sim DCBM(n, K, B, \Psi)$.

Let $\M =\{(\K, \delta) \in [K_{max}] \times \{0, 1\} \}$ be the set of candidate models, where $K_{max}$ is the largest number of communities considered and $\delta$ is the degree correction parameter. Here $\delta = 0$ means SBM and $\delta = 1$ DCBM. To adapt \code{NETCROP} to select the best $(\hat{K}, \hat{\delta}) \in [K_{max}] \times \{0, 1\}$, the stitching and the test set prediction step need to be adjusted. The method is illustrated using spectral clustering (\code{SC}, Algorithm \ref{algo:spec_clus} in Appendix Section \ref{app:adtl_algo}) for SBM and spherical spectral clustering (\code{SSC}, Algorithm \ref{algo:spher_spec_clus} in Appendix Section \ref{app:adtl_algo}) for DCBM, but any suitable and reasonably accurate community detection algorithms for SBM and DCBM can be used with \code{NETCROP}.
We note that this approach was introduced by \citet{sonnet_comm} for scalable community detection.

Given the subnetworks $S_{01}, \ldots, S_{0s}$ from the division step, both \code{SC} and \code{SSC} are applied on each subnetwork with $\K$ communities for each $\K \in [K_{max}]$. Let the $(o+m) \times \K$ output community membership matrices corresponding to the $q$-th subnetwork be $\hat{C}^{(0 \K)}_{q}$ and $\hat{C}^{(1\K)}_{q}$ from \code{SC} and \code{SSC}, respectively, with $\K$ communities, $q \in [s]$.

Since the community labels are identifiable only up to permutation, for each $\K \in [K_{max}]$, the columns of the $s$ estimated community membership matrices need to be permuted to match with a standard set of labels. 
We describe two suitable label matching algorithms, \code{MatchBF} and \code{MatchGreedy}, in Appendix Section \ref{app:adtl_algo}.
The overlap nodes are used for the label matching in the stitching step. There are $s$ sets of estimated community labels with shared overlap nodes. We set the estimated labels for the overlap part from the first subnetwork $\hat{C}^{(\delta \K)}_{1, S_0\cdot}$ as the standard labels, and match the labels from $\hat{C}^{(\delta \K)}_{2, S_0\cdot}, \ldots, \hat{C}^{(\delta \K)}_{s, S_0\cdot}$ to obtain the best matching permutation matrices $H^{(\delta \K)}_2, \ldots, H^{(\delta \K)}_s \in \H_{\K}$ for $\K \in [K_{max}]$ and $\delta \in \{0, 1\}$. The choice of standard labels is arbitrary and any one of the estimated labels can be used as the standard for label matching. Then all the labels in each of the subnetworks are updated using their corresponding permutation matrices as $\hat{C}^{(\delta \K)}_{q} \leftarrow \hat{C}^{(\delta \K)}_{q} H^{(\delta \K)}_q$, $q \in [2:s]$, $\K \in [K_{max}]$, $\delta \in \{0, 1\}$.

For the loss computation step, the model parameters $(g, B)$ for SBM and $(g, B, \Psi)$ for DCBM need to be estimated. The estimated vectors of community labels from the subnetworks are directly obtained from the estimated community membership matrices as $\hat{g}^{(\delta \K)}_{q} = \hat{C}^{(\delta \K)}_{q} \1_{\K}$. Define $\hat{G}^{(\delta \K)}_{q, k} \coloneqq \{i \in S_{0q} : \hat{g}^{(\delta \K)}_{q, i} = k\}$ as the set of nodes in the $q$-th subnetwork that are predicted to be in the $k$-th community for $q \in [s]$, $k \in [\K]$, $\K \in [K_{max}]$, $\delta \in \{0, 1\}$.

For SBM, the plug-in estimator of the community connectivity matrix $B$ for the $q$-th subnetwork and $\K$ communities is computed as
\begin{align}
    \hat{B}^{(\K)}_{q, k k^\prime} = \frac{\sum\limits_{(i,j) \in (S_{0q} \times S_{0q})} A_{ij} \I\left( (i,j) \in \hat{\eta}^{(0\K)}_{q, k k^\prime} \right)}{\left\lvert \hat{\eta}^{(0\K)}_{q, k k^\prime} \right\rvert}, \label{eqn:sbm.b.hat}
\end{align}
where
\begin{align}
    \hat{\eta}^{(\delta\K)}_{q, k k^\prime} = 
        \left\{
	      \begin{array}{ll}
	   	   \left\{ (i, j) \in (S_{0q} \times S_{0q}) : \hat{g}^{(\delta \K)}_{q, i} = k,\ \hat{g}^{(\delta \K)}_{q, j} = k^\prime \right\},  & \ k \neq k^\prime\\
       \\
		      \left\{ (i, j) \in (S_{0q} \times S_{0q}) : i < j, \ \hat{g}^{(\delta \K)}_{q, i} = \hat{g}^{(\delta \K)}_{q, j} = k \right\}, & \ k = k^\prime
            \end{array}
    \right.,\ \delta \in \{0, 1\},
\end{align}
for $q \in [s],\ k, k^\prime \in [\K],\ \K \in [K_{max}]$. Now the estimates of $B$ in (\ref{eqn:sbm.b.hat}) are averaged over the subnetworks to obtain the final estimates of community connectivity matrix and the test set edge probabilities as
\begin{align}
    \hat{B}^{(\K)} = \frac{1}{s} \sum\limits_{q=1}^s \hat{B}^{(\K)}_{q}\ \text{and}\ 
    \hat{P}^{(0\K)}_{ij} = \hat{B}^{(\K)}_{\hat{g}^{(0\K)}_{p, i} \hat{g}^{(0\K)}_{q, j}}\label{eqn:sbm.est}
\end{align}
for $(i,j) \in S_p \times S_q$, $1 \leq p < q \leq s$, $\K \in [K_{max}]$.

DCBM has an additional degree heterogeneity parameter $\psi_i$ corresponding to each node $i$ that needs to be estimated along with the community connectivity matrix. The DCBM parameters $B$ and $\psi$ are only identifiable up to some factor. \citet{cv_comm} used the top $K$ right singular vectors to estimate a community-normalized version of the DCBM parameters as
\begin{align}
    \psi^\prime_i = \frac{\psi_i}{\sqrt{\sum\limits_{j : g_j = g_i} \psi_j^2}},\ \quad \mbox{and}\ \quad B^\prime_{k k^\prime} = B_{k k ^\prime} \sqrt{\sum\limits_{i : g_i = k} \psi_i^2 \sum\limits_{j : g_j = k^\prime} \psi_j^2}.
\end{align}
Note that $P_{ij} = B_{g_i g_j} \psi_i \psi_j = B^\prime_{g_i g_j} \psi^\prime_i \psi^\prime_j$. Let $\hat{U}^{(\K)}_{q}$ be the $(o+m) \times \K$ matrix of eigenvectors of $A_{S_{0q}}$. Then $\psi^\prime_{i}$ corresponding to the $q$-th subnetwork is estimated as the $\ell_2$ norm of the $i$th row of $\hat{U}^{(\K)}_{q}$, and it is used to compute the corresponding plug-in estimate of $B^\prime$. The estimates are 
\begin{align}
    \hat{\psi}^{\prime (\K)}_i = \lVert \hat{U}^{(\K)}_{q, i\cdot} \rVert,\ \quad \mbox{and}\ \quad \hat{B}^{\prime(\K)}_{q, k k^\prime} = \frac{\sum\limits_{(i,j) \in (S_{0q} \times S_{0q})} A_{ij} \I\left( (i,j) \in \hat{\eta}^{(1\K)}_{q, k k^\prime} \right) }{ \sum\limits_{(i,j) \in (S_{0q} \times S_{0q})} \hat{\psi}^{\prime (\K)}_i \hat{\psi}^{\prime (\K)}_j \I\left( (i,j) \in \hat{\eta}^{(1\K)}_{q, k k^\prime} \right) }
\end{align}
for $i \in S_{0q}$, $q \in [s],\ k, k^\prime \in [\K],\ \K \in [K_{max}]$. Finally, the mean of the estimated scaled community connectivity matrices is used along with the estimated scaled degree heterogeneity parameters to predict the edge probabilities of nodes in the test set as
\begin{align}
    \hat{B}^{\prime(\K)} = \frac{1}{s} \sum\limits_{q=1}^{s} \hat{B}^{\prime(\K)}_q,\ \quad \mbox{and}\ \quad \hat{P}^{\prime(1\K)}_{ij} = \hat{B}^{\prime(\K)}_{\hat{g}^{(1\K)}_{p, i} \hat{g}^{(1\K)}_{q, j}} \hat{\psi}^\prime_i \hat{\psi}^\prime_j \label{eqn:dc.est.eigen}
\end{align}
for $(i,j) \in S_p \times S_q$, $1 \leq p < q \leq s$, $\K \in [K_{max}]$.
An alternate estimator of the probability matrix under DCBM that uses Poisson approximation \citep{KarrerDCBM} is given as
\begin{align}
    &\hat{O}_{q, k k^\prime}^{(\K)} = \sum\limits_{(i,j) \in (S_{0q} \times S_{0q})} A_{ij} \I\left( (i,j) \in \hat{\eta}^{(1\K)}_{q, k k^\prime} \right), \ \quad \hat{\psi}_i = \left(\sum\limits_{j \in S_{0q}} A_{ij}\right) \bigg / \left( \sum\limits_{k=1}^{\K} \hat{O}_{q, \hat{g}^{(1\K)}_{q,i} k}^{(\K)} \right),\nonumber\\
    &\hat{O}^{(\K)} = \frac{1}{s} \hat{O}_{q}^{(\K)}, \ \quad \mbox{and}\ \quad \hat{P}_{ij}^{(1\K)} = \hat{O}^{(\K)}_{\hat{g}_{1,i} \hat{g}_{1,j}} \hat{\psi}_i^{\K} \hat{\psi}_j^{\K} \label{eqn:dc.est.pois}
\end{align}
for $(i,j) \in S_p \times S_q$, $1 \leq p < q \leq s$, $\K \in [K_{max}]$.

Empirically, the performance of \code{NETCROP} was observed to not depend on the choice of the DCBM parameter estimator.
Both estimators in (\ref{eqn:dc.est.eigen}) and (\ref{eqn:dc.est.pois}) provide sufficiently accurate predicted test set edge probabilities. The theoretical derivations for DCBM are done with the \citet{cv_comm} estimators in (\ref{eqn:dc.est.eigen}) as the properties of the network eigenvalues can be exploited to connect the estimators with the scaled versions of the true parameters. However, unless otherwise specified, we used the \citet{KarrerDCBM} estimate in (\ref{eqn:dc.est.pois}) for all the numerical examples as it is easier and more intuitive to implement.

With the predicted edge probabilities for SBM and DCBM, the number of communities $\hat{K}$ and the degree correction indicator $\hat{\delta}$ are selected based on the computed losses. In case of multiple repetitions, the best outcome is chosen by majority voting on the outcomes from the repetitions. Algorithm \ref{algo:SBM_DCBM} outlines \code{NETCROP} for selecting $K$ and degree correction indicator in a blockmodel.

\begin{algorithm}[!tb]
    \caption{\code{NETCROP} for Determining the Number of Communities and Choosing between SBM and DCBM} 
    \label{algo:SBM_DCBM}
    \begin{algorithmic}
      \State \textbf{Input} An $n \times n$ network adjacency matrix $A$ with node set $S$, maximum number of communities $K_{max}$ to consider, number of subnetworks $s$, overlap size $o$ (non-overlap size $m = (n-o)/s$), number of repetitions $R$, and a loss function $\l$.
    \Procedure{\code{NETCROP}}{$A, K_{max}, s, o, R, \l$}
    \State \textbf{1.} \textbf{for}{ ($r$ from $1$ to $R$)}\ \{
        \State \textbf{1.1.} Randomly select $o$ nodes from $S$ to form the overlap part $S_0$.
        \State \textbf{1.2.} Split $S\setminus S_0$ randomly into $s$ exhaustive parts of equal size $m = (n-o)/s$ as $S_{1}, \ldots, S_{s}$.
        \State \textbf{1.3.} Form the subset of nodes $S_{0q} = S_0 \cup S_q$, $q \in [s]$.
        \State \textbf{1.4.} For each $q \in [s]$, compute $\hat{U}_q$ as the $(o+m) \times K_{max}$ matrix with the eigenvectors corresponding to the largest $K_{max}$ eigenvalues of $A_{S_{0q}}$. 
        \State \textbf{1.5.} \textbf{for}{ ($\K$ from $1$ to $K_{max}$)} \ \{
            \State \textbf{1.5.1.} For each $q \in [s]$, compute $\hat{C}^{(0 \K)}_{q}$ using $k$-means clustering on $\hat{U}_{q, [\K]\cdot}$ and $\hat{C}^{(1 \K)}_{q}$ using $k$-median spherical clustering on the row normalized version of $\hat{U}_{q, [\K]\cdot}$.
            \State \textbf{1.5.2.} Compute the best matching permutation matrices as $H^{(\delta \K)}_q = \argmin\limits_{H \in \H_{\K}}\left\lVert \hat{C}^{(\delta \K)}_{q, S_0\cdot} H -  \hat{C}^{(\delta \K)}_{1, S_0\cdot}\right\rVert_0$,\ $q \in [2:s],\ \delta \in \{0,1\}$ and update $\hat{C}^{(\delta \K)}_q \leftarrow \hat{C}^{(\delta \K)}_q H^{(\delta \K)}_q$, $q\in[2:s]$, $\delta\in \{0,1\}$.
            \State \textbf{1.5.3.} For each $i \in S_p$, $j \in S_q$, $1\leq p < q \leq s$, predict the edge probabilities for SBM as $\hat{P}_{ij}^{(0\K)}$ in (\ref{eqn:sbm.est}) and for DCBM as either $\hat{P}_{ij}^{\prime(1\K)}$ in (\ref{eqn:dc.est.eigen}) or $\hat{P}_{ij}^{(1\K)}$ in (\ref{eqn:dc.est.pois}).
            \State \textbf{1.5.4.} Compute $\hat{L}_{r}^{(\delta\K)} = L\left(A_{\S^c}, \hat{P}_{\S^c}^{(\delta\K)}\right) = \sum\limits_{1 \leq p < q \leq s} \ \mathop{\sum\limits_{i \in S_p}}\limits_{j \in S_q} \l\left( A_{ij}, \hat{P}^{(\delta\K)}_{ij} \right)$, $\delta \in \{0, 1\}$.\ \} \}
    \State \textbf{2.} Compute $(\hat{K}_r, \hat{\delta}_r) = \argmin\limits_{(\K, \delta) \in [K_{max}] \times \{0,1\}} \hat{L}_{r}^{(\delta\K)}$, $r \in [R]$.
    \State \textbf{3.} Return $(\hat{K}, \hat{\delta})$ as the most occurring outcome from $\{(\hat{K}_r, \hat{\delta}_r) : r \in [R]\}$.
    \EndProcedure
    \end{algorithmic}
\end{algorithm}

\subsection{Estimating Latent Space Dimension of RDPG}\label{sec:method:RDPG_algo} 
Here we describe \code{NETCROP} for selecting the latent space dimension of RDPG. An RDPG of size $n$ with latent space dimension $d$ is characterized by an $n \times d$ matrix $X$, where each row of $X$ is the latent position of the corresponding node in a $d$-dimensional Euclidean space. The probability matrix is given by $P = X X^\top$ and $A_{ij} \sim \operatorname{Bernoulli}(P_{ij} = X_{i\cdot}^\top X_{j \cdot} )$ independently for all $i < j$, where the rows $X_{i \cdot}$ of the latent position matrix $X$ are treated as $d \times 1$ vectors. We also use the notation $A \sim RDPG(n, d, X)$ or $A \sim RDPG(X)$ to denote this model.

A major inference task for an observed network from RDPG is recovering the latent positions and predicting the edge probabilities. Adjacency spectral embedding, or \code{ASE} \citep{RDPG_survey}, is a spectral method used to estimate the latent position matrix using the eigenvalues and the eigenvectors of the network adjacency matrix. Given a latent space dimension $d$, \code{ASE} uses the top $d$ eigenvectors corresponding to the largest $d$ eigenvalues of $A$ to estimate the latent position matrix $X$. \code{ASE} is described in Algorithm \ref{algo:ASE} in Appendix Section \ref{app:adtl_algo}.

\code{NETCROP} (Algorithm \ref{algo:gen}) can be modified for the problem of selecting the latent space dimension $d$ from a set of candidate dimensions $\{1, \ldots, d_{max}\}$. The division step splits the set of nodes $S$ into $s$ subnetworks $S_{01}, \ldots, S_{0s}$, where the subnetworks share a common overlap part $S_0$ of size $o$. 
In the first step, the latent position matrix for each subnetwork is estimated using \code{ASE} with dimension $d_{max}$ --- let it be denoted by $\hat{X}_{q}$ for the $q$-th subnetwork, $q \in [s]$. Next, the stitching step matches $\hat{X}_{q}$ with $\hat{X}_1$ for $q \in [2:s]$. Then the test set edge probabilities with candidate dimension $\d \ (\leq d_{max})$ are predicted using the first $\d$ columns of the relevant estimated latent position matrices. This is slightly different from the detection and the stitching steps used in \code{NETCROP} for SBM and DCBM (Algorithm \ref{algo:SBM_DCBM}), as they are done only once for $d_{max}$ instead of for each $\d \in [d_{max}]$.

Since the latent positions are identifiable only up to an orthogonal rotation in the $d$-dimensional Euclidean space, the estimated latent positions from the subnetworks may not be in the same rotation in $\R^d$. For example, for any $d$-dimensional orthogonal matrix $W \in \O_d$, $X$ and $\Tilde{X} = XW$ yield the same probability matrix $P$. To ensure that the estimated latent positions from all the subnetworks are rotated similarly, one of the subnetworks is randomly chosen as the standard and the estimated latent positions of the overlap nodes of the other subnetworks are rotated to match with the latent position of the overlap nodes of the standard subnetwork. Since the choice of the standard subnetwork is arbitrary, we treat the first subnetwork as the standard throughout this paper. We use the standard Procrustes transformation \citep{og-procrustes} to match the latent positions. For any two $n \times d$ matrices $M_1$ and $M_2$, the Procrustes transformation finds the best orthogonal matrix $W \in \O_d$ that minimizes $\lVert M_1 W - M_2 \rVert_F$. If the singular value decomposition of $M_2^\top M_1 = U \Sigma V^\top$, the minimizer is $W = U V^\top$.
Let $\hat{W}_q \in \O_{d_{max}}$ be the Procrustes transformation that matches $\hat{X}_{q, S_0 \cdot}$ with $\hat{X}_{1, S_0 \cdot}$ for $q \in [2:s]$. Then the estimated latent positions from the subnetworks are updated as $\hat{X}_{q} \leftarrow \hat{X}_{q} \hat{W}_{q}$, $q \in [2:s]$.

In the loss computation step, the edge probabilities are predicted for each candidate dimension $\d \in [d_{max}]$ by only using the first $\d$ columns of the estimated latent position matrices from the subnetworks. For node pairs between the $p$-th non-overlap part $S_p$ and the $q$-th non-overlap part $S_q$, the predicted latent positions and the edge probabilities for candidate dimension $\d$ are
\begin{align}
    \hat{X}_{q}^{(\d)} = \hat{X}_{q, \cdot [\d]}, \ q\in[s],\ \quad \mbox{and}\ \quad \hat{P}^{(\d)}_{S_p S_q} = \hat{X}_{p}^{(\d)} \left( \hat{X}^{(\d)}_q \right)^\top 
\end{align}
for $1 \leq p < q \leq s$ and $\d \in [d_{max}]$.
Since the test set consists of node pairs from all possible pairs of non-overlap parts, the final losses are computed as the sum of losses between $A_{S_p S_q}$ and $\hat{P}^{(\d)}_{S_p S_q}$ over $1 \leq p < q \leq s$ for each $\d \in [d_{max}]$. Finally the dimension $\hat{d}$ with the lowest computed loss is returned as the estimated latent space dimension. The procedure is repeated $R$ times for stability and the most frequent value of $\hat{d}$ is the final outcome. Algorithm \ref{algo:RDPG} outlines the steps.

\begin{algorithm}[!tb]
    \caption{\code{NETCROP} for Determining the Dimension for RDPGs} 
    \label{algo:RDPG}
    \begin{algorithmic}
      \State \textbf{Input} An $n \times n$ network adjacency matrix $A$ with node set $S$, maximum latent space dimension  $d_{max}$ to consider, number of subnetworks $s$, overlap size $o$ (non-overlap size $m = (n-o)/s$), number of repetitions $R$, and a loss function $\l$.
    \Procedure{\code{NETCROP}}{$A, d_{max}, s, o, R, \l$}
    \State \textbf{1.} \textbf{for}{ ($r$ from $1$ to $R$)}\ \{
        \State \textbf{1.1.} Randomly select $o$ nodes from $S$ to form the overlap part $S_0$.
        \State \textbf{1.2.} Split $S\setminus S_0$ randomly into $s$ exhaustive parts of equal size: $S_{1}, \ldots, S_{s}$.
        \State \textbf{1.3.} Form the subset of nodes $S_{0q} = S_0 \cup S_q$, $q \in [s]$.
        \State \textbf{1.4.} For each $q \in [s]$, compute $\hat{\Lambda}_{q} = diag\left(\hat\lambda_{q, 1},  \ldots, \hat\lambda_{q, d_{max}}\right)$ as the diagonal matrix containing the $d_{max}$ largest eigenvalues  $\hat\lambda_{q, 1}\geq \cdots \geq \hat\lambda_{q, d_{max}}$ and $\hat{U}_{q}$ as the $(o+m) \times d_{max}$ matrix of the corresponding eigenvectors of $A_{S_{0q}}$.
        \State \textbf{1.5.} Compute $\hat{X}_{q} = \hat{U}_{q} \hat{\Lambda}_{q}^{\frac{1}{2}}$, $q \in [s]$.
        \State \textbf{1.6.} For each $q \in [2:s]$, $\hat{W}_q = \argmin\limits_{W \in \O_{d_{max}}} \left\lVert \hat{X}_{q, S_0 \cdot}\hat{W} - \hat{X}_{1, S_0 \cdot} \right\rVert_F$ and update $\hat{X}_{q} \leftarrow \hat{X}_q W_q$.
        \State \textbf{1.7.} \textbf{for}{ ($\d$ from $1$ to $d_{max}$)}\ \{
            \State \textbf{1.7.1.} Set $\hat{X}_{q}^{(\d)} = \hat{X}_{q, \cdot [\d]}$ for $q \in [s]$.
            \State \textbf{1.7.2.} Estimate $\hat{P}^{(\d)}_{S_p S_q} = \hat{X}^{(\d)}_p \mathop{\hat{X}^{(\d)}_q}^\top$ for $1 \leq p < q \leq s$.
            \State \textbf{1.7.3.} Compute the loss function 
            $\hat L^{(\d)}_r = L\left(A_{\S^c}, \hat{P}^{(\d)}_{\S^c} \right) = \sum\limits_{1\leq p < q \leq s} \l\left( A_{S_p S_q}, \hat{P}^{(\d)}_{S_p S_q}\right)$.\ \}\}
    \State \textbf{2.} Compute $\hat{d}_r = \argmin\limits_{\d \in [d_{max}]} \hat L^{(\d)}_r$ for $r \in [R]$.
    \State \textbf{3.} Return $\hat d$ as the modal value of $\hat{d}_1, \ldots, \hat{d}_R$.
    \EndProcedure
    \end{algorithmic}
\end{algorithm}

\subsection{Estimating the Dimension of Latent Space Model}
The latent space model \citep{Hoff2002LatentSA} assumes that the edge probabilities are inversely related to the distance between the assumed latent positions of the nodes in a low dimensional ($d << n$) Euclidean space. Any inference on networks from this model, including estimating the latent positions, requires the knowledge of the latent space dimension $d$. \code{NETCROP} for RDPG (Algorithm \ref{algo:RDPG}) can be adapted for selecting the latent space dimension for a network from latent space model by replacing \code{ASE} for estimating latent positions in RDPG with the maximum likelihood estimator for the latent space model, and using generalized Procrustes transformation for the matching step as the latent positions are identifiable only up to translations and orthogonal rotations. The outcome dimension is chosen from the candidate set $[d_{max}]$ that has the lowest test set loss. The latent space model along with the specific changes to \code{NETCROP} are presented in Appendix Section \ref{app:adtl_algo}.

\subsection{Tuning the Regularization Parameter in Regularized Spectral Clustering}
Regularized spectral clustering or \code{RSC} \citep{qin_reg_spec} is commonly used for community detection in DCBMs. Classical spectral clustering (Algorithm \ref{algo:spec_clus}) performs poorly in the presence of degree heterogeneity and in sparse networks. \code{RSC} addresses this issue by inflating the node degrees by $\tau$ times the edge density of the network for a tuning parameter $\tau \geq 0$ while computing the graph Laplacian and applying spectral clustering with row normalization on the modified Laplacian. Cross-validation on networks can be used to perform a grid search on a range of candidate values of $\tau$ to select the best tuning parameter for \code{RSC}. \code{NETCROP} for selecting the tuning parameter is similar to Algorithm \ref{algo:SBM_DCBM}, except that \code{RSC} is applied on each subnetwork with each candidate value of $\tau$. The stitching or label matching step and the testing step remain the same as Algorithm \ref{algo:SBM_DCBM}. Finally, the value of $\tau$ that produces the lowest prediction error on the test set is chosen as the best tuning parameter. Without additional computations, the community labels can be predicted using the matched community labels from the subnetworks with the selected value of $\tau$ from \code{NETCROP}.

\section{Theoretical Results}\label{sec:theory}
In this section, we derive theoretical results for accuracy of \code{NETCROP} for model selection under some suitable assumptions. In Sections \ref{sec:theory:SBM_DCBM} and \ref{sec:theory:dcbm}, we show that the probability of underestimating the number of communities for SBM and DCBM tends to zero as the network size increases. 
Although \code{NETCROP} for blockmodels is designed to detect degree heterogeneity and the number of communities simultaneously and the numerical examples are for such cases, the theoretical results for SBM and DCBM are limited to detecting the number of communities in a candidate set of only SBMs and only DCBMs, respectively.
In Section \ref{sec:theory_RDPG}, we show the same for the latent space dimension of RDPG. The general idea for both cases is that the computed loss between the observed network (adjacency matrix) and the predicted edge probabilities on the test set is much larger for the wrong candidate models, but smaller for the true model. We compare the computed loss with the oracle loss to show this. The oracle loss is defined as the sum of losses between the adjacency matrix and the true edge probabilities over the test set, which can be written as
\begin{align*}
    L(A_{\S^c}, P_{\S^c}) = \sum\limits_{(i,j) \in \S^c} \l(A_{ij}, P_{ij}). 
\end{align*}
For the true model, the computed loss is shown to be concentrated around the oracle loss and for the wrong models, the computed loss is shown to be much larger than the oracle loss with high probabilities. Thus, the true model is chosen with high probability.

All theoretical derivations are for \code{NETCROP} with $s$ subnetworks, overlap size $o$, and $R=1$ repetition. Since the repetitions are independent, the same consistency results should hold with high probability using the union bound when $R>1$. We also assume the squared error loss, defined as $\ell_2(a,b) = (a-b)^2$ for all theoretical results.

As is known in the cross-validation literature, there is generally no theoretical guarantee against overestimation \citep{StoneCV1978, shaolinear1993, zhangmodel1993, yangconsistency2007}. Please see the paragraph after Corollary 1 of \citet{cv_comm} and Section 1.1 of \citet{Leicv2020} for more discussion on this. In line with the existing network cross-validation methods \citep{cv_comm, cv_edge}, we only prove one-sided consistency that ensures that the number of communities for SBM, DCBM and latent position dimension for RDPG are not underestimated with \code{NETCROP}.

Although \code{NETCROP} is not theoretically guaranteed against overestimation, we provide some intuition for why \code{NETCROP} is unlikely to overestimate the network models in practice.
For blockmodels, if the candidate number of communities is greater than the true number of communities, a few of the true communities are expected to split artificially and be distributed among multiple predicted communities. Node pairs from those predicted communities are expected to make the computed loss from \code{NETCROP} larger. Similarly for RDPG, using larger candidate dimension than the true dimension will incorporate extra terms in the coordinates of the latent positions that will push the predicted edge probabilities away from the true edge probabilities, eventually increasing the computed loss.

\subsection{Theory for SBM}\label{sec:theory:SBM_DCBM}
Here we consider the model selection problem of choosing the number of communities $\hat{K}$ from a candidate set of SBM with the number of communities in $[K_{max}]$. \code{NETCROP} for blockmodels (Algorithm \ref{algo:SBM_DCBM}) is applied with \code{SC} (Algorithm \ref{algo:spec_clus}), $s$ subnetworks, overlap size $o$ (thus, non-overlap size $m = (n-o)/s$), $R = 1$ repetition, and the squared error loss. We assume that the true model is SBM with $K$ communities.

Since the results for SBM are for candidate set consisting of only SBMs, we use $\hat{P}_{\S^c}^{(\K)}$ to mean $\hat{P}_{\S^c}^{(0\K)}$ from (\ref{eqn:sbm.est}).
The following assumptions are needed to establish the results for SBM.
\vspace{-\parskip}
\begin{assumption}{(Model assumptions for SBM)}\label{assump:SBM}
\begin{itemize}
    \item Let $K$ be fixed, and $g_i \in [K]$ be the true community membership of node $i \in [n]$. Then, for $1 \leq i \neq j \leq n$, $0 \leq P_{ij} = B_{g_i g_j} = \rho_n B_{0, g_i g_j} = \rho_n P^0_{n, ij} \leq 1$ for a fixed symmetric matrix $B_0 \in [0, \infty)^{K \times K}$ and a sparsity parameter $\rho_n > 0$. Additionally, the rows of $B_0$ are all distinct, $\min\limits_{k, k^\prime \in [K]} B_{0, kk^\prime} > c_1$ for some $c_1 > 0$, and the smallest eigenvalue of $B_0$ is strictly positive.
    \item $\rho_n = \Omega(n^{-1} \log{n})$.
    \item There exists a $\gamma \in (0,1)$ such that $n_k \geq \gamma n$, for all  $k \in [K]$.
\end{itemize}
\end{assumption}

We show that the distance between the computed loss and the oracle loss $\left\lvert L\left(A_{\S^c}, \hat{P}_{\S^c}^{(K)}\right) \right.-$ $\left. L(A_{\S^c}, P_{\S^c})\right\rvert$ is very small for $\K = K$ and $L\left(A_{\S^c}, \hat{P}_{\S^c}^{(\K)}\right) \gg L(A_{\S^c}, P_{\S^c})$ for $\K < K$. To show the former, we prove that $\hat{P}^{(K)}$, although estimated from the training set $\S$, is a good estimator of the test set edge probabilities $P_{\S^c}$ as long as the community detection algorithm reasonably accurately recovers the community labels for each subnetwork. To establish accuracy, we use \citet[Theorem~4]{sonnet_comm} which bounds the error rate of \code{SC} on a random subset of the network. For the latter case of $\K < K$, we use the pigeonhole principle to show that there is at least one incidence of multiple true communities merging into one, and the number of node pairs in the test set corresponding to those merged communities is large with high probability. The prediction accuracy is low for these node pairs, making the computed loss very large for $\K < K$.

\begin{theorem}\label{theorem:sbm_lossdiff}
    If $A \sim SBM(n, K, B)$, then under the conditions in Assumption \ref{assump:SBM}, the following bounds hold for \code{NETCROP} for SBM (Algorithm \ref{algo:SBM_DCBM}), run with pointwise squared error loss $\ell(a, p) = (a - p)^2$, as $n \to \infty$:
    \begin{align}
       &L\left(A_{\S^c}, \hat{P}^{(\K)}_{\S^c}\right) - L(A_{\S^c}, P_{\S^c}) = \Omega_{p}\left(s(s-1)m^2 \rho_n^2 \right),\ \text{for} \ \K < K, \label{eqn:SBM:lower}\\
       &\left\lvert L\left(A_{\S^c}, \hat{P}^{(\K)}_{\S^c}\right) - L(A_{\S^c}, P_{\S^c}) \right\rvert = O_p(s(s-1)m^2(o+m)^{-1}\rho_n), \ \text{for} \ \K = K. \label{eqn:SBM:upper}
    \end{align}
\end{theorem}
The proof of this theorem and all the other results are provided in Section \ref{app:proofs} of the Appendix. The following consistency result follows directly from Theorem \ref{theorem:sbm_lossdiff}. 
\begin{theorem}\label{coro:sbm_final}
    Under the same setup as in Theorem \ref{theorem:sbm_lossdiff} with $\hat{K}$ being the output number of communities from \code{NETCROP} for SBM (Algorithm \ref{algo:SBM_DCBM}), then as $n \to \infty$,
    \begin{equation*}
        \P( \hat{K} < K ) \to 0.
    \end{equation*}
\end{theorem}

Theorem \ref{theorem:sbm_lossdiff} provides some insights into the effects of the \code{NETCROP} parameters. For (\ref{eqn:SBM:lower}) to be greater than (\ref{eqn:SBM:upper}), the training subnetwork size $o+m$ must be greater than the order of $\rho_n^{-1}$. This holds automatically from the design of \code{NETCROP} as $o+m \asymp n$ and from the assumption $\rho_n = \Omega(n^{-1} \log{n})$. This requirement also implies that for sparser networks with smaller $\rho_n$, larger training subnetworks are required to ensure separation between the two cases.

Additionally, note that for both $\K < K$ and $\K = K$ cases, \code{NETCROP} achieves the same rates for the loss difference as \code{ECV} under similar assumptions \citep[Theorem~3]{cv_edge}. \code{NETCROP} has the same rate for the loss difference as \code{NCV} for $\K < K$ case and a better rate as in (\ref{eqn:SBM:upper}) compared to the rate $O_p(n)$ for \code{NCV}. \code{NCV} also requires a stronger assumption of $\rho_n = \omega\left(n^{-1/2}\right)$ \citep[Theorem~2]{cv_comm}. Thus, \code{NETCROP} achieves similar or better rates for the loss difference than the existing methods under milder or equivalent assumptions, and is significantly faster than them.

\subsection{Theory for DCBM}\label{sec:theory:dcbm}
Here we consider the model selection problem of choosing the number of communities $\hat{K}$ from a candidate set of DCBM with the number of communities in $[K_{max}]$. \code{NETCROP} for blockmodels (Algorithm \ref{algo:SBM_DCBM}) is applied with \code{SSC} (Algorithm \ref{algo:spher_spec_clus}), $s$ subnetworks, overlap size $o$ (thus, non-overlap size $m = (n-o)/s$), $R = 1$ repetition, and the squared error loss. We assume that the true model is DCBM with $K$ communities. We also assume that the test set edge probabilities are estimated using (\ref{eqn:dc.est.eigen}). We need additional identifiability constraints on the degree heterogeneity parameters $\psi_i$ along with similar assumptions as in Assumption \ref{assump:SBM}.
\begin{assumption}{(Model assumptions for DCBM)}\label{assump:DCBM}
    \begin{itemize}
        \item Let $K$ be fixed, and $g_i \in [K]$ and $\psi_i > 0$ be the true community membership and degree-correction parameter of node $i \in [n]$. Then, for $1 \leq i \neq j \leq n$, $0 \leq P_{ij} = B_{g_i g_j} \psi_i \psi_j = \rho_n B_{0, g_i g_j} \psi_i \psi_j = \rho_n P^0_{n, ij} \leq 1$ for a fixed symmetric matrix $B_0 \in [0, \infty)^{K \times K}$ and a sparsity parameter $\rho_n > 0$. Additionally, the rows of $B_0$ are all distinct, $\min\limits_{k, k^\prime \in [K]} B_{0, kk^\prime} > c_1$ for some $c_1 > 0$, and the smallest eigenvalue of $B_0$ is strictly positive.
        \item $\rho_n = \omega\left(n^{-1/3}\right)$.
        \item There exists a $\gamma \in (0,1)$ such that $n_k \geq \gamma n$, for all  $k \in [K]$.
        \item $\max\limits_{i \in G_k} \psi_i = 1$ for all $k \in [K]$.
        \item There exists a constant $\psi_0 > 0$ such that $\min\limits_{i \in [n]} \psi_{i} \geq \psi_0$.
    \end{itemize}
\end{assumption}

For $\K < K$, the same argument as SBM is used with some additional steps for adjusting the degree heterogeneity parameter in the distance between the computed loss and the oracle loss. The $\K = K$ case for DCBM is different as the type of estimator used to predict the test set edge probabilities is different from the SBM case. For $\K = K$, we follow a similar argument as \citet{cv_comm} to show that $\hat{B}_q^{\prime(K)}$ and $\psi^{\prime(K)}$ from the $q$-th subnetwork are good estimates of scaled $B^\prime$ and community normalized $\psi^\prime$, provided \code{SSC} has a low misclustering rate on the subnetwork \citep[Theorem~6]{sonnet_comm}. This ensures that $\hat{P}_{ij}^{\prime(K)}$ is an accurate estimator of $P_{ij}$, making the computed loss closely concentrated around the oracle loss.

\begin{theorem}\label{theorem:dcbm_lossdiff}
    If $A \sim DCBM(n, K, \Psi, B)$, then under the conditions in Assumption \ref{assump:DCBM}, the following bounds hold for \code{NETCROP} for DCBM (Algorithm \ref{algo:SBM_DCBM}), run with pointwise squared error loss $\ell(a, p) = (a - p)^2$, as $n \to \infty$:
    \begin{align}
       &L\left(A_{\S^c}, \hat{P}^{(\K)}_{\S^c}\right) - L(A_{\S^c}, P_{\S^c}) = \Omega_p\left(s(s-1)m^2 \rho_n^2\right), \ \text{for} \ \K < K, \label{eqn:dc-lower}\\
       &\left\lvert L\left(A_{\S^c}, \hat{P}^{(\K)}_{\S^c}\right) - L(A_{\S^c}, P_{\S^c}) \right\rvert = O_p\left(s(s-1)m^2 n^{-\frac{1}{2}} \rho_n^\frac{1}{2}\right), \ \text{for} \ \K = K. \label{eqn:dc-upper}
    \end{align}
\end{theorem}
From a direct application of Theorem \ref{theorem:dcbm_lossdiff}, we have the following theorem.
\begin{theorem}\label{coro:DCBM_final}
        Under the same setup as in Theorem \ref{theorem:dcbm_lossdiff} with $\hat{K}$ being the output number of communities from \code{NETCROP} for DCBM (Algorithm \ref{algo:SBM_DCBM}), then as $n \to \infty$,
    \begin{align*}
        \P\left( \hat{K} < K \right) \to 0.
    \end{align*}
\end{theorem}

Note that the assumption $\rho_n = \omega(n^{-1/3})$ ensures that the lower bound of the difference between the computed loss and the oracle loss in (\ref{eqn:dc-lower}) for $\K < K$ is of larger order than the upper bound of the same for $\K = K$. To the best of our knowledge, these are the only theoretical bounds of the loss difference for detecting the number of communities in a DCBM using cross-validation. \code{NCV} has a partial result on the consistency of $\hat{P}_{ij}^{\prime(K)}$ on all but a vanishing subset of node pairs for $\K = K$, under the assumption $\rho_n = \omega\left(n^{-1/3}\right)$ \citep[Theorem~4]{cv_comm}. \code{NETCROP} also attains a similar consistency (see \eqref{eqn:V-pdiff-max} of the Appendix) under the same assumption. This consistency result is further used to derive the upper bound in \eqref{eqn:dc-upper}.

\subsection{Theory for RDPG}\label{sec:theory_RDPG}
In this section, we present the theoretical guarantee of \code{NETCROP} for selecting the latent space dimension in an RDPG (Algorithm \ref{algo:RDPG}). The theory for RDPG follows a similar track as that of SBM and DCBM. To establish the bounds on computed loss for RDPG with true and unknown latent space dimension $d$, we show that the predicted edge probabilities for the node pairs in the test set are inaccurate for the candidate dimension $\d < d$ and are accurate when $\d = d$.

For $\d=d$, Theorem \ref{theorem:theo-stitch-mat} in the Appendix ensures the existence of orthogonal rotation matrices $W_1, \ldots, W_s \in \O_d$ such that $\hat{X}_1 W_1 \approx X_{S_{01} \cdot}, \ldots, \hat{X}_s W_s \approx X_{S_{0s} \cdot}$, where the approximate symbol is used as a shorthand to represent that the two vectors are aligned up to a small error. This means that the latent position estimate from the $q$th subnetwork can be orthogonally rotated using $W_q$ to be close to the true latent position for each $q \in [s]$, i.e., after the orthogonal rotations, all the estimated latent positions from different subnetworks are aligned close to the true positions. Since the true positions are unknown, we align the estimated latent positions from different subnetworks within themselves using Procrustes transformation. An arbitrary subnetwork is chosen as the standard (for this paper, we take the first subnetwork as the standard) and the estimated latent positions of the other subnetworks are aligned with that of the standard subnetwork using the overlap nodes. As defined in Algorithm \ref{algo:RDPG}, $\hat{W}_q$ matches $\hat{X}_{q, S_0 \cdot}$ with $\hat{X}_{1, S_0\cdot}$ for $q \in [2:s]$. Setting $\hat{W}_1 = \I_d$, we have $\hat{X}_1 \hat{W}_1 \approx \cdots \approx \hat{X}_s \hat{W}_s$, which, in turn, are closely aligned to $XW$ for some $W \in \O_d$ \citep{chakraborty2025scalable}. Since the orthogonal matrix $W$ can be absorbed into the definitions of the theoretical orthogonal rotations $W_1, \ldots, W_s$, we derive the theoretical results for the estimation error of $\hat{X}_q W_q$ for estimating $X_{S_{0q}\cdot}$, $q \in [s]$.

For $\d < d$, the difference between the computed loss and the oracle loss reduces to the accuracy of $\hat{P}^{(\d)}$ for estimating $P$ on the test set. Using Lemma \ref{lemma:rdpg_php} in the Appendix, it is shown that the rank $\d$ truncated version of $P$, denoted by $P^{(\d)}$, is distinct from $P$ on a sufficiently large subset of the test set. We also show that the estimation error of $\hat{P}^{(\d)}$ for $P^{(\d)}$ is bounded by the error of estimating $P$ using $\hat{P}^{(d)}$. Finally, we show that the latter is dominated by the former, providing a lower bound for the loss difference.

Next, we state the results establishing the consistency of \code{NETCROP} for selecting the latent space dimension in RDPG. Since the predicted edge probabilities are obtained using the largest $\d$ eigenvalues and the corresponding eigenvectors, \code{NETCROP} requires the eigenvalues and eigenvectors of the adjacency matrix to be sufficiently distinct. Assumption \ref{assump:RDPG} lists the assumptions on the eigenvalues and the eigenvectors along with the assumptions on the network sparsity.

\begin{assumption}{(Model assumptions for RDPG)}\label{assump:RDPG}
\begin{enumerate}
    \item \label{assump:RDPG:p-define} $P = \rho_n P^0_n$ for some sparsity parameter $\rho_n > 0$ and an $n \times n$ matrix $P^0_n$ with positive entries, as defined in Section \ref{sec:method}, such that $\min\limits_{i, j \in [n]} P^0_{n, ij} \geq c_1$ and $\max\limits_{i, j \in [n]} P^0_{n, ij} \leq c_2$ for some $c_1, c_2 > 0$.
    
    \item \label{assump:RDPG:rho-cond} There exists $a_0 > 0$ and a constant $c_0>0$ such that $\rho_n \geq c_0 n^{-1/2} \log^{2+a_0/2}(n)$.
    
    \item \label{assump:RDPG:eigen} If $P^0_n = U_n \Lambda^0 U_n^\top$ is the eigen decomposition of $P^0_n$ with $\Lambda^0 = diag(\lambda_{1}^0, \ldots, \lambda_{d}^0)$, where $\lambda_{1}^0 \geq \cdots \geq \lambda_{d}^0$, there exist $\phi > 0$ and $a > 1$ such that the following holds for sufficiently large $n$:
    \begin{align*}
     n\phi^{-1} \leq \lambda_{d}^0 \leq \cdots \leq \lambda_{1}^0 \leq n\phi\ \mbox{ and }\ \max\limits_{i \in [n]} \frac{n}{d}\lVert U_{n, i\cdot} \rVert^2 \leq a.
    \end{align*}
\end{enumerate}
\end{assumption}

The general idea behind the RDPG results is the same as the SBM and the DCBM cases. We show that for a candidate dimension $\d < d$, the computed loss is much larger than the oracle loss and for $\d = d$, the true latent space dimension, the computed loss is centred closely around the oracle loss.

\begin{theorem}\label{theorem:rdpg_lossdiff}
    If $A \sim RDPG(n, d, X)$, then under the conditions in Assumption \ref{assump:RDPG}, the following bounds hold for \code{NETCROP} for RDPG (Algorithm \ref{algo:RDPG}), run with pointwise squared error loss $\ell(a, p) = (a - p)^2$, as $n \to \infty$:
    \begin{align}
       &L\left(A_{\S^c}, \hat{P}^{(\d)}_{\S^c}\right) - L(A_{\S^c}, P_{\S^c}) = \Omega_p\left( s(s-1)m^2 \rho_n^2 \right), \ \ \text{for} \ \d < d,\\
       &\left\lvert L\left(A_{\S^c}, \hat{P}^{(\d)}_{\S^c}\right) - L(A_{\S^c}, P_{\S^c}) \right\rvert = O_p\left( s(s-1)m^2 \frac{\rho_n \log^2(o+m)}{\sqrt{o+m}} \right), \ \ \text{for} \ \d = d.\label{eqn:rdpg:upper}
    \end{align}
\end{theorem}

\begin{theorem}\label{theorem:RDPG_main}
    Under the same setup as in Theorem \ref{theorem:rdpg_lossdiff} with $\hat{d}$ being the output latent space dimension from \code{NETCROP} for RDPG (Algorithm \ref{algo:RDPG}), then as $n \to \infty$, 
    \begin{align}
        \P\left(\hat{d} < d \right) \to 0.
    \end{align}
\end{theorem}

Note that for both $\d < d$ and $\d = d$ cases, \code{NETCROP} achieves the same rates for the loss difference as \code{ECV} under similar assumptions (\citet[Theorem~2]{cv_edge}, \citet{li_correction_2025}). Thus, \code{NETCROP} achieves similar rates for the loss difference to the existing methods under equivalent assumptions while being significantly faster than the competing method.

\section{Numerical Results}\label{sec:numerical}
We apply \code{NETCROP} (\href{https://github.com/sayan-ch/NETCROP.git}{GitHub Repository}) for model selection on a few simulated and real networks, comparing its performance with two other cross-validation methods for networks: \code{NCV} and \code{ECV}.
First, we applied \code{NETCROP} for detecting the number of communities and degree correction on networks, modeled by SBM and DCBM. Here, \code{NETCROP} was compared with both \code{ECV} and \code{NCV}. In the second part, we applied \code{NETCROP} for detecting the latent space dimension of RDPG. In this case, \code{NETCROP} was only compared to \code{ECV} as \code{NCV} was not designed for this problem. In the third part, we applied \code{NETCROP} for detecting the latent space dimension in latent space models. 
Finally, we applied \code{NETCROP} for tuning the regularizing parameter in regularized spectral clustering for sparse DCBM. In Section \ref{app:numerical:small} of the Appendix, we have included some additional simulations for selecting the number of communities and degree heterogeneity in smaller networks.

In all the instances in this section, \code{NETCROP} was applied with the parameters chosen using the parameter selection strategy outlined in Section \ref{sec:method:partune} with test proportion $p_{test} = 0.02$ and overlap size $o = n p_o \approx n (1 - \sqrt{2p_{test}}) = 0.8 n$ that produced the number of subnetworks $s = 3$. 
We ran each case with $R \in \{1, 5\}$ repetitions.
For all the applications, we used the authors' recommended parameters for \code{NCV} ($3$ folds), \code{ECV} ($3$ folds and $10\%$ held-out nodes), and their stabilized versions \code{NCV+St} and \code{ECV+St} ($20$ repetitions)
Unless otherwise mentioned, the computations in this section were done on a university campus cluster with $40$ processors and $200$ gigabytes of RAM.
R version $4.5.1$ was used for the computations. The codes for \code{NCV} and \code{ECV} were taken from the CRAN package \textsl{randnet} \citep{randnet-package}. All algorithms were parallelized whenever possible.

\subsection{Simulation Examples}\label{sec:numerical:sim}
\subsubsection*{Detecting the Number of Communities and Degree Correction in SBM and DCBM}
We simulated $100$ networks from the true model (SBM or DCBM) with $n = 10000$ nodes and $K$ communities. Each node was randomly assigned one of the $K$ communities. The $K \times K$ community connectivity matrix was formed as $B = \alpha((1-\beta) I_K + \beta J_K)$. 
The out-in ratio $\beta$ is related to the strength of the communities and $\alpha$ determines the sparsity of the network. 
Additionally, the degree parameters $\psi_i$'s were generated from the $1 / Beta(4,1)$ distribution, independently for each node. The adjacency matrix was generated as $A_{ij} \sim \operatorname{Bernoulli}(P_{ij})$ for the true model SBM and $A_{ij} \sim \operatorname{Bernoulli}(P_{ij}\psi_i \psi_j)$ for the true model DCBM, independently for each node pair $i < j$. 
The candidate sets of number of communities for true $K=5, 10,$ and $20$ were $\K = [10], [20],$ and $[30]$, respectively. 
For label matching, a simpler adaptation of the Hungarian algorithm \citep{kuhn1955hungarian}, called \code{MatchGreedy} (Algorithm \ref{algo:match-greedy}, \citet{mukherjee2017provably}), was used for faster computations.
Spectral clustering (Algorithm \ref{algo:spec_clus}) and spherical spectral clustering (Algorithm \ref{algo:spher_spec_clus}) were used to fit SBM and DCBM, respectively.
The accuracy of each method is measured by the percentage of times they were able to detect the true $K$ and the degree correction out of the $100$ simulations. We also reported the mean absolute deviation (MAD) from the true value of $K$ and the computation time. The simulation results are in Table \ref{tab:sim:sbm_dcbm}.

\begin{table}[!tb]
      \centering
      {
      \resizebox{\linewidth}{!}{
      \begin{tabular}{c c c c c r c c r }
        \toprule
        \textbf{Network} & \multicolumn{5}{c}{ \textbf{\code{NETCROP}}} & \multicolumn{3}{c}{\textbf{Other Algorithms}} \\
        & \textbf{$s$} & \textbf{$o$} & \textbf{$R$} & \textbf{Accu. \%} \textsuperscript{\ddag} \ \textbf{(MAD}\textsuperscript{\textsection}\textbf{)} & \textbf{Time (sec.}\textsuperscript{\textsection \textsection}\textbf{)} & \textbf{Algo.}\textsuperscript{*} & \textbf{Accu. \% \ (MAD)} & \textbf{Time (sec.)}\\
        \midrule

        
        SBM, $K = 5$ & \cellcolor{lightgray}{$3$} & \cellcolor{lightgray}{$8002$} & \cellcolor{lightgray}{$1$} & \cellcolor{lightgray}{$100 \% \ (0)$} & \cellcolor{lightgray}{$6.8$} & \cellcolor{lightgray}\code{ECV} & \cellcolor{lightgray}$100 \% \ (0)$ & \cellcolor{lightgray} {$118.2$}\\

        $\beta = 0.3, \alpha = 0.1$ & & & {$5$} & {$100 \% \ (0)$} & {$12.1$} &  \code{ECV+St}\textsuperscript{**} & $100 \% \ (0)$ & {$912.9$} \\

        $\lambda = 439.8$\textsuperscript{\dag} & & & & & & \cellcolor{lightgray}\code{NCV} & \cellcolor{lightgray}{$97 \% \ (0.04)$} & \cellcolor{lightgray} {$52.3$}\\

        & & & & & & \code{NCV+St}\textsuperscript{**} & $100 \% \ (0)$ & {$771.9$} \\
        \\
        
        
        SBM, $K = 20$ & \cellcolor{lightgray}{$3$} & \cellcolor{lightgray} {$8002$} & \cellcolor{lightgray}{$1$} & \cellcolor{lightgray}{$100 \% \ (0)$}  & \cellcolor{lightgray}{$21.2$} & \cellcolor{lightgray}\code{ECV} & \cellcolor{lightgray}{$0 \% \ (2.81)$}  & \cellcolor{lightgray}{$442.2$} \\

        $\beta = 0.33, \alpha = 0.3$ & & & {$5$} & {$100 \% \ (0)$} & {$77.4$} &  \code{ECV+St} & {$0 \% \ (2.67)$} & {$3300.3$}\\

        $\lambda = 1099.8$ & & & & & & \cellcolor{lightgray}\code{NCV} & \cellcolor{lightgray}{$0 \% \ (3.18)$}  & \cellcolor{lightgray}{$353.0$}\\

        & & & & & & \code{NCV+St} & {$0 \% \ (2.72)$} & {$2947.3$}\\
        \\


        DCBM, $K = 10$& \cellcolor{lightgray}{$3$} & \cellcolor{lightgray}{$8002$} & \cellcolor{lightgray}{$1$} & \cellcolor{lightgray}{$73 \% (0.51)$}  & \cellcolor{lightgray}{$2.3$} & \cellcolor{lightgray}\code{ECV} & \cellcolor{lightgray}{$54 \% \ (0.95)$}  & \cellcolor{lightgray}{$168.3$}\\

        $\beta = 0.2, \alpha = 0.5$ & & & {$5$} & {$96 \% (0.04)$} & {$5.0$} &  \code{ECV+St} & {$57 \% \ (0.77)$}  &{$503.2$}\\

        $\lambda = 79.4$& & & & & & \cellcolor{lightgray}\code{NCV} & \cellcolor{lightgray}{$36 \% \ (1.08)$}  & \cellcolor{lightgray}{$67.8$} \\

        & & & & & & \code{NCV+St} & {$39 \% \ (0.97)$}  & {$451.3$} \\
        \\


        DCBM, $K = 20$& \cellcolor{lightgray}{$3$} & \cellcolor{lightgray}{$8002$} & \cellcolor{lightgray}{$1$} & \cellcolor{lightgray}{$94 \% (0.09)$}  & \cellcolor{lightgray}{$13.0$} & \cellcolor{lightgray}\code{ECV} & \cellcolor{lightgray}{$70 \% (1.06)$} & \cellcolor{lightgray}{$346.7$}\\

        $\beta = 0.33, \alpha = 3$ & & & {$5$} & {$99 \% (0.02)$}  & {$44.1$}  & \code{ECV+St} & {$76 \% (0.65)$} & {$1539.3$} \\

        $\lambda = 659.1$ & & & & & & \cellcolor{lightgray}\code{NCV} & \cellcolor{lightgray}{$10 \% \ (4.77)$}  & \cellcolor{lightgray}{$165.4$}\\

        & & & & & & \code{NCV+St} & {$16 \% (4.17)$}  &{$1884.9$}\\
        
    \bottomrule
      \end{tabular}}
      }
      \caption{\normalsize{Results for detecting $K$ and degree Correction in simulated networks from SBM and DCBM.} \newline {\footnotesize \textsuperscript{\dag}$\lambda$ is the mean of the average densities of the simulated networks; \textsuperscript{\ddag}\textbf{Accu. \%} denotes the percentage of times the correct model was selected out of $100$ simulations; \textsuperscript{\textsection}\textbf{MAD} denotes Mean Absolute Deviation from true $K$; \textsuperscript{\textsection\textsection}\textbf{sec.} denotes seconds; \textsuperscript{*}\textbf{Algo.} denotes Algorithm; \textsuperscript{**}\code{ECV+St} and \code{NCV+St} represent the stabilized versions of \code{ECV} and \code{NCV}, respectively, with 20 repetitions.}}\label{tab:sim:sbm_dcbm}
      
    \end{table}

Based on the findings presented in Table \ref{tab:sim:sbm_dcbm}, \code{NETCROP} demonstrates superior accuracy compared to \code{ECV} and \code{NCV} in estimating both the number of communities $K$ and the degree correction across all scenarios. Moreover, \code{NETCROP} was observed to be approximately $7$ to $100$ times faster than the other methods, depending on the simulation setting.
In scenarios involving a true SBM model with a large $K = 20$, \code{NETCROP} achieves an accuracy level of $100 \%$ in just $21.2$ seconds, whereas \code{ECV}, \code{NCV}, and their stabilized versions achieve
$0 \%$ accuracy with much longer time.
In the case of true DCBM models, \code{NETCROP} demonstrates notably higher accuracy compared to \code{ECV} and \code{NCV}, while also requiring significantly less computational time. Repetition improved the performance of \code{NETCROP} significantly in all the simulation scenarios. 
Merely five repetitions were enough to stabilize the outcomes of \code{NETCROP} compared to twenty repetitions for \code{NCV} and \code{ECV}.
In general, \code{NETCROP} achieved lower MADs from the true number of communities compared to the other methods. {We also observed that all the methods were able to detect the correct degree heterogeneity (SBM or DCBM) with $100\%$ accuracy in all the simulation settings. Thus, we did not report those accuracies in Table \ref{tab:sim:sbm_dcbm}}.

\subsubsection*{Estimating the Latent Space Dimension of RDPG}
We simulated networks from RDPG with various levels of sparsity and applied \code{NETCROP} for RDPG (Algorithm \ref{algo:RDPG}), \code{ECV} and \code{ECV+St} on them. We generated the latent positions of the $n$ nodes as $X_{ik} \sim U(0,1)$ independently for $i \in [n]$, $k \in [d]$. Then the probability matrix was computed as $P = \zeta XX^\top/max(XX^\top)$, where $\zeta$ was used to control the sparsity of the simulated networks. We fixed $n = 10000$ and $d = 10$, and varied the sparsity parameter $\zeta \in \{0.75, 0.70, 0.65\}$ to generate the networks. For each simulation setting, $100$ networks were generated. 
\code{NETCROP} was applied on each simulated network with the largest candidate dimension $d_{max} = 20$. Accuracy is defined as the percentage of times \code{NETCROP} detects the true dimension out of the $100$ simulated networks. 
We also reported the MAD from the true latent space dimension. The results are in Table \ref{tab:sim_rdpg}.

\begin{table}[!tb]
      \centering
      {
      \resizebox{\linewidth}{!}{
      \begin{tabular}{c c c c c c c r c c r}
        \toprule
        & & & \multicolumn{5}{c}{ \textbf{\code{NETCROP}}} & \multicolumn{3}{c}{\textbf{Other Algorithms}} \\
        $n$ & $d$ & $\zeta$ & \textbf{$s$} & \textbf{$o$} & \textbf{$R$} & \textbf{Accu. \%} \textbf{(MAD)} & \textbf{Time (sec.)} & \textbf{Algo.} & \textbf{Accu. \% \ (MAD)} & \textbf{Time (sec.)}\\
        \midrule

        $10^4$ & $10$ & $0.75$ & \cellcolor{lightgray}{$3$} & \cellcolor{lightgray}{$8002$} & \cellcolor{lightgray}{$1$} & \cellcolor{lightgray}{$100 \% \ (0)$} & \cellcolor{lightgray}{$14.7$} & \cellcolor{lightgray}\code{ECV} & \cellcolor{lightgray}{$85 \% \ (0.25)$} &\cellcolor{lightgray} {$157.9$}\\

        & & & & & {$5$} & {$100 \% \ (0)$} & {$30.5$} & \code{ECV+St} & {$100 \% (0)$} & {$1133.5$} \\
        \\

        $10^4$ & $10$ & $0.70$ & \cellcolor{lightgray}{$3$} & \cellcolor{lightgray}{$8002$} & \cellcolor{lightgray}{$1$} & \cellcolor{lightgray}{$99 \% \ (0.01)$} & \cellcolor{lightgray}{$14.0$} & \cellcolor{lightgray}\code{ECV} & \cellcolor{lightgray}{$27 \% (2.30)$} &\cellcolor{lightgray}{$190.6$}\\

        & & & & & {$5$} & {$100 \% \ (0)$} & {$29.4$} & \code{ECV+St} & {$53 \% (1.6)$} & {$1150.7$} \\
        \\

        $10^4$ & $10$ & $0.65$ & \cellcolor{lightgray}{$3$} & \cellcolor{lightgray}{$8002$} & \cellcolor{lightgray}{$1$} & \cellcolor{lightgray}{$100 \% \ (0)$} & \cellcolor{lightgray}{$13.4$} & \cellcolor{lightgray}\code{ECV} & \cellcolor{lightgray}{$1 \% (6.7)$} &\cellcolor{lightgray}{$179.8$}\\

        & & & & & {$5$} & {$100 \% \ (0)$} & {$28.3$} & \code{ECV+St} & {$0 \% (6.8)$} & {$1122.0$} \\
    \\
    \bottomrule
      \end{tabular}}
      }
      \caption{\normalsize{Results for detecting the latent space dimension $d$ in simulated RDPGs.
      }
      }\label{tab:sim_rdpg}
      
    \end{table}

Table \ref{tab:sim_rdpg} shows that across all the scenarios, \code{NETCROP} outperforms \code{ECV} and its stabilized version in terms of accuracy and MAD, while being approximately $10$ to $40$ times faster. As the sparsity of simulated networks increases (as $\zeta$ decreases), the accuracy of \code{ECV} and \code{ECV+St} drops, while \code{NETCROP} has almost perfect accuracy in all the settings.

\subsubsection*{Estimating the Dimension of Latent Space Model} \label{sec:numerical:latent}
We conducted simulations on networks based on the latent space model outlined in (\ref{eqn:latent.mod}) and used \code{NETCROP} to estimate their latent space dimension.
We used projected gradient descent \citep{ma2020universal} to obtain the estimated latent positions within each training set.
We considered two options for the intercept $\alpha \in \{0,1\}$ and two alternatives for the latent dimension $d \in \{2, 5\}$. $100$ networks of size $n = 1000$ were generated using the latent space model for each combination of $\alpha$ and $d$. Subsequently, \code{NETCROP} was applied to each simulated network with a candidate set of $\d = [5]$ for the $d = 2$ scenario and $\d = [10]$ for the $d = 5$ scenario.
We present the accuracy, defined as the percentage of correct dimension selections out of 100 simulations, the MAD from the true $d$, and the computation time in seconds. The results are provided in Table \ref{tab:sim:latent}.

\begin{table}[!tb]
      \centering
      {
      \begin{NiceTabular}{c c c c c c c r}
        \toprule
         
        \textbf{$n$} & \textbf{$\alpha$} & \textbf{$d$} & \textbf{$s$} & \textbf{$o$} & \textbf{$R$} & \textbf{Accu. \%} \textbf{(MAD)} & \textbf{Time (sec.)} \\
        \midrule

        $10^3$ & $0$ & $2$ & \cellcolor{lightgray}{$3$} & \cellcolor{lightgray}{$802$} & \cellcolor{lightgray}{$1$} & \cellcolor{lightgray}{$99 \% (0.01)$} & \cellcolor{lightgray}{$5.8$}\\
         
        & & &  &  & {$5$} & {$100 \% (0)$} & {$27.3$}\\
        \\

        $10^3$ & $1$ & $2$ & \cellcolor{lightgray}{$3$} & \cellcolor{lightgray}{$802$} & \cellcolor{lightgray}{$1$} & \cellcolor{lightgray}{$100 \% (0)$} & \cellcolor{lightgray}{$6.0$}\\
         
         & & &  &  & {$5$} & {$100 \% (0)$} & {$33.5$}\\
        \\

        $10^3$ & $0$ & $5$ & \cellcolor{lightgray}{$3$} & \cellcolor{lightgray}{$802$} & \cellcolor{lightgray}{$1$} & \cellcolor{lightgray}{$100 \% (0)$} & \cellcolor{lightgray}{$11.1$}\\
         
        & & &  &  & {$5$} & {$100 \% (0)$} & {$55.1$}\\
        \\

        $10^3$ & $1$ & $5$ & \cellcolor{lightgray}{$3$} & \cellcolor{lightgray}{$802$} & \cellcolor{lightgray}{$1$} & \cellcolor{lightgray}{$99 \% (0.01)$} & \cellcolor{lightgray}{$8.5$}\\
         
        & & &  &  & {$5$} & {$100 \% (0)$} & {$57.2$}\\
         
        \bottomrule
      \end{NiceTabular}}
      \caption{\normalsize{Results for detecting the latent space dimension $d$ using \code{NETCROP} in simulated networks from latent space models.
      } 
      }\label{tab:sim:latent}
      
    \end{table}

We observe that \code{NETCROP} with $R = 1$ repetition had $99\%$ to $100\%$ accuracy for all four cases in Table \ref{tab:sim:latent}, while taking approximately $6$ to $11$ seconds per simulation. \code{NETCROP} with $R = 5$ repetitions achieved perfect accuracy in approximately $27$ to $57$ seconds per simulation.
Overall, \code{NETCROP} effectively identified the dimension $d$ in latent space models even with one repetition.


\subsubsection*{Parameter Tuning for Regularized Spectral Clustering}
We applied \code{NETCROP} for tuning the regularizing parameter in regularized spectral clustering. We generated $100$ DCBM networks with $n = 10000$ nodes, $K = 5$ communities, $\beta = 0.33$, and $\alpha = 0.3$. We then applied \code{NETCROP} to find an appropriate value of the tuning parameter $\tau$ from a candidate set of $\{0, 0.1, \ldots, 2\}$. 
We used $R \in \{1, 5\}$ repetitions. For the $R = 5$ case, we used the mean and the mode of the estimated $\hat{\tau}$ from \code{NETCROP}. Then, we used those estimated $\hat{\tau}$ to perform regularized spectral clustering on the entire network to compute the average clustering accuracies with $\tau = 0$ for $\tau = 0$ (no regularization), the oracle estimator of $\tau$ given by the value of $\tau$ in the candidate set that produces the highest clustering accuracy,
$\hat{\tau}$ selected by \code{NETCROP} with $R = 1$, and $\hat{\tau}_{mean}$ and $\hat{\tau}_{mode}$ from \code{NETCROP} with $R = 5$. We compared our method with the performance of the Davis-Kahan estimator of $\tau$ \citep{joseph2016impact}. We used $\ell_2$ loss function for all the implementations of \code{NETCROP}. We report these average clustering accuracies and their standard deviations (sd) in Figure \ref{fig:tune_par}.

\begin{figure}[!tb]
    \centering
    \includegraphics[width=\textwidth, height=8cm, keepaspectratio]{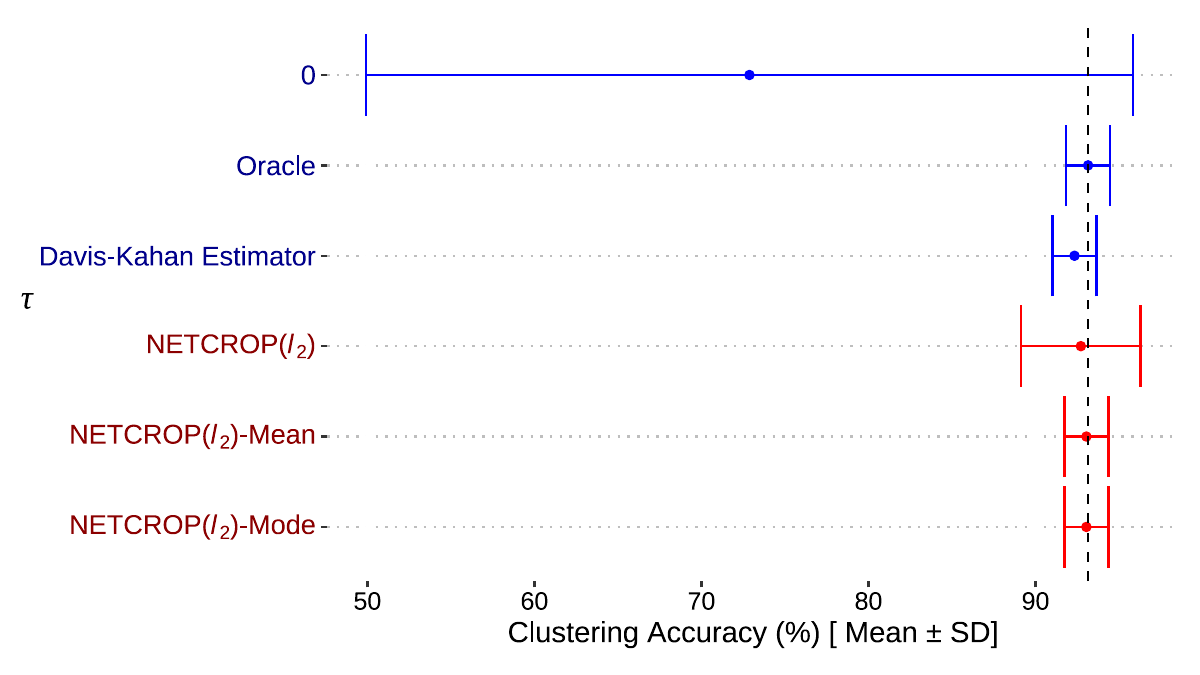} 

    \caption{\normalsize{Mean accuracy ($\%$) of regularized spectral clustering with multiple choices of $\tau$ and with $\hat{\tau}$ chosen from \code{NETCROP} with $R = 1 $, and their mean and mode with $R = 5$.}}
    \label{fig:tune_par}

     \par\smallskip
    {\raggedright
    \noindent\textbf{Alt text:}
    Horizontal point and error bar plot of clustering accuracy. The Oracle, Davis-Kahan estimator and three NETCROP ell-two estimates have averages around 92 to 93 percent with narrow error bars near the dashed Oracle benchmark. NETCROP ell-two mean and mode variants have slightly higher means than Davis-Kahan. Tau equal zero has an average about 73 percent and substantially greater variability.
    \par}
    
\end{figure}

From Figure \ref{fig:tune_par}, it can be observed that the average clustering accuracy is at its lowest when no regularization is applied ($\tau = 0$) with a clustering accuracy of $76.5\%$ with a very large sd of $21.60\%$.
The average accuracies of the oracle and the Davis-Kahan estimator are $93.3\%$ (sd $=1.29\%$) and $92.5\%$ (sd $=1.27\%$), respectively.
Using $\hat{\tau}$, the clustering accuracy is $92.5\%$ with an sd of $5.19\%$. Repetitions help increase the average accuracy of \code{NETCROP} close to the oracle accuracy and decrease the standard deviation.
Mean of $\hat{\tau}$ from 5 repetitions of \code{NETCROP} has average clustering accuracy of $93.2\%$ with sd $1.30\%$. The same for mode is $93.2\%$ with sd $1.29\%$.
In summary, \code{NETCROP} effectively selects the optimal tuning parameter that produced average accuracies closest to that of the oracle estimator and combining outcomes from multiple repetitions improved the variation in accuracies significantly. Most instances of \code{NETCROP} outperformed the Davis-Kahan estimator as well.

\subsection{Real Data Examples}\label{sec:numerical:realdata}
We considered two real data examples to demonstrate the performance of \code{NETCROP} for detecting the number of communities and degree correction in a blockmodel. The first dataset is \href{https://gammagl.readthedocs.io/en/latest/generated/gammagl.datasets.DBLP.html}{DBLP} (Digital Bibliography \& Library Project) four-area network. It was curated by \cite{Gao_DBLP} and \cite{Ji_DBLP} and previously analyzed by \cite{sengupta2015spectral}. The data consist of $n=4057$ data mining researchers from the research areas of database, data mining, information retrieval, and artificial intelligence. Two nodes are connected if the authors have presented at the same conference, and there are $K=4$ ground-truth communities representing the research areas. The other real network is \href{https://snap.stanford.edu/data/twitch_gamers.html}{Twitch gamers social network}  \citep{sarkar2021twitch}. It has $n=32407$ Twitch users from $K=20$ language communities. Two Twitch users are connected by an edge if they have a mutual follower.

We applied \code{NETCROP}, \code{ECV} and \code{NCV} for detecting the number of communities and degree correction in blockmodels on DBLP and Twitch networks. The candidate number of communities was $\K = [10]$ for DBLP and $\K = [30]$ for Twitch. 
Community detection was performed on the real networks with the predicted $\hat{K}$ and using spectral clustering if the predicted model was SBM and spherical spectral clustering if it was DCBM. The results are in Table \ref{tab:real:sbm_dcbm}.

\begin{table}[!tb]
    \centering{
      \resizebox{\linewidth}{!}{
      \begin{tabular}{c c c c c r c c c r c}
        \toprule
        \textbf{Real} & \multicolumn{5}{c}{ \textbf{\code{NETCROP}}} & \multicolumn{3}{c}{\textbf{Other Algorithms}} \\
        \textbf{Network} & \textbf{$s$} & \textbf{$o$} & \textbf{$R$} & \textbf{Mode (\%Freq.)} & \textbf{Time (sec.)} & \textbf{AUC} & \textbf{Algo.} & \textbf{Mode (\%Freq.)} & \textbf{Time (sec.)} & \textbf{AUC}\\
        \midrule

        DBLP & \cellcolor{lightgray}{$3$} & \cellcolor{lightgray}{$3247$} & \cellcolor{lightgray}{$1$} & \cellcolor{lightgray}{DCBM-$4$ \ $(61\%)$} & \cellcolor{lightgray}{$2.0$} & \cellcolor{lightgray}{$0.949$} & \cellcolor{lightgray}\code{ECV} & \cellcolor{lightgray}SBM-$10$ \ $(100\%)$ & \cellcolor{lightgray}{$23.6$} & \cellcolor{lightgray} {$0.866$} \\

        & & & {$5$} & {DCBM-$4$ \ $(93\%)$} & {$4.4$} & {$0.943$} & \code{ECV+St}\textsuperscript{*} & SBM-$10$ \ $(100\%)$ & {$129.2$} & {$0.877$}\\

        & & & \cellcolor{lightgray}{$10$} & \cellcolor{lightgray}{DCBM-$4$ \ $(97\%)$} & \cellcolor{lightgray}{$7.4$} & \cellcolor{lightgray}{$0.947$} &\cellcolor{lightgray}\code{NCV} & \cellcolor{lightgray}SBM-$10$ \ $(100\%)$ & \cellcolor{lightgray}{$10.2$} & \cellcolor{lightgray}{$0.863$}\\

        & & & {$20$} & {DCBM-$4$ \ $(100\%)$} & {$13.5$} & {$0.950$} & \code{NCV+St}\textsuperscript{*} & SBM-$10$ \ $(100\%)$ & {$60.3$} & {$0.867$}\\
        \\
        
        & {$3$} & {$25927$} & {$1$} & {DCBM-$20$ \ $(85 \%)$} & {$70.5$} & {$0.924$} &  &  &  & \\

        & & & {$5$} & {DCBM-$20$ \ $(92 \%)$} & {$230.3$} & {$0.932$} &  &  &  & \\
        \bottomrule
      \end{tabular}}
    }
    \caption{\normalsize{Results for detecting the number of communities $K$ and degree correction in DBLP and Twitch networks.}}\label{tab:real:sbm_dcbm}
    
\end{table}

In Table \ref{tab:real:sbm_dcbm}, the performance of \code{NETCROP} stands out in terms of accurately selecting the number of communities and achieving a higher AUC compared to \code{ECV}, \code{NCV}, and their stabilized versions for the DBLP network. Moreover, \code{NETCROP} operates significantly faster than the other methods. 
\code{NETCROP} identified DCBM as the preferred model, while \code{NCV} and \code{ECV} favored SBM. Notably, \code{NETCROP} consistently estimated the number of communities as $4$, aligning well with the actual number of communities. In contrast, \code{NCV} and \code{ECV} overestimated it at the highest value of $10$ within the candidate set. Additionally, \code{NETCROP} exhibited superior AUC values compared to \code{ECV} and \code{NCV}, 
while being approximately $5$ to $10$ times faster than them.
For the Twitch network, \code{NETCROP} consistently opted for DCBM with $20$ communities in $85\%$ of cases, achieving an average AUC of $0.924$. The proportion improved to $92\%$ with $5$ repetitions. The estimated number of communities matched the actual number of user languages. The runtime was $70.5$ and $230.3$ seconds for the Twitch network without and with repetitions, respectively. \code{NCV} and \code{ECV} could not be implemented even with a total of $400$ gigabytes of RAM.

\section{Discussion} \label{sec:discussion}
In this paper, a subsampling-based computationally efficient general cross-validation technique for networks has been developed. The method is adaptable to a broad range of model selection and parameter tuning problems for networks. The consistency of \code{NETCROP} has been theoretically analysed for tasks such as identifying the number of communities in SBM and DCBM, as well as determining the latent space dimension of RDPG. Extensive empirical evaluations have been conducted using both simulated and real networks to address issues like community detection and degree correction in blockmodels, latent space dimension estimation in RDPG and latent space models, and optimal parameter selection for regularized spectral clustering for community detection in DCBM. The findings from both theoretical analyses and empirical investigations demonstrate that \code{NETCROP} delivers accurate model selection outcomes in considerably less time compared to existing network cross-validation approaches in the literature.

Current research lacks adequate methods for subsampling or cross-validation in the context of complex networks like dynamic networks, multilayer networks, hypergraphs, and heterogeneous networks. Broadening the concept of \code{NETCROP} for inference, model selection, and parameter tuning for these types of networks will also be an interesting direction for future research.


%
\bibliographystyle{agsm}
\bibliography{reference}

@article{KarrerDCBM,
  title = {Stochastic Blockmodels and Community Structure in Networks},
  author = {Karrer, Brian and Newman, M. E. J.},
  journal = {Phys. Rev. E},
  volume = {83},
  pages = {016107},
  year = {2011}
}

@article{HOLLAND1983109,
title = "Stochastic Blockmodels: First Steps",
journal = "Social Networks",
volume = "5",
pages = "109--137",
year = "1983",
author = "Paul W. Holland and Kathryn Blackmond Laskey and Samuel Leinhardt",
}

@article{lei2015,
author = "Lei, Jing and Rinaldo, Alessandro",
journal = "Annals of Statistics",
pages = "215--237",
publisher = "The Institute of Mathematical Statistics",
title = "Consistency of Spectral Clustering in Stochastic Block Models",
volume = "43",
year = "2015"
}

@article{mukherjee2017provably,
author = {Soumendu Sundar Mukherjee  and Purnamrita Sarkar  and Peter J. Bickel },
title = {Two Provably Consistent Divide-and-Conquer Clustering Algorithms for Large Networks},
journal = {Proceedings of the National Academy of Sciences},
volume = {118},
pages = {e2100482118},
year = {2021}
}

@article{rohe,
author = "Rohe, Karl and Chatterjee, Sourav and Yu, Bin",
journal = "Ann. Statist.",
pages = "1878--1915",
title = "Spectral Clustering and the High-Dimensional Stochastic Blockmodel",
volume = "39",
year = "2011"
}

@article{bickel_sc,
author = {Purnamrita Sarkar and Peter J. Bickel},
title = {Role of Normalization in Spectral Clustering for Stochastic Blockmodels},
volume = {43},
journal = {The Annals of Statistics},
pages = {962--990},
year = {2015}
}

@article{likeli1,
author = {Arash A. Amini and Aiyou Chen and Peter J. Bickel and Elizaveta Levina},
title = {Pseudo-Likelihood Methods for Community Detection in Large Sparse Networks},
volume = {41},
journal = {The Annals of Statistics},
pages = {2097--2122},
year = {2013}
}

@article{likeli2,
author = {Yunpeng Zhao and Elizaveta Levina and Ji Zhu},
title = {Consistency of Community Detection in Networks under Degree-Corrected Stochastic Block Models},
volume = {40},
journal = {The Annals of Statistics},
pages = {2266--2292},
year = {2012}
}

@inproceedings{Gao_DBLP,
title = "Graph-Based Consensus Maximization among Multiple Supervised and Unsupervised Models",
author = "Jing Gao and Feng Liang and Wei Fan and Yizhou Sun and Jiawei Han",
year = "2009",
pages = "585--593",
booktitle = "Advances in Neural Information Processing Systems 22 "
}

@InProceedings{Ji_DBLP,
author="Ji, Ming
and Sun, Yizhou
and Danilevsky, Marina
and Han, Jiawei
and Gao, Jing",
title="Graph Regularized Transductive Classification on Heterogeneous Information Networks",
booktitle="Machine Learning and Knowledge Discovery in Databases",
year="2010",
publisher="Springer",
address="Berlin, Heidelberg",
pages="570--586"
}

@book{Serfling,
  title={Approximation Theorems of Mathematical Statistics},
  author={Serfling, Robert J},
  year={2009},
  publisher={John Wiley \& Sons}
}

@INPROCEEDINGS{fb-network,  
author={Akhtar, Nadeem and Javed, Hira and Sengar, Geetanjali},
booktitle={2013 5th International Conference and Computational Intelligence and Communication Networks},
title={Analysis of Facebook Social Network}, 
year={2013},
volume={},
number={},
pages={451--454}
}

@article{likelihood,
author = {Krzysztof Nowicki and Tom A. B Snijders},
title = {Estimation and Prediction for Stochastic Blockstructures},
journal = {Journal of the American Statistical Association},
volume = {96},
pages = {1077--1087},
year  = {2001}
}

@article{sengupta2015spectral,
  Title                    = {Spectral Clustering in Heterogeneous Networks},
  Author                   = {Sengupta, Srijan and Chen, Yuguo},
  Journal                  = {Statistica Sinica},
  Year                     = {2015},
  Pages                    = {1081--1106},
  Volume                   = {25}
}

@article{cv_edge,
    author = {Li, Tianxi and Levina, Elizaveta and Zhu, Ji},
    title = "{Network Cross-Validation by Edge Sampling}",
    journal = {Biometrika},
    volume = {107},
    pages = {257--276},
    year = {2020}
}

@article{ms1,
author = { Daniel L.   Sussman  and  Minh   Tang  and  Donniell E.   Fishkind  and  Carey E.   Priebe },
title = {A Consistent Adjacency Spectral Embedding for Stochastic Blockmodel Graphs},
journal = {Journal of the American Statistical Association},
volume = {107},
pages = {1119--1128},
year  = {2012}
}

@article{cv_comm,
author = {Kehui Chen and Jing Lei},
title = {Network Cross-Validation for Determining the Number of Communities in Network Data},
journal = {Journal of the American Statistical Association},
volume = {113},
pages = {241--251},
year  = {2018}
}

@article{lei2016goodness,
  title={A Goodness-of-Fit Test for Stochastic Block Models},
  author={Lei, Jing},
  journal={The Annals of Statistics},
  volume={44},
  pages={401--424},
  year={2016}
}

@article{ma2021determining,
  title={Determining the Number of Communities in Degree-Corrected Stochastic Block Models},
  author={Ma, Shujie and Su, Liangjun and Zhang, Yichong},
  journal={The Journal of Machine Learning Research},
  volume={22},
  pages={1--63},
  year={2021}
}

@article{RDPG_survey,
author = {Athreya, Avanti and Fishkind, Donniell E. and Tang, Minh and Priebe, Carey E. and Park, Youngser and Vogelstein, Joshua T. and Levin, Keith and Lyzinski, Vince and Qin, Yichen},
title = {Statistical Inference on Random Dot Product Graphs: A Survey},
year = {2017},
volume = {18},
journal = {The Journal of Machine Learning Research},
pages = {8393--8484}
}

@article{Hoff2002LatentSA,
  title={Latent Space Approaches to Social Network Analysis},
  author={Peter D. Hoff and Adrian E. Raftery and Mark S. Handcock},
  journal={Journal of the American Statistical Association},
  year={2002},
  volume={97},
  pages={1090--1098}
}

@article{cerqueira2023consistent,
      title={Consistent Model Selection for the Degree Corrected Stochastic Blockmodel}, 
      author={Andressa Cerqueira and Sandro Gallo and Florencia Leonardi and Cristel Vera},
      journal = {Latin American Journal of Probability and Mathematical Statistics},
      year={2024},
      volume = {21},
      pages = {267-292}
}

@inproceedings{sarkar2021twitch,
  title={Twitch Gamers: a Dataset for Evaluating Proximity Preserving and Structural Role-based Node Embeddings},
    pages={},
  author={Sarkar, Rik and R{\'o}zemberczki, Benedek},
  booktitle={Workshop on Graph Learning Benchmarks@ TheWebConf 2021},
  year={2021}
}

@article{stability,
  title={Stability Selection},
  author={Meinshausen, Nicolai and B{\"u}hlmann, Peter},
  journal={Journal of the Royal Statistical Society Series B: Statistical Methodology},
  volume={72},
  pages={417--473},
  year={2010}
}

@book{og-procrustes,
    author = {Gower, J. C. and Dijksterhuis, G. B.},
    title = {Introduction: Procrustes Problems},
    publisher = {Oxford University Press},
    year = {2004}
}

@article{sonnet_comm,
author = {Chakrabarty, Sayan and Sengupta, Srijan and Chen, Yuguo},
title = {Subsampling Based Community Detection for Large Networks},
year = {2025},
volume = {35},
journal = {Statistica Sinica},
pages = {1627--1648}
}

@article{deng2023subsamplingbased,
      title={Subsampling-Based Modified {B}ayesian Information Criterion for Large-Scale Stochastic Block Models}, 
      author={Jiayi Deng and Danyang Huang and Xiangyu Chang and Bo Zhang},
      year={2023},
      journal = {arXiv:2304.06900}
}

@inproceedings{qin_reg_spec,
author = {Qin, Tai and Rohe, Karl},
title = {Regularized Spectral Clustering under the Degree-Corrected Stochastic Blockmodel},
year = {2013},
publisher = {Curran Associates Inc.},
address = {Red Hook, NY, USA},
booktitle = {Proceedings of the 26th International Conference on Neural Information Processing Systems - Volume 2},
pages = {3120--3128}
}

@article{Handcock2007,
author = {Handcock, Mark S. and Raftery, Adrian E. and Tantrum, Jeremy M.},
title = {Model-Based Clustering for Social Networks},
journal = {Journal of the Royal Statistical Society: Series A (Statistics in Society)},
volume = {170},
pages = {301--354},
year = {2007}
}

@article{sublinear-time,
    author = {Bhattacharjee, Rajarshi and Dexter, Gregory and Drineas, Petros and Musco, Cameron and Ray, Archan},
    title = {Sublinear Time Eigenvalue Approximation via Random Sampling},
    journal = {Algorithmica},
    volume = {86},
    year = {2024},
    pages = {1764--1829},
}

@article{singularinterlacing_Thompson,
title = {Principal Submatrices IX: Interlacing Inequalities for Singular Values of Submatrices},
journal = {Linear Algebra and its Applications},
volume = {5},
pages = {1--12},
year = {1972},
author = {R.C. Thompson}
}

@article {senguptapabm,
author = {Sengupta, Srijan and Chen, Yuguo},
title = {A Block Model for Node Popularity in Networks with Community Structure},
journal = {Journal of the Royal Statistical Society: Series B (Statistical Methodology)},
 volume={80},
  pages={365--386},
  year={2018},
}

@article{komolafe2022scalable,
  title={Scalable Community Extraction of Text Networks for Automated Grouping in Medical Databases},
  author={Komolafe, Tomilayo and Fong, Allan and Sengupta, Srijan},
  journal={Journal of Data Science},
  volume={21},
  pages={470--489},
  year={2022},
  publisher={School of Statistics, Renmin University of China}
}

@article{dasgupta2022scalable,
  title={Scalable Estimation of Epidemic Thresholds via Node Sampling},
  author={Dasgupta, Anirban and Sengupta, Srijan},
  journal={Sankhya A},
  volume={84},
  pages={321--344},
  year={2022},
  publisher={Springer}
}

@article{chakraborty2025scalable,
  title={Scalable Estimation and Two-Sample Testing for Large Networks via Subsampling},
  author={Chakraborty, Kaustav and Sengupta, Srijan and Chen, Yuguo},
  journal={Journal of Computational and Graphical Statistics},
  volume={34},
  pages={1127--1139},
  year={2025},
  publisher={Taylor \& Francis}
}

@article{sengupta2016subsampled,
  title={A Subsampled Double Bootstrap for Massive Data},
  author={Sengupta, Srijan and Volgushev, Stanislav and Shao, Xiaofeng},
  journal={Journal of the American Statistical Association},
  volume={111},
  pages={1222--1232},
  year={2016}
}

@book{politis1999subsampling,
  title={Subsampling},
  author={Politis, Dimitris N and Romano, Joseph P and Wolf, Michael},
  year={1999},
  publisher={Springer Science \& Business Media}
}

@article{ganguly2023scalable,
  title={Scalable Resampling in Massive Generalized Linear Models via Subsampled Residual Bootstrap},
  author={Indrila Ganguly and Srijan Sengupta and Sujit Ghosh},
  journal={arXiv preprint arXiv:2307.07068},
  year={2023}
}

@article{standard-cv,
  title={Cross-Validation: What Does It Estimate and How Well Does It Do It?},
  author={Bates, Stephen and Hastie, Trevor and Tibshirani, Robert},
  journal={Journal of the American Statistical Association},
  volume={119},
  pages={1434--1445},
  year={2024}
}

@article{breiman1996bagging,
  title={Bagging Predictors},
  author={Breiman, Leo},
  journal={Machine Learning},
  volume={24},
  pages={123--140},
  year={1996}
}

@article{joseph2016impact,
  title={Impact of Regularization on Spectral Clustering},
  author={Joseph, Antony and Yu, Bin},
  year={2016}, 
  journal={The Annals of Statistics},
  volume={44},
  pages={1765--1791}
}

@article{randnet-package,
    title = {randnet: Random Network Model Estimation, Selection and Parameter Tuning},
    author = {Li, Tianxi and Levina, Elizeveta and Zhu, Ji and Le, Can M.},
    year = {2023},
    note= {{R} Package Version 0.7, \url{https://CRAN.R-project.org/package=randnet}},
    journal = {CRAN}
  }

@article{li_correction_2025,
	title = {Correction to: ‘{Network} cross-validation by edge sampling’},
	volume = {112},
	journal = {Biometrika},
	author = {Li, Tianxi and Levina, Elizaveta and Zhu, Ji},
	year = {2025},
	pages = {asaf023},
}

@article{kuhn1955hungarian,
  title={The Hungarian Method for the Assignment Problem},
  author={Kuhn, Harold W},
  journal={Naval Research Logistics Quarterly},
  volume={2},
  pages={83--97},
  year={1955}
}

@article{quinn2017peakRAM,
    title = {peakRAM: Monitor the Total and Peak RAM Used by an Expression or Function},
    author = {Thomas Quinn},
    year = {2017},
    journal = {The Comprehensive R Archive Network}
}

@article{ma2020universal,
  title={Universal Latent Space Model Fitting for Large Networks with Edge Covariates},
  author={Ma, Zhuang and Ma, Zongming and Yuan, Hongsong},
  journal={Journal of Machine Learning Research},
  volume={21},
  pages={1--67},
  year={2020}
}

@article{Soloffbagging2024,
  author  = {Jake A. Soloff and Rina Foygel Barber and Rebecca Willett},
  title   = {Bagging Provides Assumption-free Stability},
  journal = {Journal of Machine Learning Research},
  year    = {2024},
  volume  = {25},
  pages   = {1--35}
}

@article{StoneCV1978,
author = {M. Stone},
title = {Cross-validation: A Review 2 },
journal = {Series Statistics},
volume = {9},
pages = {127--139},
year = {1978},
}

@article{shaolinear1993,
author = {Jun Shao},
title = {Linear Model Selection by Cross-validation},
journal = {Journal of the American Statistical Association},
volume = {88},
pages = {486--494},
year = {1993}
}

@article{zhangmodel1993,
author = {Ping Zhang},
title = {{Model Selection Via Multifold Cross Validation}},
volume = {21},
journal = {The Annals of Statistics},
pages = {299 -- 313},
year = {1993}
}

@article{yangconsistency2007,
author = {Yuhong Yang},
title = {{Consistency of cross validation for comparing regression procedures}},
volume = {35},
journal = {The Annals of Statistics},
pages = {2450 -- 2473},
year = {2007}
}

@article{Leicv2020,
author = {Jing Lei},
title = {Cross-Validation With Confidence},
journal = {Journal of the American Statistical Association},
volume = {115},
pages = {1978--1997},
year = {2020}
}

\section*{Appendix}
Appendices \ref{app:adtl_algo}, \ref{app:proofs}, and \ref{app:numerical:small} provide information on additional definitions and algorithms, proofs of all the theoretical results in this paper, and additional numerical results, respectively.

\begin{appendices}
    \setcounter{equation}{0}
    \setcounter{algorithm}{0}
    \setcounter{lemma}{0}
    \setcounter{theorem}{0}
    \setcounter{figure}{0}

    \def\theequation{S\arabic{equation}}
    \def\thealgorithm{S\arabic{algorithm}}
    \def\thelemma{S\arabic{lemma}}
    \def\thetheorem{S\arabic{theorem}}
    \def\thefigure{S\arabic{figure}}

\section{Relevant Algorithms}\label{app:adtl_algo}
\subsection{Algorithms for Community Detection}
This section describes two community detection algorithms used in the paper to derive the theoretic results and for numerical evaluations of \code{NETCROP}. Both are spectral methods that exploit the eigen structure of the graph adjacency matrix or the graph Laplacian to cluster the nodes in multiple communities. Spectral clustering (\code{SC}) is mainly used on SBM networks and spherical spectral clustering (\code{SSC}) on DCBM networks.

\paragraph*{Spectral Clustering} Spectral clustering is a simple community detection method that recovers the underlying community structures of a network using the eigen decomposition of the corresponding adjacency matrix. For an undirected simple network with adjacency matrix $A$, spectral clustering computes the eigenvectors and eigenvalues of $A$. Then it clusters the $K$ eigenvectors corresponding to the largest $K$ eigenvalues in terms of their absolute values. Any clustering algorithm can be used at this step. However, we consider $K$-means clustering at this step. Since solving a $K$-means clustering problem is NP-hard, we use $(1+\delta)$-approximate solution to the $K$-means problem. Spectral clustering with approximate $K$-means is summarized in Algorithm \ref{algo:spec_clus}.
\begin{algorithm}[tbp]
    \caption{Spectral Clustering with Approximate $K$-means (\code{SC})}
    \label{algo:spec_clus}
    \begin{algorithmic}
      \State \textbf{Input} An adjacency matrix $A_{n \times n}$, number of communities $K$, and an approximating parameter $\delta>0$ for $K$-means clustering. 
    \State \textbf{Output} A membership matrix $\hat C_{n\times K}$.
    \Procedure{SC}{$A, K$}
        \State \textbf 1. Calculate $\hat U \in \R^{n \times K}$ consisting of the leading $K$ eigenvectors (ordered in absolute eigenvalue) of $A$. 
        \State \textbf2. Let $(\hat C, \hat X) \in \C_{n\times K} \times \R^{K \times K}$ such that $\lVert \hat C \hat X - \hat U \rVert_F^2 \leq (1+\delta) \min\limits_{(C, X) \in \C_{n\times K} \times \R^{K \times K}} \lVert CX - \hat U \rVert_F^2$, where $\lVert \cdot \rVert_F$ is the Frobenius norm. 
        \State \textbf3. Output $\hat C$.
    \EndProcedure
    \end{algorithmic}
\end{algorithm}

\paragraph*{Spherical $K$-median Spectral Clustering} Community recovery is difficult for a DCBM due to the presence of degree heterogeneity. Small values in $\psi$ make it hard to identify the community membership of the corresponding nodes as few edges are observed for those nodes. Spherical $K$-median spectral clustering overcomes this issue by row-normalizing the top $K$ eigenvector matrix $\hat U$ and then minimizing the matrix $2,1$ distance between the points and cluster centers. The results on DCBM discussed in \cite{lei2015} were derived for spherical $(1+\delta)$-approximate $K$-median spectral clustering. The algorithm is presented in Algorithm \ref{algo:spher_spec_clus}. 
\begin{algorithm}[tbp]
    \caption{Spherical $K$-median Spectral Clustering (\code{SSC})}
    \label{algo:spher_spec_clus}
    \begin{algorithmic}
      \State \textbf{Input} An adjacency matrix $A_{n \times n}$, the number of communities $K$, and an approximating parameter $\delta > 0$ for $K$-median clustering.
    \State \textbf{Output} A membership matrix $\hat C_{n\times K}$.
    \Procedure{SSC}{$A, K$}
        \State \textbf 1. Calculate $\hat U \in \R^{n \times K}$ consisting of the leading $K$ eigenvectors (ordered in absolute eigenvalue) of $A$. 
        
        \State \textbf 2. Let $I_{+} = \left\{ i :\ \lVert \hat{U}_{i*} \rVert > 0 \right\}$ and $\hat{U}^{+} = (\hat{U}_{I_{+} *})$.
        
        \State \textbf 3. Let $\hat{U}^\prime$ be the row-normalized version of $\hat{U}^{+}$.
        
        \State \textbf 4. Let $(\hat{C}^+, \hat X) \in \C_{n\times K} \times \R^{K \times K}$ such that $\lVert \hat{C}^+ \hat X - \hat{U}^\prime \rVert_{2,1} \leq (1+\delta) \min\limits_{(C, X) \in \C_{n\times K} \times \R^{K \times K}} \lVert CX - \hat{U}^\prime \rVert_{2,1}$.
        
        \State \textbf 5. Output $\hat C$ with $\hat C_{i*}$ being the corresponding row in $\hat{C}^+$ if $i \in I_+$, and $\hat C_{i*} = (1,0, \ldots, 0)$ if $i \notin I_+$.
    \EndProcedure
    \end{algorithmic}
\end{algorithm}

\subsection{Algorithm for RDPG}
Algorithm \ref{algo:ASE} describes adjacency spectral embedding for recovering the latent positions in an RDPG.
\begin{algorithm}[!tb]
    \caption{Adjacency Spectral Embedding (\code{ASE})}
    \label{algo:ASE}
    \begin{algorithmic}
      \State \textbf{Input} An adjacency matrix $A_{n \times n}$, latent position dimension $d$.
    \State \textbf{Output} A latent position matrix $\hat{ X}_{n\times d}$ and the estimated edge probability matrix $\hat{P}_{n \times n}$.
    \Procedure{ASE}{$A, d$}
        \State \textbf 1. Calculate $\hat U \in \R^{n \times d}$ consisting of the leading $d$ eigenvectors of $A$, sorted in descending order of the largest $d$ eigenvalues $\hat{\lambda}_1 \geq \ldots \geq \hat{\lambda}_d$ .
        \State \textbf2. Let $\hat{\Lambda} = diag(\hat{\lambda}_1 \ldots \hat{\lambda}_d)$. Compute $\hat{X} = \hat{U} \hat{\Lambda}^{\frac{1}{2}}$ and $\hat{P} = \hat{X} \hat{X}^\top$.
        \State \textbf3. Output $(\hat{X}, \hat{P})$.
    \EndProcedure
    \end{algorithmic}
\end{algorithm}

\subsection{Algorithms for Community Label Matching} \label{app:match-algo}
Here, we present two algorithms for matching two sets of community labels - \code{MatchBF} (Algorithm \ref{algo:brute-force}) and \code{MatchGreedy} (Algorithm \ref{algo:match-greedy}). \code{MatchBF} searches over all possible permutations of the first set of labels to select the permutation that has the least number of mismatched nodes with the second set of labels. Although this algorithm gives the best permutation of the labels, it has a computation complexity of $O(K!)$ for $K$ communities and is not computationally feasible for large $K$. Algorithm \code{MatchGreedy} computes the number of nodes in each pair of communities between the two sets of labels. It swaps the communities between the two sets of labels that have the highest number of nodes between them. It has a complexity of $O(K^2)$. \code{MatchGreedy} is a faster and simpler adaptation of the Hungarian algorithm \citep{kuhn1955hungarian} for label matching, and is shown to produce the same permutation matrix as \code{MatchBF} when one of the community membership matrices can be expressed as a permutation of the other plus some error term that satisfies certain conditions \citep{mukherjee2017provably}.

\begin{algorithm}[tbp]
    \caption{\code{MatchBF}}
    \label{algo:brute-force}
    \begin{algorithmic}
        \State \textbf{Input} Two community membership matrices $C_1$ and $C_2$ with the same number of communities 
        \State \textbf{Output} A permutation matrix $P$ that best aligns $C_1$ with $C_2$
    
        \Procedure{\code{MatchBF}}{$C_1, C_2$}
            \State $K \gets$ \text{number of columns of }$C_1$ 
            \State $\E_K \gets$ list of all permutation matrices of order $K \times K$
            \State \textbf{initialize} a vector $mismatch$ of length $K!$
            \State $i \gets 1$
            \For{$E \in \E_K$}
                \State $mismatch[i] \gets$ $\lVert(C_1E - C_2) \rVert_0$
                \State $i \gets i + 1$
            \EndFor
            \State \textbf{return} $\E_K[\argmin\limits_i  mismatch[i]]$
        \EndProcedure
    \end{algorithmic}
\end{algorithm}

\begin{algorithm}[tbp]
    \caption{\code{MatchGreedy}}
    \label{algo:match-greedy}
    \begin{algorithmic}
    
    \State \textbf{Input} Two community membership matrices $C_1$ and $C_2$ with the same number of communities 
        \State \textbf{Output} A permutation matrix $P$ that approximates the best alignment of $C_1$ with $C_2$
        
        \Procedure{\code{MatchGreedy}}{$C_1, C_2$} 
                \State $K \gets$ \text{number of columns of } $C_1$
                \State $P \gets \mathbf{0}_{K\times K}$ (null matrix of order $K \times K$)
                \State $M \gets$ $C_1^T C_2$
                
                \While{there are rows or columns of $M$ left with positive values}
                    \State \textbf{find} $(i,j) = \argmax\limits_{i,j} M_{ij}$ (ties are broken arbitrarily) 
                    \State $P_{ij} \gets 1$
                    \State \textbf{replace} the $i$-th row and the $j$-th column of $M$ by $-1$
                \EndWhile
            \State \textbf{return} $P$
        \EndProcedure
    \end{algorithmic}
\end{algorithm}

\subsection{\code{NETCROP} for Estimating the Dimension of Latent Space Model}\label{supp:est_latent}
The latent space model \citep{Hoff2002LatentSA, ma2020universal} is a general class of models that assumes that the connection probabilities between the nodes are influenced by the positions of the nodes in an unobserved latent space. Let $z_i \in \R^d$ be the latent position of the $i$th node. Then in a distance based latent space model, the probability of an edge between nodes $i$ and $j$ is modelled as
\begin{align}
    \eta_{ij} \coloneqq \logodds (A_{ij} = 1 | z_i, z_j, \alpha) = \log \left(\frac{P_{ij}}{1 - P_{ij}}\right) = \alpha -  \lVert z_i - z_j\rVert^2.\label{eqn:latent.mod}
\end{align}

\citet{ma2020universal} provided a projected gradient descent algorithm that finds the maximum likelihood estimator of $\alpha$ and $Z$. However, their algorithm requires the knowledge of the true latent space dimension $d$ and requires a few tuning parameters including the step size and the initializer for the gradient descent algorithm. Here, we focus on applying \code{NETCROP} for selecting the latent space dimension $d$ from a set of candidate dimensions $[d_{max}]$. The method is similar to \code{NETCROP} for RDPG (Algorithm \ref{algo:RDPG}) except for the estimation and stitching steps. Given the subsets $S_{01}, \ldots, S_{0s}$ with overlap nodes in $S_0$, the method of \citet{ma2020universal} is applied on subnetwork $A_{S_{0q}}$ to obtain $\hat{Z}^{(\d)}_{q}$ and $\hat{\alpha}^{(\d)}_{q}$ for each $\d \in [d_{max}]$ and $q \in [s]$. The latent positions are only identifiable up to rotation, reflection, and translation. We use generalized Procrustes transformation to match the estimated latent positions of the overlap nodes from the subnetworks with respect to an arbitrarily selected standard subnetwork. Then for each candidate dimension $\d \in [d_{max}]$ and node pair $(i,j) \in \S^c$ such that $i \in S_p$ and $j \in S_q$, $1 \leq p < q \leq s$, the intercept term is estimated as the average $\hat{\alpha}^{(\d)}_{pq} = (\hat{\alpha}^{(\d)}_{p} + \hat{\alpha}^{(\d)}_{q})/2 $ and the edge probability is predicted as
\begin{align}
    \hat{\eta}_{ij}^{(\d)} = \hat{\alpha}^{(\d)}_{pq} - \lVert \hat{z}^{(\d)}_{p, i} - \hat{z}^{(\d)}_{q, j} \rVert^2 \ \mbox{and} \ \hat{P}_{ij}^{(\d)} = \frac{\exp{\left(\hat{\eta}_{ij}^{(\d)}\right)}}{1 + \exp{\left(\hat{\eta}_{ij}^{(\d)}\right)}}.
\end{align}

The loss functions are computed as the sum of losses between the predicted edge probabilities and the observed adjacency matrix over the entries in the test set for each candidate $\d \in [d_{max}]$. Then, the outcome $\hat{d}$ is chosen as the minimizer of the computed losses. The procedure for selecting the latent space dimension $d$ may be repeated $R$ times and then a majority voting can be used to increase the estimation accuracy as in Algorithm \ref{algo:RDPG}.
    \section{Proofs of Theoretical Results}\label{app:proofs}
This section presents the proofs of the theorems and results in the paper. Throughout all the proofs, it is assumed that \code{NETCROP} is applied on a network of size $n$ with $s$ subgraphs and overlap size $o$. The size of each non-overlap part is $m = (n-o)/s$ with the size of each working subgraph being $(o+m)$. The overlap part is denoted by $S_0$ and the $s$ non-overlap parts by $S_1, \ldots, S_s$. The subgraphs are spanned by the subsets of nodes $S_{0q} = S_0 \cup S_q$, for $q \in [s]$. The training and test sets of \code{NETCROP} are 
\begin{align*}
\text{Training set } \ &= \S \coloneqq \bigcup\limits_{q \in [s]} (S_{0q} \times S_{0q}),\\
    \text{Test set } \ &= \S^c = (S \times S) \setminus \S = \mathop{\mathop{\bigcup}\limits_{p < q}}\limits_{p,q \in [s]} (S_p \times S_q).
\end{align*}
Note that $\lvert \S^c \rvert = \binom{s}{2}m^2$. Additionally, we redefine the notations for the computed loss and the oracle loss as
\begin{align*}
    L(A, \hat{P}) \coloneqq L(A_{\S^c}, \hat{P}_{\S^c}) \ \ \mbox{ and } \ \ L(A) \coloneqq L(A_{\S^c}, P_{\S^c}).
\end{align*}

We state an upper bound of Bernoulli sums that follows directly from Bernstein's inequality. This bound is used throughout the appendix.
\begin{lemma}\label{lemma:bern}
    Let $M_i \sim Bernoulli(p_i)$ independently for $i \in [n]$ and $\mu_n = \frac{1}{n}\sum\limits_{i \in [n]} p_i$. Then the following bounds hold for sufficiently large $n$:
    \begin{align}
        P\left( \sum\limits_{i \in [n]} M_i  - n\mu_n \geq \delta n\mu_n \right) &\leq \exp{\left(- \frac{\delta^2 n\mu_n/2}{\delta/3 + 1} \right)} \ \mbox{for any }\ \delta > 0, \\
        &= \exp{\left( - \Omega\left(\delta n\mu_n\right) \right)}, \ \mbox{for }\ \delta \coloneqq \delta_n = \Omega(1).
    \end{align}

    Additionally, 
    \begin{align}
        \left\lvert \frac{1}{n}\sum_{i \in [n]} (M_i - p_i) \right\rvert = O_p\left( \sqrt{\frac{\mu_n}{n}} \right) \ \mbox{ as }\ n \to \infty.
    \end{align}

\end{lemma}

\begin{proof}
    From Bernstein's inequality,
    \begin{align*}
        P\left( \sum\limits_{i \in [n]} M_i  - n\mu_n \geq \delta n\mu_n \right) &\leq \exp{\left(- \frac{\delta^2 n^2\mu_n^2/2}{\delta n \mu_n/3 + \sum\limits_{i \in [n]} p_i (1-p_i)} \right)} \\
        &\leq \exp{\left(- \frac{\delta^2 n^2\mu_n^2/2}{\delta n \mu_n/3 + n\mu_n} \right)}\\
        &= \exp{\left(- \frac{\delta^2 n\mu_n/2}{\delta/3 + 1} \right)}.
    \end{align*}

    Based on the definition of $O_p$ in \citet[Section~1.2.5]{Serfling}, for sequences of random variables $X_n$ and positive real numbers $a_n$, $X_n = O_p(a_n)$ if for every $\epsilon > 0$, there exist $M(\epsilon)$ and  $N(\epsilon)$ such that $P(\lvert X_n \rvert \geq M(\epsilon) a_n ) \leq \epsilon$ for any $n \geq N(\epsilon)$.
    
    For any $\epsilon > 0$, define $M(\epsilon) = \epsilon^{-\frac{1}{2}}$ and $N(\epsilon) = 1$. Then by Chebyshev's inequality and for $n = 1, 2, \ldots$, 
    \begin{align}
        P\left( \left\lvert \frac{1}{n}\sum_{i \in [n]} (M_i - p_i) \right\rvert \geq M(\epsilon) \sqrt{Var\left(\frac{1}{n}\sum_{i \in [n]} (M_i - p_i) \right)} \right) \leq \epsilon.
    \end{align}

    Thus from the earlier definition, 
    \begin{align*}
     \left\lvert \frac{1}{n}\sum_{i \in [n]} (M_i - p_i) \right\rvert = O_p\left(\sqrt{Var\left(\frac{1}{n}\sum_{i \in [n]} (M_i - p_i)\right)}\right) = O_p\left(\sqrt{ \sum\limits_{i \in [n]} \frac{p_i(1-p_i)}{n^2}}\right) = O_p\left(\sqrt{\frac{\mu_n}{n}}\right).   
    \end{align*}
\end{proof}

\subsection{Proof of Theorem \ref{theorem:sbm_lossdiff}}\label{app:proofs:SBM}
In this section, we prove the bounds on the difference between computed loss and the oracle loss for SBM. Case 1 deals with the scenario $\K < K$ and Case 2 treats the $\K = K$ case.
We reintroduce and redefine some of the notations for the proof.
Recall that $g_i$ is the true community label of the $i$-th node and $G_{k}$ consists of the nodes in true community $k$, $i \in [n]$ and $k \in [K]$. For a candidate value $\K \in [K_{max}]$, the predicted community label of the $i$-th node from \code{NETCROP} with $\K$ communities is $\hat{g}_{ i}^{(\K)}$ and the set of nodes in the predicted $k$-th community is $\hat{G}_{k}^{(\K)}$, $i \in [n]$, and $k \in [\K]$. The superscript is dropped for $\hat{P}, \hat{g}$ and $\hat{G}$ whenever it is clear from the context. We also define the following quantities - 
\begin{align}
    \nonumber & T_{k, k^\prime, l, l^\prime} = \left\{  (i,j) \in \S^c : i \neq j \text{ and } i \in \hat{G}_k \cap G_l, j \in \hat{G}_{k^\prime} \cap G_{l^\prime} \right\},\\ \nonumber
    &T_{k, k^\prime, l, l^\prime}^{(pq)} = \left\{  (i,j) \in S_p \times S_q : i \neq j \text{ and } i \in \hat{G}_k \cap G_l, j \in \hat{G}_{k^\prime} \cap G_{l^\prime} \right\},\\ \nonumber
    & T_{l l^\prime} = \left\{  (i,j) \in \S^c : i \neq j \text{ and } i \in G_l, j \in G_{l^\prime} \right\},\\ \nonumber
    & \hat{T}_{kk^\prime} = \left\{  (i,j) \in \S^c : i \neq j \text{ and } i \in \hat{G}_k, j \in \hat{G}_{k^\prime} \right\},\\ \nonumber
    & \U_{l l^\prime} = \left\{  (i,j) \in \S : i \neq j \text{ and } i \in G_l, j \in G_{l^\prime} \right\},\\ \nonumber
    & \U_{l l^\prime}^{(q)} = \left\{  (i,j) \in S_{0q} \times S_{0q} : i \neq j \text{ and } i \in G_l, j \in G_{l^\prime} \right\},\\ \nonumber
    & \hat{\U}_{kk^\prime} = \left\{  (i,j) \in \S : i \neq j \text{ and } i \in \hat{G}_k, j \in \hat{G}_{k^\prime} \right\},\\ 
    &\hat{\U}_{k k^\prime}^{(q)} = \left\{\ \ (i,j) \in S_{0q} \times S_{0q} : i \neq j \ \text{and}\ i \in \hat{G}_k, j \in \hat{G}_{k^\prime}  \right\}. \label{notation:test_subsets}
\end{align}
We restate a few results from \citet{sonnet_comm} that are used in the proofs of Theorem \ref{theorem:sbm_lossdiff}. 
The first lemma bounds the sizes of the smallest and the largest communities in a subnetwork from a blockmodel. The next theorem bounds the error rate of spectral clustering on a subnetwork from SBM. All the results have been modified based on Assumption \ref{assump:SBM} and switched to the notations used in this paper.

\begin{lemma}{\citep[Lemma~S3]{sonnet_comm}}\label{lemma:maxmincomm}
    Let $\Tilde{n}_{min}$ and $\Tilde{n}_{max}$ be the sizes of the smallest and the largest communities in a random subnetwork of size $\eta$ from an SBM or a DCBM with $n$ nodes and $K$ communities. Let $\{n_1, \ldots, n_K\}$ be the expected sizes of the communities in the entire network satisfying $\min\limits_{k \in [K]} n_k > \gamma n$ for some $\gamma > 0$. Then the following bounds hold 
    \begin{align}
        \Tilde{n}_{min} \geq \frac{\eta \gamma}{1+\gamma}\ \mbox{ and }\ \Tilde{n}_{max} \leq \frac{\eta (1-\gamma)}{\gamma},
    \end{align}
    each with probability tending to $1$ as $\eta/n > c$ for some $c > 0$ and $n \to \infty$.
\end{lemma}

\begin{theorem}{\citep[Theorem~4]{sonnet_comm}}\label{SBM_bound}
    Under Assumption \ref{assump:SBM}, the number of misclustered nodes from \code{SC} (Algorithm \ref{algo:spec_clus}) on $A_{S_{0q}}$, where $A$ is from SBM with $n$ nodes and $K$ communities, is $O_p\left(\rho_n^{-1}\right)$ for $q \in [s]$, as $o/n > c$ for some $c>0$ and $n \to \infty$.
\end{theorem}

Next we state a lemma bounding sum of $A_{ij}$'s over node pairs in different communities.
\begin{lemma}\label{lemma:sum-aij}
    Let $\S_{kk^\prime} \subset G_{k} \times G_{k^\prime}$ be an arbitrary subset of node pairs from a blockmodel with $n$ nodes and $K$ communities with $G_k$ being the set of all nodes in the $k$-th community. If $\lvert \S_{kk^\prime} \rvert \geq c n^2$ for some $c > 0$ and $n \rho_n = \omega(1)$ as $n \to \infty$, then
    \begin{align}
        &\frac{1}{\lvert \S_{k k^\prime}\rvert} \sum\limits_{(i,j) \in \S_{k k^\prime}} (A_{ij} - B_{kk^\prime}) = O_p\left(\sqrt{\frac{\rho_n}{\lvert S_{kk^\prime} \rvert}}\right), \label{eq:sqrho}\\
        \mbox{and }\ \ &\frac{1}{\lvert \S_{k k^\prime}\rvert} \sum\limits_{(i,j) \in \S_{k k^\prime}} A_{ij} = O_p(\rho_n).\label{eq:halfrho}
    \end{align}
\end{lemma}
\begin{proof}
Note that $E(A_{ij}) = B_{kk^\prime}$ for all $(i,j) \in \S_{kk^\prime}$ and $\lvert \S_{kk^\prime} \rvert B_{kk^\prime} = \Omega(n^2 \rho_n) = \omega(n)$. Then from Lemma \ref{lemma:bern}, as $n \to \infty$, for any $k,k^\prime \in [K]$, we have 
\begin{align}
    \left\lvert\frac{1}{\lvert \S_{k k^\prime}\rvert} \sum\limits_{(i,j) \in \lvert \S_{k k^\prime}\rvert} (A_{ij} - B_{kk^\prime})\right\rvert = O_p\left(\sqrt{\frac{B_{kk^\prime}}{\lvert \S_{k k^\prime}\rvert}}\right).
\end{align}

    For any $\delta > 0$, we have
\begin{align}
    P\left(\sum\limits_{(i,j) \in \S_{k k^\prime}} (A_{ij} - B_{k k^\prime}) \geq \delta \lvert\S_{k k^\prime}\rvert B_{k k^\prime} \right)\nonumber
    \leq \exp{\left(- \frac{\delta^2 \lvert \S_{kk^\prime} \rvert B_{kk^\prime}/2}{\delta/3 + 1}\right)}.
\end{align}
Setting $\delta = 1/2$, we get
\begin{align}
    P\left( \frac{1}{\lvert\S_{k k^\prime}\rvert}\sum\limits_{(i,j) \in \S_{k k^\prime}} A_{ij} \geq \frac{3}{2}B_{kk^\prime}\right) \leq \exp{\left( - \Omega\left( \lvert \S_{k k^\prime}\rvert B_{kk^\prime}/ 2 \right)\right)}. \label{eq:halfrho_back}
\end{align}
As $B_{k k^\prime} = \rho_n B_{kk^\prime}^0 \asymp \rho_n$, we get the bounds in the lemma as $n \to \infty$.
\end{proof}

The following lemma uses a Hoeffding bound and the pigeonhole principle to show that there exists a set of node pairs between two predicted communities 
from \code{NETCROP} that has a large intersection with the set of node pairs between three true communities. These sets of node pairs will be later used to derive a lower bound for the distance between the computed loss of \code{NETCROP} and the oracle loss.

\begin{lemma}\label{lemma:sbm-dcbm_php}
Consider a blockmodel with $n$ nodes, $K$ communities of sizes $n_1, \ldots, n_K > \gamma n$ for some $\gamma > 0$, community connectivity matrix $B = \rho_n B_0$ with all distinct rows.
Then for each pair of distinct non-overlap parts $S_p$ and $S_q$, $1 \leq p < q \leq s$, there exist $k_1, k_2 \in [\K]$, $l_1, l_2, l_3 \in [K]$, depending on $p, q$, and a constant $\Tilde{c} > 0$ such that the following inequalities hold for $\K < K$ with high probability - 
\begin{align}
\nonumber &1.\ \ \lvert T_{k_1, k_2, l_1, l_3}^{(pq)} \rvert \geq \Tilde{c} m^2\\ \nonumber
&2.\ \ \lvert T_{k_1, k_2, l_2, l_3}^{(pq)} \rvert \geq  \Tilde{c} m^2\\ \nonumber
&3.\ \ B_0(l_1, l_3) \neq B_0(l_2, l_3). 
\end{align}
\end{lemma}

\begin{proof}
    Recall that $\lvert G_{l} \rvert = n_l > \gamma n$ and $\lvert S_p \rvert = m$. For any $p \in [s]$, note that $\lvert G_l \cap S_p \rvert \sim Hyp(n, n_l, m)$. Then by hypergeometric tail bound,
    \begin{align}
    \nonumber P\left( \lvert G_l \cap S_p\rvert \geq \frac{\gamma m}{2} \right) &= P\left( \lvert G_l \cap S_p \rvert \geq \frac{m}{2n} \gamma n \right)\\ \nonumber
    &\geq P\left( \lvert G_l \cap S_p \rvert \geq \frac{m n_{l}}{2n} \right) \\ \nonumber
    &= P\left( \lvert G_l \cap S_p \rvert \geq \frac{1}{2} E\left(\lvert G_l \cap S_p \rvert\right) \right) \\ \nonumber
    &\geq 1 - \exp{\left( - \frac{2 E\left(\lvert H_l \cap S_p \rvert\right)^2}{2^2 m}  \right)}\\ \nonumber
    &= 1 - \exp{\left( -\frac{m n_{l}^2}{2n^2} \right)}\\
    &\geq 1 - \exp{\left( -\frac{\gamma^2 m}{2} \right)}.\label{eqn:hyper-tail}
    \end{align}
    Thus, as $m \to \infty$, $\lvert G_l \cap S_p \rvert \geq \frac{\gamma m}{2}$ with probability going to $1$. Note that this bound is uniform over all communities $l \in [K]$ as $K$ is fixed.

    For $\Tilde{K} < K$ and any $p \in [s]$, by the pigeonhole principle, there exist $1 \leq l_1 \neq l_2 \leq K$ and $k_1 \in [\Tilde{K}]$ such that the following holds with high probability using the bound in (\ref{eqn:hyper-tail})
    \begin{align}
        \lvert \hat{G}_{k_1} \cap G_{l_1} \cap S_p \rvert &\geq \frac{\lvert G_{l_1} \cap S_p \rvert}{\Tilde{K}} > \frac{\gamma m}{2K}, \nonumber\\
        \mbox{and} \ \ \lvert \hat{G}_{k_1} \cap G_{l_2} \cap S_p \rvert &\geq \frac{\lvert G_{l_2} \cap S_p \rvert}{\Tilde{K}} > \frac{\gamma m}{2K}.
    \end{align}

    Since the rows of $B_0$ are distinct, there exists $l_3 \in [K]$ such that $B_{0, l_1 l_3} \neq B_{0, l_2 l_3}$. Then there exists $k_2$ such that for any $q > p$ and $q \in [s]$
    \begin{align}
        \lvert \hat{G}_{k_2} \cap G_{l_3} \cap S_q \rvert &\geq \frac{\lvert G_{l_3} \cap S_q \rvert}{\Tilde{K}} > \frac{\gamma m}{2K}.
    \end{align}

    Construct $T_{k_1, k_2, l_1, l_3}^{(pq)} \coloneqq (\hat{G}_{k_1} \cap G_{l_1} \cap S_p) \times (\hat{G}_{k_2} \cap G_{l_3} \cap S_q)$ and $T_{k_1, k_2, l_2, l_3}^{(pq)} \coloneqq (\hat{G}_{k_1} \cap G_{l_2} \cap S_p) \times (\hat{G}_{k_2} \cap G_{l_3} \cap S_q)$. Then
    \begin{align}
        \lvert T_{k_1, k_2, l_1, l_3}^{(pq)} \rvert &> \frac{\gamma^2 m^2}{4K^2}, \nonumber\\
        \mbox{and }\ \ \lvert T_{k_1, k_2, l_2, l_3}^{(pq)} \rvert &> \frac{\gamma^2 m^2}{4K^2}.
    \end{align}

    Taking $\Tilde{c} = \gamma^2/4K^2$, the lemma follows.
    
\end{proof}

\paragraph*{Proof of Theorem \ref{theorem:sbm_lossdiff} for SBM:}
Case 1: ($\K < K$)
Without loss of generality assume $k_1 = 1, k_2 = 2, l_1 = 1, l_2 = 2, l_3 = 3$ in Lemma \ref{lemma:sbm-dcbm_php} for any $1 \leq p < q \leq s$. Also, define the following quantities:
\begin{align}
    & \hat{p}_{k, k^\prime, l, l^\prime}^{(pq)} = \frac{1}{\lvert T_{k, k^\prime, l, l^\prime}^{(pq)} \rvert} \sum\limits_{(i,j) \in T_{k, k^\prime, l, l^\prime}^{(pq)}} A_{ij}\nonumber \\
    & \lambda = \frac{\lvert T_{1,2,1,3}^{(pq)} \rvert}{\lvert T_{1,2,1,3}^{(pq)} \rvert + \lvert T_{1,2,2,3}^{(pq)} \rvert}, \nonumber\\
    & \hat{p} = \frac{1}{\lvert T_{1,2,1,3}^{(pq)} \rvert + \lvert T_{1,2,2,3}^{(pq)} \rvert} \sum\limits_{(i,j) \in T_{1,2,1,3}^{(pq)} \cup T_{1,2,2,3}^{(pq)}} A_{ij}
    = \lambda \hat{p}_{1,2,1,3}^{(pq)} + (1-\lambda) \hat{p}_{1,2,2,3}^{(pq)}. \nonumber
\end{align}

Consider the difference between the loss function and the oracle loss on $S_p \times S_q$ for any $1 \leq p < q \leq s$:

\begin{align}
    & L^{(pq)} (A, \hat{P}^{(\K)}) - L^{(pq)}(A) \nonumber\\
    =& \sum\limits_{(i,j) \in S_p \times S_q} \left[ (A_{ij} - \hat{P}_{ij}^{(\K)})^2 - (A_{ij} - P_{ij})^2 \right] \nonumber\\
    =& \sum\limits_{(k,k^\prime, l, l^\prime) \in [\Tilde{K}]^2 \times [K]^2} \sum\limits_{(i,j) \in T_{k,k^\prime, l, l^\prime}^{(pq)}} \left[ (A_{ij} - \hat{B}_{kk^\prime}^{(\K)})^2 - (A_{ij} - B_{ll^\prime})^2 \right] \nonumber\\
    =& \sum\limits_{(i,j) \in T_{1,2,1,3}^{(pq)}} \left[ (A_{ij} - \hat{B}_{12}^{(\K)})^2 - (A_{ij} - B_{13})^2 \right] \nonumber\\
    &+ \sum\limits_{(i,j) \in T_{1,2,2,3}^{(pq)}} \left[ (A_{ij} - \hat{B}_{12}^{(\K)})^2 - (A_{ij} - B_{23})^2 \right] \nonumber\\
    &+ \sum\limits_{(k,k^\prime, l, l^\prime) \notin \{ (1,2,1,3), (1,2,2,3) \} } \sum\limits_{(i,j) \in T_{k,k^\prime, l, l^\prime}^{(pq)}} \left[ (A_{ij} - \hat{B}_{kk^\prime}^{(\K)})^2 - (A_{ij} - B_{ll^\prime})^2 \right] \nonumber\\
    \geq& \sum\limits_{(i,j) \in T_{1,2,1,3}^{(pq)}} \left[ (A_{ij} - \hat{p})^2 - (A_{ij} - B_{13})^2 \right] \nonumber\\
    &+ \sum\limits_{(i,j) \in T_{1,2,2,3}^{(pq)}} \left[ (A_{ij} - \hat{p})^2 - (A_{ij} - B_{23})^2 \right] \nonumber\\
    &+ \sum\limits_{(k,k^\prime, l, l^\prime) \notin \{ (1,2,1,3), (1,2,2,3) \} } \sum\limits_{(i,j) \in T_{k,k^\prime, l, l^\prime}^{(pq)}} \left[ (A_{ij} - \hat{p}_{k, k^\prime, l, l^\prime}^{(pq)})^2 - (A_{ij} - B_{ll^\prime})^2 \right] \nonumber\\
    =& I + II + III. 
\end{align}

Then,

\begin{align}\label{term-I}
    I &= \sum\limits_{(i,j) \in T_{1,2,1,3}^{(pq)}} \left[ (A_{ij} - \hat{p})^2 - (A_{ij} - B_{13})^2 \right]\nonumber \\ \nonumber
    &= \sum\limits_{(i,j) \in T_{1,2,1,3}^{(pq)}} \left[ A_{ij}^2 - 2A_{ij}\hat{p} + \hat{p}^2 - A_{ij}^2  + 2A_{ij}B_{13} - B_{13}^2  \right]\\ \nonumber
    &= -2\hat{p} \sum\limits_{(i,j) \in T_{1,2,1,3}^{(pq)}}A_{ij} + \lvert T_{1,2,1,3} \rvert \hat{p}^2 + 2B_{13} \sum\limits_{(i,j) \in T_{1,2,1,3}^{(pq)}} A_{ij} - \lvert T_{1,2,1,3}^{(pq)} \rvert B_{13}^2\\ \nonumber
    &= \lvert T_{1,2,1,3}^{(pq)} \rvert \left( \hat{p}^2 - 2 \hat{p}\hat{p}_{1,2,1,3}^{(pq)} + (\hat{p}_{1,2,1,3}^{(pq)})^2 - (\hat{p}_{1,2,1,3}^{(pq)})^2 + 2\hat{p}_{1,2,1,3}^{(pq)} B_{13} - B_{13}^2 \right)\\ \nonumber
    &= \lvert T_{1,2,1,3}^{(pq)} \rvert \left[ (\hat{p} - \hat{p}_{1,2,1,3}^{(pq)})^2 - (\hat{p}_{1,2,1,3}^{(pq)} - B_{13})^2 \right]\\ \nonumber
    &= \lvert T_{1,2,1,3}^{(pq)} \rvert \left[ (1-\lambda)^2(\hat{p}_{1,2,1,3}^{(pq)} - \hat{p}_{1,2,2,3}^{(pq)})^2 - (\hat{p}_{1,2,1,3}^{(pq)} - B_{13})^2 \right]\\ \nonumber
    &\geq \lvert T_{1,2,1,3}^{(pq)} \rvert \Bigg[ \frac{(1-\lambda)^2}{2} (B_{13} - B_{23})^2 - (1-\lambda)^2 ((\hat{p}_{1,2,1,3}^{(pq)} - B_{13}) - (\hat{p}_{1,2,2,3}^{(pq)} - B_{23}))^2\\ \nonumber
    &\hspace{10cm} - (\hat{p}_{1,2,1,3}^{(pq)} - B_{13})^2 \Bigg]\\ \nonumber
    &\hspace{8cm} \left(\text{as } (a+b)^2 \geq \frac{b^2}{2} - a^2 \right)\\ \nonumber
    &\geq \lvert T_{1,2,1,3}^{(pq)} \rvert \Bigg[ \frac{(1-\lambda)^2}{2} (B_{13} - B_{23})^2 - 2(1-\lambda)^2 (\hat{p}_{1,2,1,3}^{(pq)} - B_{13})^2\\ \nonumber
    &\hspace{5cm}  - 2(1-\lambda)^2 (\hat{p}_{1,2,2,3}^{(pq)} - B_{23})^2 - (\hat{p}_{1,2,1,3}^{(pq)} - B_{13})^2 \Bigg]\\
    &\hspace{8cm} \left( \text{ as } (a-b)^2 \leq 2(a^2 + b^2) \right).
\end{align}

From Lemma \ref{lemma:sbm-dcbm_php}, 
\begin{align*}
    \Tilde{c} m^2 \leq \lvert T_{1,2,1,3}^{(pq)} \rvert \leq \lvert S_p \times S_q \rvert = m^2 \ \mbox{and} \ \Tilde{c} m^2 \leq \lvert T_{1,2,2,3}^{(pq)} \rvert \leq \lvert S_p \times S_q \rvert = m^2.
\end{align*}

Then 
\begin{align*}
    \frac{(1-\lambda)^2}{2} (B_{13} - B_{23})^2 = \frac{\lvert T_{1,2,2,3}^{(pq)} \rvert^2 }{2(\lvert T_{1,2,2,3}^{(pq)} \rvert + \lvert T_{1,2,2,3}^{(pq)} \rvert)^2} (B_{13} - B_{23})^2 \geq c_1 \rho_n^2.
\end{align*}

For $(k, k^\prime, l, l^\prime) \in \{(1,2,1,3), (1,2,2,3)\}$, using Lemma \ref{lemma:bern} we have
\begin{align*}
    \lvert \hat{p}_{k,k^\prime,l,l^\prime}^{(pq)} - B_{ll^\prime} \rvert = O_p\left( \sqrt{\frac{B_{ll^\prime}}{\lvert T_{k,k^\prime,l,l^\prime}^{(pq)} \rvert}} \right) = O_p\left( \sqrt{\frac{\rho_n}{\lvert T_{k,k^\prime,l,l^\prime}^{(pq)} \rvert}} \right).
\end{align*}
From Equation \ref{term-I}, we have
\begin{align}
    I \geq c_1 m^2\rho_n^2 - O_p(\rho_n).
\end{align}
Similarly, 
\begin{align}
    II \geq c_1 m^2\rho_n^2 - O_p(\rho_n).
\end{align}

Let $\hat{p}^\prime$ be the average of $A_{ij}$ over $(i,j) \in T_{k, k^\prime, l, l^\prime}^{(pq)}$ for all $(k, k^\prime, l, l^\prime) \notin \{(1,2,1,3), (1,2,2,3)\}$. Then,
\begin{align}
    III &= \sum\limits_{(k,k^\prime, l, l^\prime) \notin \{ (1,2,1,3), (1,2,2,3) \} } \sum\limits_{(i,j) \in T_{k,k^\prime, l, l^\prime}^{(pq)}} \left[ (A_{ij} - \hat{p}_{k, k^\prime, l, l^\prime})^2 - (A_{ij} - B_{ll^\prime})^2 \right]\nonumber\\
    &=  \sum\limits_{(k,k^\prime, l, l^\prime) \notin \{ (1,2,1,3), (1,2,2,3) \} } \sum\limits_{(i,j) \in T_{k,k^\prime, l, l^\prime}^{(pq)}} \left[ A_{ij}^2 - 2A_{ij}\hat{p}_{k,k^\prime,l,l^\prime} + \hat{p}_{k,k^\prime,l,l^\prime}^2 - A_{ij}^2 + 2A_{ij} B_{ll^\prime} - B_{ll^\prime}^2 \right] \nonumber\\
    &=  \sum\limits_{(k,k^\prime, l, l^\prime) \notin \{ (1,2,1,3), (1,2,2,3) \} } \left\lvert T_{k,k^\prime,l,l^\prime}^{(pq)} \right\rvert \left[ -2 \hat{p}_{k,k^\prime,l,l^\prime}^2 + \hat{p}_{k,k^\prime,l,l^\prime}^2 + 2 B_{ll^\prime} \hat{p}_{k,k^\prime,l,l^\prime} - B_{ll^\prime}^2  \right] \nonumber\\
    &= - \sum\limits_{(k,k^\prime, l, l^\prime) \notin \{ (1,2,1,3), (1,2,2,3) \} } \left\lvert T_{k,k^\prime, l, l^\prime}^{(pq)} \right\rvert  (\hat{p}_{k, k^\prime, l, l^\prime} - B_{ll^\prime})^2\nonumber\\
    &= - O_p(\rho_n),
\end{align}
where the last inequality follows from Lemma \ref{lemma:bern} since the sum can be split into $(K-1)^2$ many sums, each taken over $n_l^{(p)} n_{l^\prime}^{(q)} \geq \gamma^2 (o+m)^2 \geq c^2 \gamma^2 n^2$ many terms for each $(k, k^\prime, l, l^\prime) \notin \{(1,2,1,3), (1,2,2,3)\}$. Combining
\begin{align}
    L^{(pq)}(A, \hat{P}) - L^{(pq)}(A) \geq I + II + III \geq cm^2\rho_n^2 - O_p(\rho_n).
\end{align}

Now consider the difference of computed loss and the oracle loss over the entire test set $\S^c = \bigcup\limits_{1 \leq p < q \leq s} (S_p \times S_q)$ as
\begin{align}
    L(A, \hat{P}^{(\K)}) - L(A) &= \sum_{1 \leq p < q \leq s} \left[ L(A, \hat{P}^{(\K)})^{(pq)} - L^{(pq)}(A) \right] \nonumber\\
    &\geq c s(s-1)m^2\rho_n^2 - O_p(s(s-1)\rho_n).
\end{align}

\paragraph*{Case 2: ($\K = K$ for SBM)}
Define $\epsilon_n = \sup_{1\leq l, l^\prime \leq K} \lvert \hat{B}_{ll^\prime} - B_{ll^\prime} \rvert$. The absolute difference between computed loss and oracle loss is
\begin{align}
    \lvert L(A, \hat{P}) - L(A) \rvert &\leq  \sum\limits_{l, l^\prime \in [K]} \sum\limits_{(i,j) \in T_{ll^\prime}} \left\lvert (A_{ij} - \hat{B}_{ll^\prime})^2 - (A_{ij} - B_{ll^\prime})^2 \right\rvert  \nonumber\\ 
    & \hspace{2cm} + \mathop{\sum\limits_{(k, k^\prime)\neq(l, l^\prime)}}_{\in [K]} \sum\limits_{(i,j) \in T_{k,k^\prime, l,l^\prime}} \left\lvert (A_{ij} - \hat{B}_{kk^\prime})^2 - (A_{ij} - B_{ll^\prime})^2 \right\rvert \nonumber\\
    &= \sum\limits_{l, l^\prime \in [K]} \sum\limits_{(i,j) \in T_{ll^\prime}} \left\lvert -2A_{ij}(\hat{B}_{ll^\prime} - B_{ll^\prime}) + (\hat{B}_{ll^\prime} + B_{ll^\prime})(\hat{B}_{ll^\prime} - B_{ll^\prime}) \right\rvert + V \nonumber\\
    &\leq \sum\limits_{l, l^\prime \in [K]} \sum\limits_{(i,j) \in T_{ll^\prime}} [2\epsilon_n A_{ij} + (2B_{ll^\prime} + \epsilon_n)\epsilon_n] + V \nonumber\\
    &= IV +V.\label{SBM_lossdiff}
\end{align}

\begin{figure}[!tb]
    \centering
    

    \includegraphics[width=9cm,keepaspectratio]{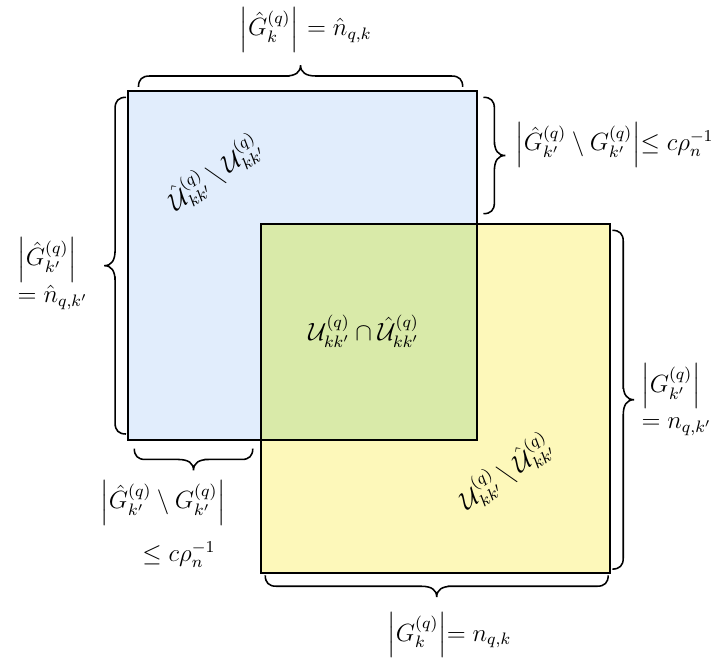} 
    
    \caption{\normalsize{Intersecting sets of $\U_{kk^\prime}^{(q)}$ and $\hat{\U}_{kk^\prime}^{(q)}$ for SBM.}}
    \label{fig:SBM_cardinal}

    {\raggedright
    \noindent\textbf{Alt text:}
    Two offset rectangles represent the estimated and true node pair sets for communities k and k prime in subnetwork q. They coincide in a central region, representing their intersection, while the remaining L-shaped regions represent the two set differences. Braces identify the estimated and true community sizes and show that each marginal mismatch is bounded by c divided by rho sub n.
    \par}
    
\end{figure}

Let $n_{q,k}$ and $\hat{n}_{q,k}$ be the number of nodes in the true and predicted $k$-th community in $S_{0q}$, respectively. 
Using Figure \ref{fig:SBM_cardinal} and Theorem \ref{SBM_bound}, we get
\begin{align}
    \lvert \hat{\U}_{kk^\prime}^{(q)} \Delta {\U}_{kk^\prime}^{(q)} \rvert &= \lvert \hat{\U}_{kk^\prime}^{(q)} \setminus \U_{kk^\prime}^{(q)} \rvert + \lvert   \U_{kk^\prime}^{(q)} \setminus \hat{\U}_{kk^\prime}^{(q)}\rvert \nonumber\\
    &\leq \lvert \hat{G}_k^{(q)} \rvert \lvert \hat{G}_k^{(q)} \setminus G_k^{(q)} \rvert + \lvert \hat{G}_{k^\prime}^{(q)} \rvert \lvert \hat{G}_{k^\prime}^{(q)} \setminus G_k^{(q)} \rvert \nonumber \\
    & \hspace{1cm} + \lvert {G}_k^{(q)} \rvert \lvert {G}_k^{(q)} \setminus \hat{G}_k^{(q)} \rvert + \lvert {G}_{k^\prime}^{(q)} \rvert \lvert {G}_{k^\prime}^{(q)} \setminus \hat{G}_k^{(q)} \rvert \nonumber\\
    &= (\hat{n}_{q,k} + \hat{n}_{q, k^\prime} + n_{q,k} + n_{q, k^\prime}) O_p(\rho_n^{-1}) \nonumber\\
    &= O_p((o+m)\rho_n^{-1}).
\end{align}
We also have with high probability
\begin{align*}
(o+m)^2 &\geq \lvert \U_{kk^\prime}^{(q)} \rvert = n_{q,k} n_{q,k^\prime} \geq \frac{\gamma^2 (o+m)^2}{(1+\gamma)^2} = \Omega_p((o+m)^2), \\
\mbox{and }\ (o+m)^2 &\geq \lvert \hat{\U}_{kk^\prime}^{(q)} \rvert \geq \lvert {\U}_{kk^\prime}^{(q)} \cap \hat{\U}_{kk^\prime}^{(q)} \rvert \geq \lvert {\U}_{kk^\prime}^{(q)}  \rvert - \lvert {\U}_{kk^\prime}^{(q)} \setminus \hat{\U}_{kk^\prime}^{(q)} \rvert \nonumber\\
&\geq \Omega_p((o+m)^2) - O_p((o+m)\rho_n^{-1}) = \Omega_{p}\left((o+m)^2\right).
\end{align*}

Thus, $\lvert \U_{kk^\prime}^{(q)} \rvert \asymp (o+m)^2$ and $\lvert \hat{\U}_{kk^\prime}^{(q)} \rvert \asymp (o+m)^2$. Then
\begin{align}
&\left\lvert \lvert \hat{\U}_{kk^\prime}^{(q)} \rvert - \lvert \U_{kk^\prime}^{(q)} \rvert \right\rvert \leq \lvert \hat{\U}_{kk^\prime}^{(q)} \Delta \hat{\U}_{kk^\prime}^{(q)} \rvert = O_p((o+m)\rho_n^{-1})\nonumber\\
\implies & \frac{\lvert \U_{kk^\prime}^{(q)} \rvert}{\lvert \hat{\U}_{kk^\prime}^{(q)} \rvert} = 1 + O_p((o+m)^{-1}\rho_n^{-1}).\label{eq:uquqhat}
\end{align}

Consider the difference
\begin{align}
        \nonumber \lvert \hat{B}_{kk^\prime} - B_{kk^\prime} \rvert &= \left\lvert \frac{1}{s}\sum\limits_{q\in [s]} \frac{1}{\lvert \hat{\U}_{kk^\prime}^{(q)} \rvert} \sum\limits_{(i,j) \in \hat{\U}_{kk^\prime}^{(q)}} A_{ij} - B_{kk^\prime} \right\rvert\\ \nonumber
        &\leq \frac{1}{s} \sum\limits_{q\in [s]} \left\lvert \frac{1}{\lvert \hat{\U}_{kk^\prime}^{(q)} \rvert} \sum\limits_{(i,j) \in \hat{\U}_{kk^\prime}^{(q)}}  A_{ij} - B_{kk^\prime} \right\rvert \\ \nonumber
        &\leq \frac{1}{s}\sum_{q \in [s]} \Bigg( \frac{\lvert \U_{kk^\prime}^{(q)} \rvert}{\lvert \hat{\U}_{kk^\prime}^{(q)} \rvert} \left \lvert \frac{1}{\lvert \U_{kk^\prime}^{(q)} \rvert} \sum\limits_{(i,j) \in \U_{kk^\prime}^{(q)}}  A_{ij} - B_{kk^\prime} \right\rvert + \left \lvert \frac{\lvert \U_{kk^\prime}^{(q)} \rvert}{\lvert \hat{\U}_{kk^\prime}^{(q)} \rvert}  - 1\right\rvert B_{kk^\prime} \nonumber\\
        &\hspace{7cm} + \frac{1}{\lvert \hat{\U}_{kk^\prime}^{(q)} \rvert}\mathop{\sum\limits_{(i,j) \in }}_{\hat{\U}_{kk^\prime}^{(q)} \Delta \U_{kk^\prime}^{(q)}} A_{ij} \Bigg)\\
        &\leq \frac{1}{s} \sum\limits_{q \in [s]} \left[ Q_{1q} + Q_{2q} + Q_{3q} \right] \ \mbox{ (say)}.
\end{align}
Then
\begin{align}
Q_{1q} &= \frac{\lvert \U_{kk^\prime}^{(q)} \rvert}{\lvert \hat{\U}_{kk^\prime}^{(q)} \rvert} \left \lvert \frac{1}{\lvert \U_{kk^\prime}^{(q)} \rvert} \sum\limits_{(i,j) \in \U_{kk^\prime}^{(q)}}  A_{ij} - B_{kk^\prime} \right\rvert \nonumber\\
&= \left(1 + O_p\left((o+m)^{-1}\rho_n^{-1}\right)\right) O_p\left( \sqrt{\frac{\rho_n}{(o+m)^2}} \right)\nonumber\\
&= O_p(\rho_n^{1/2} (o+m)^{-1}),
\end{align}
where the first bound follows from (\ref{eq:uquqhat}) and the second bound follows from (\ref{eq:sqrho}).

\begin{align}
  Q_{2q} = \left \lvert \frac{\lvert \U_{kk^\prime}^{(q)} \rvert}{\lvert \hat{\U}_{kk^\prime}^{(q)} \rvert}  - 1\right\rvert B_{kk^\prime} = O_p\left((o+m)^{-1} \rho_n^{-1}\right) O(\rho_n) = O_p\left((o+m)^{-1}\right).  
\end{align}

From (\ref{eq:halfrho}), 
\begin{align}
    Q_{3q} &= \frac{1}{\lvert \hat{\U}_{kk^\prime}^{(q)} \rvert}\mathop{\sum\limits_{(i,j) \in }}_{\hat{\U}_{kk^\prime}^{(q)} \Delta \U_{kk^\prime}^{(q)}} A_{ij} \leq \frac{\lvert \hat{\U}_{kk^\prime}^{(q)} \Delta \U_{kk^\prime}^{(q)} \rvert}{\lvert \hat{\U}_{kk^\prime}^{(q)} \rvert} \frac{1}{\lvert \hat{\U}_{kk^\prime}^{(q)} \Delta \U_{kk^\prime}^{(q)} \rvert} \mathop{\sum\limits_{(i,j) \in }}_{\hat{\U}_{kk^\prime}^{(q)} \Delta \U_{kk^\prime}^{(q)}} A_{ij} \nonumber\\
    &= O_p\left((o+m)^{-1} \rho_n^{-1}\right) O_{p}(\rho_n) = O_p\left((o+m)^{-1}\right).    
\end{align}

Thus, for any $k, k^\prime \in [K]$,
\begin{align} 
\lvert \hat{B}_{kk^\prime} - B_{kk^\prime} \rvert = O_p\left( (o+m)^{-1} \right).
\end{align}
Continuing from (\ref{SBM_lossdiff}),
\begin{align}\nonumber
    IV &= \sum\limits_{l, l^\prime \in [K]} \sum\limits_{(i,j) \in T_{ll^\prime}} [2\epsilon_n A_{ij} + (2B_{ll^\prime} + \epsilon_n)\epsilon_n]\\ \nonumber
    &= O_p\left((o+m)^{-1} \sum\limits_{l, l^\prime \in [K]} \sum\limits_{(i,j) \in T_{ll^\prime}} A_{ij} \right)\nonumber\\
    &= O_p\left( s(s-1)m^2 (o+m)^{-1} \rho_n \right), \label{eqn:IV}
\end{align}
where the last bound follows from Lemma \ref{lemma:sum-aij}.

Now,
\begin{align}
    V &= \mathop{\sum\limits_{(k, k^\prime)\neq(l, l^\prime)}}_{\in [K]} \sum\limits_{(i,j) \in T_{k,k^\prime, l,l^\prime}} \left\lvert (A_{ij} - \hat{B}_{kk^\prime})^2 - (A_{ij} - B_{ll^\prime})^2 \right\rvert \nonumber\\
    &\leq \mathop{\sum\limits_{(k, k^\prime)\neq(l, l^\prime)}}_{\in [K]} \sum\limits_{(i,j) \in T_{k,k^\prime, l,l^\prime}} 2\left(A_{ij} + B_{ll^\prime} + \lvert \hat{B}_{kk^\prime} - B_{ll^\prime} \rvert  \right) \lvert \hat{B}_{kk^\prime} - B_{ll^\prime} \rvert.
\end{align}
Consider any $(k, k^\prime) \neq (l, l^\prime) \in [K]$
\begin{align}
    \lvert \hat{B}_{kk^\prime} - B_{ll^\prime} \rvert &\leq \lvert \hat{B}_{kk^\prime} - B_{kk^\prime} \rvert + \lvert B_{kk^\prime} \rvert + \lvert B_{ll^\prime} \rvert = O_p((o+m)^{-1}) + O(\rho_n) = O_p(\rho_n).
\end{align}
Following an argument similar to the one in Fig. \ref{fig:SBM_cardinal}, it can be shown that for any $(k, k^\prime) \neq (l, l^\prime) \in [K]$
\begin{align*}
    \left\lvert T_{k,k^\prime,l,l^\prime} \right\rvert = \left\lvert \hat{T}_{kk^\prime} \cap {T}_{ll^\prime} \right\rvert \leq \left\lvert \hat{T}_{ll^\prime} \Delta {T}_{ll^\prime} \right\rvert = O_p\left(m \rho_n^{-1}\right),
\end{align*}
where the inequality follows from the fact that the node pairs in $\hat{T}_{k,k^\prime} \cap T_{ll^\prime}$ are contained in the set of node pairs that are mislabelled in $T_{ll^\prime}$ by the community detection algorithm. Therefore,
\begin{align}
    V = O_p(m \rho_n). \label{eqn:V}
\end{align}

Combining (\ref{eqn:IV}) and (\ref{eqn:V}), we get
\begin{align}
    \lvert L(A, \hat{P}) - L(A) \rvert &= O_p(s(s-1)m^2 (o+m)^{-1}\rho_n) + O_p(m \rho_n) \nonumber\\
    &= O_p(s(s-1)m^2 (o+m)^{-1}\rho_n).
\end{align}

\subsection{Proof of Theorem \ref{theorem:dcbm_lossdiff}}\label{sec:suppl-dcbm}
In this section, we prove the bounds on the difference between computed loss and the oracle loss for DCBM. Case 1 deals with the scenario $\K < K$ and Case 2 treats the $\K = K$ case. The same notations and the definitions for the subsets of node pairs in (\ref{notation:test_subsets}) are used in this subsection as well. We would also like to note that Lemma \ref{lemma:maxmincomm} is also applicable for the DCBM case. First, we restate a result on an upper bound of the misclustering rate of \code{SSC} (Algorithm \ref{algo:spher_spec_clus}) on subnetworks of DCBM from \citet{sonnet_comm}. This section requires two additional notations --- $\mathbf{o}$ and {$\mathbf{o}_p$}. For any two real sequences $\{a_n\}$ and $\{b_n\}$, $a_n = \mathbf{o}(b_n)$ if $\limsup\limits_{n \to \infty} \left\lvert a_n / b_n \right\rvert = 0$. For any two sequences of random variables $\{X_n\}$ and $\{Y_n\}$, $X_n = {\mathbf{o}}_p(Y_n)$ if for every $\epsilon > 0$, $\lim\limits_{n \to \infty} P\left( \left\lvert X_n/Y_n \right\rvert > \epsilon \right) = 0$.

\begin{theorem}{\citep[Theorem~6]{sonnet_comm}}\label{DCBM_bound}
    Under Assumption \ref{assump:DCBM} except for $\rho_n = \omega(n^{-1}\log{n})$, the number of misclustered nodes from \code{SSC} (Algorithm \ref{algo:spher_spec_clus}) on $A_{S_{0q}}$, where $A$ is from DCBM with $n$ nodes and $K$ communities, is ${\mathbf{o}_p}\left(\sqrt{n\rho_n^{-1}}\right)$ for $q \in [s]$, as $o/n > c$ for some $c>0$ and $n \to \infty$.
\end{theorem}

The following lemma bounds the distance between the eigenvectors of principal submatrices of $A$ and $P$. The proof is based on \citet[Lemma~6,7]{cv_comm} and \citet[Theorem 5.2]{lei2015}.
\begin{lemma}\label{lemma:dc-eigen-diff}
    For any subnetwork spanned by nodes in $\mathcal{S} \subset [n]$ of size $\eta$ from a DCBM with $n$ nodes and $K$ communities satisfying Assumption \ref{assump:DCBM}, let $\hat{U}_{\mathcal{S}}$ and $U_{\mathcal{S}}$ be the matrix of top $K$ eigenvectors of $A_{\mathcal{S}}$ and $P_{\mathcal{S}}$, respectively. Then for some constant $c^{\prime\prime}$,
    \begin{align}
        \lVert \hat{U}_{\mathcal{S}} - U_{\mathcal{S}} \rVert_F \leq \frac{c^{\prime\prime}}{\sqrt{n\rho_n}} \ \ \ \ \mbox{ with probability } \geq 1 - \frac{1}{n},
    \end{align}
    as $\eta/n >c$ for some $c > 0$ and $n \to \infty$.
\end{lemma}
\begin{proof}
    From Assumption \ref{assump:DCBM},
    \begin{align}
        n \max_{i,j} P_{ij} \leq n\rho_n \max_{kk^\prime} B_{0, kk^\prime} (\max_{i} \Psi)^2 \leq b_0 n\rho_n,
    \end{align}
    for some constant $b_0 > 0$. Thus $b = b_0 n\rho_n \geq b_1 n \times n^{-1/3} = b_1 n^{2/3} \geq b_1 \log{n}$ for some constant $b_1 > 0$. Therefore, from \citet[Theorem~5.2]{lei2015}, \begin{align}
        P(\lVert A - P \rVert \leq C\sqrt{n\rho_n}) \geq 1 - \frac{1}{n},
    \end{align}
    where $C = C(C_b, b)$.

    Note that $\lVert A_{\mathcal{S}} - P_{\mathcal{S}}\rVert = \sigma_1(A_{\mathcal{S}} - P_{\mathcal{S}}) \leq \sigma_1(A-P) = \lVert A - P \rVert$, where $\sigma_1(\cdot)$ is the largest singular value of a matrix and the inequality follows from interlacing of singular vectors \citep[Theorem~1]{singularinterlacing_Thompson}. Then
    \begin{align}
        P\left(\lVert A_{\mathcal{S}} - P_{\mathcal{S}} \rVert \leq \frac{C}{\sqrt{c}}\sqrt{\eta \rho_n} \right) \geq P\left(\lVert A - P \rVert \leq C\sqrt{n \rho_n} \right) \geq 1 - \frac{1}{n}.\label{eqn:dc-submatrix}
    \end{align}

    Let $\Psi_{\mathcal{S}}$ be the $\eta \times K$ matrix of the degree parameters of the nodes in $\mathcal{S}$ such that the $i$-th row contains all $0$ except $\psi_{i}$ at the $g_i$-th position. 
    Let $\lambda_K(M)$ denote the $K$-th largest eigenvalue of any matrix $M$. Note that $\Psi^{\top}_{\mathcal{S}} \Psi_{\mathcal{S}} = diag\left(\sum\limits_{i \in \mathcal{S} \cap G_1} \psi_{i}^2, \ldots, \sum\limits_{i \in \mathcal{S} \cap G_K} \psi_{i}^2 \right)$. 
    If $B_{0} = \beta \Gamma \beta^\top$ is the eigen decomposition of $B_{0}$, then $\Psi_{\mathcal{S}} B_{0} \Psi_{\mathcal{S}}^\top = (\Psi_{\mathcal{S}} \beta) \Gamma (\Psi_{\mathcal{S}} \beta)^\top = (\Psi^\prime_{\mathcal{S}} \beta) \Gamma^\prime (\Psi^\prime_{\mathcal{S}} \beta)^\top$ is an eigen decomposition of $\Psi B_{0} \Psi^\top$, where $\Psi^\prime_{\mathcal{S}}$ and $\Gamma^\prime$ are correspondingly scaled. Then
    \begin{align}
        &\lambda_K(P_{\mathcal{S}}) = \rho_n \lambda_K(P_{0\mathcal{S}}) = \rho_n \lambda_K(\Psi_{\mathcal{S}} B_{0} \Psi_{\mathcal{S}}^{\top})\nonumber\\
        & \hspace{1.1cm}= \rho_n \lambda_K(B_{0}) \sum_{i \in \mathcal{S}} \psi_i^2 \geq \psi_0^2 \eta \rho_n \lambda_K(B_0) \geq \psi_0^2 c^\prime n \rho_n, \label{eqn:dc-lambda}
    \end{align}
    where the last inequality follows from Assumption \ref{assump:DCBM}.
    Next, using \citet[Lemma~7]{cv_comm} along with (\ref{eqn:dc-submatrix}) and (\ref{eqn:dc-lambda}), we have
    \begin{align}
        \lVert \hat{U}_{\mathcal{S}} - U_{\mathcal{S}} \rVert_F \leq \frac{2\sqrt{2}\lVert A_{\mathcal{S}} - P_{\mathcal{S}} \rVert}{\lambda_K(P_{\mathcal{S}})} \leq \frac{c^{\prime\prime}}{\sqrt{n \rho_n}} \ \ \ \ \mbox{ with probability } \geq 1 - \frac{1}{n}.  
    \end{align}
\end{proof}

\begin{lemma}\label{lemma:algebraic-wgt}
    For any $(a_i, p_i, w_i) \in \R^2 \times [0, \infty)$, $i \in [n]$,
    \begin{align}
        \sum\limits_{i = 1}^n (a_i - p_i)^2 w_i \geq \sum\limits_{i=1}^n (a_i - \overline{a}_w)^2 w_i - 2\sum\limits_{i=1}^n (a_i - \overline{a}_w)p_i w_i,
    \end{align}
    where $\overline{a}_w = \sum_{i=1}^n a_i w_i/\sum_{i=1}^n w_i$.
\end{lemma}

\begin{proof}
    Consider
    \begin{align}
        \sum\limits_{i = 1}^n (a_i - p_i)^2 w_i &= \sum\limits_{i = 1}^n \left((a_i - \overline{a}_w) + (\overline{a}_w - p_i )\right)^2 w_i \nonumber\\
        &= \sum\limits_{i = 1}^n (a_i - \overline{a}_w)^2 w_i + \sum\limits_{i = 1}^n (\overline{a}_w - p_i)^2 w_i + 2\sum\limits_{i = 1}^n (a_i - \overline{a}_w)(\overline{a}_w - p_i) w_i \nonumber\\
        &\geq \sum\limits_{i = 1}^n (a_i - \overline{a}_w)^2 w_i + 2\overline{a}_w \sum\limits_{i = 1}^n (a_i - \overline{a}_w) w_i - 2\sum\limits_{i = 1}^n (a_i - \overline{a}_w) p_i w_i \nonumber\\
        &= \sum\limits_{i=1}^n (a_i - \overline{a}_w)^2 w_i - 2\sum\limits_{i=1}^n (a_i - \overline{a}_w)p_i w_i.
    \end{align}
\end{proof}

\paragraph*{Proof of Theorem \ref{theorem:dcbm_lossdiff} for DCBM}

\paragraph*{Case 1: ($\K < K$ for DCBM)} 
Recall that the test set is structured as $\S^c = \bigcup\limits_{1 \leq p < q \leq s} (S_p \times S_q)$, where the union is disjoint. For $\Tilde{K} < K$ and any specific $1 \leq p < q \leq s$, let $l_1, l_2, l_3 \in [K]$ and $k_1, k_2 \in [\K]$ be the indices from Lemma \ref{lemma:sbm-dcbm_php}. Without loss of generality, assume $l_1 = 1, l_2 = 2, l_3 = 3, k_1 = 1, k_2 = 2$ in Lemma \ref{lemma:sbm-dcbm_php} for some specific $1 \leq p < q \leq s$. Define
\begin{align*}
    &\Psi_{k,k^\prime,l,l^\prime}^{(pq)} = \sum\limits_{(i,j) \in T_{k,k^\prime,l,l^\prime}^{(pq)}} \psi_i^2\psi_j^2 \\
    &\lambda = \frac{\Psi_{1,2,1,3}^{(pq)}}{\Psi_{1,2,1,3}^{(pq)} + \Psi_{1,2,2,3}^{(pq)}} = \frac{\sum\limits_{(i,j) i\in T_{1,2,1,3}^{(pq)}} \psi_{i}^2 \psi_{j}^2 }{\sum\limits_{(i,j) i\in T_{1,2,1,3}^{(pq)}} \psi_{i}^2 \psi_{j}^2  + \sum\limits_{(i,j) i\in T_{1,2,2,3}^{(pq)}} \psi_{i}^2 \psi_{j}^2 } ,\\
    &\hat{p}_{k,k^\prime,l,l^\prime}^{(pq)} = \frac{\sum\limits_{(i,j) \in T_{k, k^\prime, l, l^\prime}^{(pq)}} A_{ij} \psi_{i} \psi_{j}}{\sum\limits_{(i,j) \in T_{k, k^\prime, l, l^\prime}^{(pq)}} \psi_{i}^2 \psi_{j}^2 },\\
    \mbox{and }\ \ & \hat{p} = \lambda \hat{p}_{1,2,1,3}^{(pq)} + (1-\lambda) \hat{p}_{1,2,2,3}^{(pq)} = \frac{\sum\limits_{(i,j) i\in T_{1,2,1,3}^{(pq)} \cup T_{1,2,2,3}^{(pq)}} A_{ij}\psi_{i} \psi_{j}}{\Psi_{1,2,1,3}^{(pq)} + \Psi_{1,2,2,3}^{(pq)}}.
\end{align*}
The superscript ${(pq)}$ is dropped from $\hat{p}_{k,k^\prime,l,l^\prime}^{(pq)}$ from the computations related to $L^{(pq)}(A, \hat{P})$ for notational simplicity. 
Now, consider the difference between the computed loss and the oracle loss within $S_p \times S_q$ as -
\begin{align}
    & L^{(pq)} \left(A, \hat{P}^{\prime(\K)} \right) - L^{(pq)}\left( A \right) \nonumber \\
    =& \sum\limits_{(i,j) \in S_p \times S_q} \left[ \left(A_{ij} - \hat{P}_{ij}^{\prime (\K)}\right)^2 - \left(A_{ij} - P_{ij} \right)^2 \right]\nonumber\\ 
    =& \sum\limits_{(k,k^\prime, l, l^\prime) \in [\K]^2 \times [K]^2} \sum\limits_{(i,j) \in T_{k,k^\prime, l, l^\prime}^{(pq)}} \left[ \left( A_{ij} - \hat{P}_{ij}^{\prime (\K)} \right)^2 - \left( A_{ij} - \psi_i \psi_j B_{ll^\prime} \right)^2 \right]\nonumber\\ 
    =& \sum\limits_{(i,j) \in T_{1,2,1,3}^{(pq)}} \left[ \left(\frac{A_{ij}}{\psi_i\psi_j} - \frac{\hat{P}_{ij}^{\prime (\K)}}{\psi_i\psi_j}\right)^2 \psi_i^2 \psi_j^2 - \left(\frac{A_{ij}}{\psi_i\psi_j} - B_{13}\right)^2 \psi_i^2 \psi_j^2 \right]\nonumber\\ \nonumber
    &+ \sum\limits_{(i,j) \in T_{1,2,2,3}^{(pq)}} \left[ \left(\frac{A_{ij}}{\psi_i\psi_j} - \frac{\hat{P}_{ij}^{\prime(\K)}}{\psi_i\psi_j}\right)^2 \psi_i^2 \psi_j^2 - \left(\frac{A_{ij}}{\psi_i\psi_j} - B_{23}\right)^2 \psi_i^2 \psi_j^2 \right] \\
    &+ \sum\limits_{(k,k^\prime, l, l^\prime) \notin \{ (1,2,1,3), (1,2,2,3) \} } \sum\limits_{(i,j) \in T_{k,k^\prime, l, l^\prime}^{(pq)}} \left[ \left(\frac{A_{ij}}{\psi_i\psi_j} - \frac{\hat{P}_{ij}^{\prime(\K)}}{\psi_i\psi_j}\right)^2 \psi_i^2 \psi_j^2 - \left(\frac{A_{ij}}{\psi_i\psi_j} - B_{ll^\prime}\right)^2 \psi_i^2 \psi_j^2 \right] \nonumber \\
    \geq& \sum\limits_{(i,j) \in T_{1,2,1,3}^{(pq)}} \left[ \left(\frac{A_{ij}}{\psi_i\psi_j} - \hat{p}\right)^2\psi_i^2\psi_j^2 - \left(\frac{A_{ij}}{\psi_i\psi_j} - B_{13}\right)^2\psi_i^2\psi_j^2 \right] \nonumber \\
    &+ \sum\limits_{(i,j) \in T_{1,2,2,3}^{(pq)}} \left[ \left(\frac{A_{ij}}{\psi_i\psi_j} - \hat{p}\right)^2\psi_i^2\psi_j^2 - \left(\frac{A_{ij}}{\psi_i\psi_j} - B_{23}\right)^2\psi_i^2\psi_j^2 \right]  \nonumber\\ \nonumber
    &+ \sum\limits_{(k,k^\prime, l, l^\prime) \notin \{ (1,2,1,3), (1,2,2,3) \} } \sum\limits_{(i,j) \in T_{k,k^\prime, l, l^\prime}^{(pq)}} \left[ \left(\frac{A_{ij}}{\psi_i\psi_j} - \hat{p}_{k,k^\prime,l,l^\prime}\right)^2\psi_i^2\psi_j^2 - \left(\frac{A_{ij}}{\psi_i\psi_j} - B_{ll^\prime}\right)^2\psi_i^2\psi_j^2 \right]\\
    &- \sum\limits_{(i,j) \in T_{1,2,1,3}^{(pq)} \cup T_{1,2,2,3}^{pq}} \left(\frac{A_{ij}}{\psi_i\psi_j} - \hat{p} \right) \hat{P}_{ij}^{\prime(\K)} \psi_i^2 \psi_j^2 \nonumber\\
    &- \sum\limits_{(k,k^\prime, l, l^\prime) \notin \{ (1,2,1,3), (1,2,2,3) \} } \sum\limits_{(i,j) \in T_{k,k^\prime, l, l^\prime}^{(pq)}} \left(\frac{A_{ij}}{\psi_i\psi_j} - \hat{p}_{k, k^\prime, l, l^\prime} \right) \hat{P}_{ij}^{\prime(\K)} \psi_i^2 \psi_j^2 \nonumber\\
    &\hspace{10cm}(\mbox{from Lemma \ref{lemma:algebraic-wgt}}) \nonumber\\
    =& I + II + III - E - F \ \mbox{(say)}.\label{eqn:BM_cand_less_K_loss_diff} 
\end{align}

Then,
\begin{align}
    I &= \sum\limits_{(i,j) \in T_{1,2,1,3}^{(pq)}} \left[ \left(\frac{A_{ij}}{\psi_i\psi_j} - \hat{p}\right)^2\psi_i^2\psi_j^2 - \left(\frac{A_{ij}}{\psi_i\psi_j} - B_{13}\right)^2\psi_i^2\psi_j^2 \right]\nonumber\\ \nonumber
    &= -2\hat{p} \sum\limits_{(i,j) \in T_{1,2,1,3}^{(pq)}} A_{ij}\psi_i\psi_j + \hat{p}^2 \Psi_{1,2,1,3}^{(pq)} + 2B_{13} \sum\limits_{(i,j) \in T_{1,2,1,3}^{(pq)}} A_{ij} \psi_i\psi_j -  B_{13}^2 \Psi_{1,2,1,3}^{(pq)}\\ \nonumber
    &= \Psi_{1,2,1,3}^{(pq)} \left( \hat{p}^2 - 2 \hat{p}\hat{p}_{1,2,1,3} + \hat{p}_{1,2,1,3}^2 - \hat{p}_{1,2,1,3}^2 + 2\hat{p}_{1,2,1,3} B_{13} - B_{13}^2 \right)\\ \nonumber
    &= \Psi_{1,2,1,3}^{(pq)} \left[ (\hat{p} - \hat{p}_{1,2,1,3})^2 - (\hat{p}_{1,2,1,3} - B_{13})^2 \right]\\ \nonumber
    &= \Psi_{1,2,1,3}^{(pq)} \left[ (1-\lambda)^2(\hat{p}_{1,2,1,3} - \hat{p}_{1,2,2,3})^2 - (\hat{p}_{1,2,1,3} - B_{13})^2 \right]\\ \nonumber
    &\geq \Psi_{1,2,1,3}^{(pq)} \Bigg[ \frac{(1-\lambda)^2}{2} (B_{13} - B_{23})^2 - (1-\lambda)^2 ((\hat{p}_{1,2,1,3} - B_{13}) - (\hat{p}_{1,2,2,3} - B_{23}))^2\\ \nonumber
    &\hspace{10cm} - (\hat{p}_{1,2,1,3} - B_{13})^2 \Bigg]\\ \nonumber
    &\hspace{8cm} \left(\text{as } (a+b)^2 \geq \frac{b^2}{2} - a^2 \right)\\ \nonumber
    &\geq \Psi_{1,2,1,3}^{(pq)} \Bigg[ \frac{(1-\lambda)^2}{2} (B_{13} - B_{23})^2 - 2(1-\lambda)^2 (\hat{p}_{1,2,1,3} - B_{13})^2\\
    &\hspace{5cm}  - 2(1-\lambda)^2 (\hat{p}_{1,2,2,3} - B_{23})^2 - (\hat{p}_{1,2,1,3} - B_{13})^2 \Bigg] \nonumber\\
    &\hspace{8cm} \left( \text{ as } (a-b)^2 \leq 2(a^2 + b^2) \right).\label{eqn:term_I_BM}
\end{align}

From Assumption \ref{assump:DCBM}, 
\begin{align*}
&\Psi_{1,2,1,3}^{(pq)} = \sum\limits_{(i,j) \in T_{1,2,1,3}^{(pq)}} \psi_i^2\psi_j^2 \geq \psi_{0}^4 \lvert T_{1,2,1,3}^{(pq)} \rvert \geq \Tilde{c}\psi_{0}^4 m^2,\\
\mbox{ and }\ \ &\Psi_{1,2,1,3}^{(pq)} \leq \lvert T_{1,2,1,3}^{(pq)} \rvert \leq \lvert S_p \times S_q \rvert = m^2. 
\end{align*}
Thus, $\Psi_{1,2,1,3}^{(pq)} \asymp m^2$. Similarly, $\Psi_{1,2,2,3}^{(pq)} \asymp m^2$. From Lemma \ref{lemma:sbm-dcbm_php},
\begin{align}
\frac{(1-\lambda)^2}{2} (B_{13} - B_{23})^2 &= \frac{ (\Psi_{1,2,2,3}^{(pq)})^2 }{2(\Psi_{1,2,1,3}^{(pq)} + \Psi_{1,2,2,3}^{(pq)} )^2} (B_{13} - B_{23})^2
\geq c_1 \rho_n^2. \label{eqn:bgap}
\end{align}
For $(k, k^\prime, l, l^\prime) \in  \{(1,2,1,3), (1,2,2,3)\}$, we have
\begin{align}
    \lvert \hat{p}_{k,k^\prime, l, l^\prime} - B_{ll^\prime} \rvert &= \left\lvert \frac{\sum\limits_{(i,j) \in T_{k, k^\prime, l, l^\prime}^{(pq)}} A_{ij} \psi_{i} \psi_{j}}{\sum\limits_{(i,j) \in T_{k, k^\prime, l, l^\prime}^{(pq)}} \psi_{i}^2 \psi_{j}^2 } - B_{ll^\prime} \right \rvert \nonumber\\
    & = \frac{\sum\limits_{(i,j) \in T_{k, k^\prime, l, l^\prime}^{(pq)}} \left\lvert A_{ij} - B_{ll^\prime}\psi_i \psi_j \right\rvert \psi_{i} \psi_{j}}{\sum\limits_{(i,j) \in T_{k, k^\prime, l, l^\prime}^{(pq)}} \psi_{i}^2 \psi_{j}^2 } \nonumber\\
    & \leq \frac{1}{\psi_0^4 \lvert T_{k, k^\prime, l, l^\prime}^{(pq)} \rvert} \sum\limits_{(i,j) \in T_{k, k^\prime, l, l^\prime}^{(pq)}} \left\lvert A_{ij} - B_{ll^\prime}\psi_i \psi_j \right\rvert \nonumber\\
    &= \frac{1}{\psi_0^4 \lvert T_{k, k^\prime, l, l^\prime}^{(pq)} \rvert} \sum\limits_{(i,j) \in T_{k, k^\prime, l, l^\prime}^{(pq)}} \left\lvert A_{ij} - E(A_{ij})\right\rvert  \nonumber \\
    & = O_{p}\left( \sqrt{\frac{\rho_n}{\lvert T_{k, k^\prime, l, l^\prime}^{(pq)} \rvert}}\right) \ \ \mbox{ (from Lemma \ref{lemma:sum-aij} )}. \label{eqn:pkkll_bound_dc}
\end{align}

Plugging in the inequalities in (\ref{eqn:bgap}) and (\ref{eqn:pkkll_bound_dc}) in (\ref{eqn:term_I_BM}), we get
\begin{align*}
    I \geq \Tilde{c} \psi_0^4 m^2\rho_n^2 - O_p(\rho_n).
\end{align*}
Similarly, 
\begin{align*}
    II \geq \Tilde{c} \psi_0^4 m^2\rho_n^2 - O_p(\rho_n).
\end{align*}

\begin{align}
    III &= \sum\limits_{(k,k^\prime, l, l^\prime) \notin \{ (1,2,1,3), (1,2,1,2) \} } \sum\limits_{(i,j) \in T_{k,k^\prime, l, l^\prime}^{(pq)}} \left[ \left(\frac{A_{ij}}{\psi_i\psi_j} - \hat{p}_{k,k^\prime,l,l^\prime}\right)^2\psi_i^2\psi_j^2 - \left(\frac{A_{ij}}{\psi_i\psi_j} - B_{ll^\prime}\right)^2\psi_i^2\psi_j^2 \right]\nonumber\\ 
    &= \sum\limits_{(k,k^\prime, l, l^\prime) \notin \{ (1,2,1,3), (1,2,1,2) \} } \sum\limits_{(i,j) \in T_{k,k^\prime, l, l^\prime}^{(pq)}} \left[ A_{ij}^2 - 2\hat{p}_{k,k^\prime,l,l^\prime} A_{ij} \psi_i \psi_j + \hat{p}_{k,k^\prime,l,l^\prime}^2 \psi_i^2\psi_j^2 \right. \nonumber\\
    &\hspace{8cm}\left. - A_{ij}^2 + 2B_{ll^\prime} A_{ij} \psi_i \psi_j - B_{ll^\prime}^2\psi_i^2\psi_j^2 \right] \nonumber\\
    &= \sum\limits_{(k,k^\prime, l, l^\prime) \notin \{ (1,2,1,3), (1,2,1,2) \} } \Psi_{k,k^\prime, l, l^\prime}^{(pq)} \left[ -2\hat{p}_{k,k^\prime,l,l^\prime}^2 + \hat{p}_{k,k^\prime,l,l^\prime}^2 + 2B_{ll^\prime} \hat{p}_{k,k^\prime,l,l^\prime} - B_{ll^\prime}^2 \right] \nonumber\\
    &= - \sum\limits_{(k,k^\prime, l, l^\prime) \notin \{ (1,2,1,3), (1,2,1,2) \} } \Psi_{k,k^\prime, l, l^\prime}^{(pq)}  \left(\hat{p}_{k, k^\prime, l, l^\prime} - B_{ll^\prime}\right)^2 \nonumber\\
    &= - \sum\limits_{(l,l^\prime) \notin \{(1,2), (1,3)\}} \sum\limits_{(k, k^\prime) \neq (1,2)} \Psi_{k,k^\prime, l, l^\prime}^{(pq)}  \left(\hat{p}_{k, k^\prime, l, l^\prime} - B_{ll^\prime}\right)^2, \nonumber\\
    &= - O_p(\rho_n),
\end{align}
where the last inequality follows from applying Lemma \ref{lemma:bern} on the inner sum as it contains averages of $\Omega(n^2)$ independent Bernoulli random variables.

Note that $E(A_{ij}) = B_{ll^\prime} \psi_i\psi_j$ for any $(i,j) \in T_{k,k^\prime,l,l^\prime}$. 
Also note that the quantity $E$ in \eqref{eqn:BM_cand_less_K_loss_diff} deals with the case where the node pairs are from $T_{1,2,1,3}^{(pq)} \cup T_{1,2,2,3}^{(pq)}$. Define the events $\Gamma_{k,k^\prime,l,l^\prime}^{(pq)} = \left\{ (i,j) \in T_{k,k^\prime,l,l^\prime}^{(pq)} \right\}$. We have that the events $\Gamma_{1,2,1,3}^{(pq)}$ and $\Gamma_{1,2,2,3}^{(pq)}$ are disjoint. Consider the following conditional expectation
\begin{align}
    E_{ij} &\coloneqq E\left[ \frac{A_{ij}}{\psi_i \psi_j} \left\lvert\right. \Gamma_{1,2,1,3}^{(pq)} \cup \Gamma_{1,2,2,3}^{(pq)} \right] \nonumber\\
    &= \frac{1}{\psi_i\psi_j} P\left( A_{ij} = 1 \left\lvert\right. \Gamma_{1,2,1,3}^{(pq)} \cup \Gamma_{1,2,2,3}^{(pq)} \right) \nonumber\\
    &= \frac{1}{\psi_i\psi_j} \left[ P\left( \{A_{ij} = 1\} \cap \Gamma_{1,2,1,3}^{(pq)} \left\lvert\right. \Gamma_{1,2,1,3}^{(pq)} \cup \Gamma_{1,2,2,3}^{(pq)} \right) + P\left( \{A_{ij} = 1\} \cap \Gamma_{1,2,2,3}^{(pq)} \left\lvert\right. \Gamma_{1,2,1,3}^{(pq)} \cup \Gamma_{1,2,2,3}^{(pq)} \right) \right] \nonumber\\
    &= \frac{1}{\psi_i\psi_j} \left[ P\left( \{A_{ij} = 1\} \left\lvert\right. \Gamma_{1,2,1,3}^{(pq)} \right) P\left( \Gamma_{1,2,1,3}^{(pq)} \left\lvert\right. \Gamma_{1,2,1,3}^{(pq)} \cup \Gamma_{1,2,2,3}^{(pq)} \right) \right. \nonumber\\
    & \hspace{2cm} + \left. P\left( \{A_{ij} = 1\} \left\lvert\right. \Gamma_{1,2,2,3}^{(pq)} \right) P\left( \Gamma_{1,2,2,3}^{(pq)} \left\lvert\right. \Gamma_{1,2,1,3}^{(pq)} \cup \Gamma_{1,2,2,3}^{(pq)} \right) \right] \nonumber\\
    &= B_{13} P\left( \Gamma_{1,2,2,3}^{(pq)} \left\lvert\right. \Gamma_{1,2,1,3}^{(pq)} \cup \Gamma_{1,2,1,3}^{(pq)} \right) + B_{23} P\left( \Gamma_{1,2,2,3}^{(pq)} \left\lvert\right. \Gamma_{1,2,1,3}^{(pq)} \cup \Gamma_{1,2,2,3}^{(pq)} \right) \nonumber\\
    &\approx \lambda B_{13} + (1 - \lambda) B_{23} = \mu \ \mbox{(say)}.
\end{align}

We also show that
\begin{align}
    E\left[\hat{p} \left\lvert\right. \Gamma_{1,2,1,3}^{(pq)} \cup \Gamma_{1,2,2,3}^{(pq)} \right] = \frac{ \sum\limits_{(i,j) \in T_{1,2,1,3}^{(pq)} \cup T_{1,2,2,3}^{(pq)}} E\left[\frac{A_{ij}}{\psi_i \psi_j} \left\lvert\right. \Gamma_{1,2,1,3}^{(pq)} \cup \Gamma_{1,2,2,3}^{(pq)} \right] \psi_i^2 \psi_j^2 }{\sum\limits_{(i,j) \in T_{1,2,1,3}^{(pq)} \cup T_{1,2,2,3}^{(pq)}} \psi_i^2 \psi_j^2} = \mu.
\end{align}

Then conditional under the event $\Gamma_{1,2,1,3}^{(pq)} \cup \Gamma_{1,2,2,3}^{(pq)}$
\begin{align}
    E &= \sum\limits_{(i,j) \in T_{1,2,1,3}^{(pq)} \cup T_{1,2,2,3}^{(pq)}} \left( \frac{A_{ij}}{\psi_i\psi_j} - \hat{p} \right) \hat{P}_{ij}^{\prime(\K)} \psi_i^2 \psi_j^2 \nonumber\\
    &\leq \sum\limits_{(i,j) \in T_{1,2,1,3}^{(pq)} \cup T_{1,2,2,3}^{(pq)}} \left\lvert \frac{A_{ij}}{\psi_i\psi_j} - \mu \right\rvert + \sum\limits_{(i,j) \in T_{1,2,1,3}^{(pq)} \cup T_{1,2,2,3}^{(pq)}} \left\lvert \hat{p} - \mu \right\rvert \nonumber\\
    &= \left\lvert T_{1,2,1,3}^{(pq)} \cup T_{1,2,2,3}^{(pq)} \right\rvert O_p\left( \sqrt{\frac{\mu}{\left\lvert T_{1,2,1,3}^{(pq)} \cup T_{1,2,2,3}^{(pq)} \right\rvert}} \right)  \ \mbox{(from Lemma \ref{lemma:bern})} \nonumber\\
    &= O_p\left( m \sqrt{\rho_n} \right).
\end{align}

Using a similar argument
\begin{align}
    F = O_p(m \sqrt{\rho_n}).
\end{align}

Combining all the cases, we have
\begin{align}
    L^{(pq)}(A, \hat{P}) - L^{(pq)}(A) &\geq I + II + III - E - F \geq \Tilde{c}\psi_0^4 m^2\rho_n^2 - O_p(\rho_n) - O_p(m\sqrt{\rho_n}) \nonumber\\
    &= \Omega_p(m^2\rho_n^2),
\end{align}
as the assumption $\rho_n = \omega\left(n^{-1/3}\right)$ implies that the first term dominates.

Then the difference between the computed loss and the oracle loss over the entire test set satisfies
\begin{align}
    L(A, \hat{P}) - L(A) = \sum\limits_{1 \leq p < q \leq s} \left[L^{(pq)}(A, \hat{P}) - L^{(pq)}(A)\right] = \Omega_p(s(s-1)m^2 \rho_n^2).
\end{align}

\paragraph*{Case 2.2: DCBM ($\K = K$)}
For any $q\in[s]$, the community normalized versions of the connectivity matrix and the degree heterogeneity parameters for the $q$-th subgraph are defined as
\begin{align}
    \psi^\prime_{q,i} = \frac{\psi_i}{\sqrt{\sum\limits_{j\in S_{0q} : g_j = g_i} \psi_j^2 }} \ \ \mbox{ and }\ \ B^{\prime}_{q, kk^\prime} = B_{k k ^\prime} \sqrt{\sum\limits_{i\in S_{0q} : g_i = k} \psi_i^2 \sum\limits_{j \in S_{0q} : g_j = k^\prime} \psi_j^2}.
\end{align}

From Assumption \ref{assump:DCBM}, for any $q \in [s]$, we have
\begin{align}
    &\psi^\prime_{q,i} = \frac{\psi_i}{\sqrt{\sum\limits_{j\in S_{0q} : g_j = g_i} \psi_j^2 }} \leq \frac{\psi_i}{\sqrt{(o+m)}\psi_0} \nonumber\\
    \mbox{and } \ \ & \psi^\prime_{q,i} = \frac{\psi_i}{\sqrt{\sum\limits_{j\in S_{0q} : g_j = g_i} \psi_j^2 }} \geq \frac{\psi_i}{\sqrt{o+m}}.\label{eqn:psi2}
\end{align}
Thus, $\psi_{p,i}^\prime \asymp (o+m)^{-1/2} \psi_i$ and $\psi_{q,j}^\prime \asymp (o+m)^{-1/2} \psi_i$.

Recall that $\psi^\prime_{q,i} = \lVert U_{q, i\cdot} \rVert$ and $\hat{\psi}^{\prime}_{i} = \lVert \hat{U}_{q, i\cdot} \rVert$ for $i \in S_{0q}$, $q \in [s]$. Define $\hat{\psi}^{\prime\prime}_q$ and $\psi^\prime_q$ as $(o+m) \times 1$ vectors, whose $i$th elements are $\hat{\psi}^{\prime}_{i}$ and $\psi^\prime_{q,i}$, respectively for $i \in S_{0q}$, $q \in [s]$. 
Then by triangle inequality and Cauchy-Schwarz inequality 
\begin{align}
    \lVert \hat{\psi}^{\prime\prime}_q - \psi^\prime_q \rVert_1 &= \sum\limits_{i \in S_{0q}} \left\lvert \lVert \hat{U}_{q, i\cdot} \rVert - \lVert U_{q, i\cdot} \rVert \right\rvert \leq \sum\limits_{i \in S_{0q}} \lVert \hat{U}_{q, i\cdot} - U_{q, i\cdot} \rVert \nonumber\\
    &\leq \sqrt{o+m} \lVert \hat{U}_{q} - U_{q} \rVert_F \nonumber\\
    &\leq C\sqrt{o+m}/\sqrt{n\rho_n} \ \mbox{(with high probability from Lemma \ref{lemma:dc-eigen-diff})} \nonumber\\
    &\leq C \rho_n^{-1/2} \ \mbox{ for all }\ q \in [s]. \label{eqn:dc-psidiff}
\end{align}
We also have $\min\limits_{i \in S_{0q}} \psi_{q,i}^\prime \geq \psi_0 (o+m)^{-1/2}$, $q \in [s]$ from Assumption \ref{assump:DCBM}. For $\alpha \in (0,1)$ and $q \in [s]$, define the following subset of nodes
\begin{align*}
    \mathcal{S}_{\alpha q} \coloneqq \left\{ i \in S_{0q} : \lvert \hat{\psi}^\prime_{i} - \psi^\prime_{q, i} \rvert \leq (o+m)^{-{1}/{2}} (n^\alpha \rho_n)^{-{1}/{2}} \right\}.
\end{align*}
Then for all $i \in \mathcal{S}_{\alpha q}$ and $q \in [s]$ and $\rho_n = \omega(n^{-\alpha})$ for some $\alpha \in (0,1)$, from (\ref{eqn:dc-psidiff}), we have 
\begin{align}
    &\lvert \hat{\psi}^\prime_{i} - \psi^\prime_{q, i} \rvert = {\mathbf{o}}\left((o+m)^{-1/2}\right) \ \mbox{ for all }\ i \in \mathcal{S}_{\alpha q}, \nonumber\\
    \mbox{and }\ &\frac{\lvert \mathcal{S}_{\alpha q}^c \rvert}{o+m} \leq \frac{\lVert \hat{\psi}^\prime_q - \psi^\prime_q \rVert_1}{(o+m)^{{1}/{2}} n^{-\alpha/2} \rho_n^{1/2} } \leq \frac{C \rho_n^{-1/2}}{(o+m)^{{1}/{2}} n^{-\alpha/2} \rho_n^{-1/2} } \leq C^\prime n^{-(1-\alpha)/2}. \label{eqn:math_s_size}
\end{align}
Setting $\alpha = 1/3$ and using the assumption $\rho_n = \omega(n^{-1/3})$ (Assumption \ref{assump:DCBM}), on all but a vanishing proportion of nodes $i \in S_{0q}$, $q \in [s]$,
\begin{align}
    \hat{\psi}^\prime_{i} = (1 + {\mathbf{o}}(1)) \psi^\prime_{q, i} \asymp (1 + {\mathbf{o}}(1))(o+m)^{-1/2} \psi_i. \label{eqn:psi_prime}    
\end{align}

Consider the following oracle estimator of the normalized community probability matrix defined as
\begin{align}
    B^{\prime *}_{q, k k^\prime} &\coloneqq \frac{\sum\limits_{(i,j) \in (S_{0q} \times S_{0q})} A_{ij} \I\left( {g}_{i} = k,\ {g}_{j} = k^\prime \right)}{\sum\limits_{(i,j) \in (S_{0q} \times S_{0q})} {\psi}^\prime_{q,i} {\psi}^\prime_{q,j} \I\left( {g}_{i} = k,\ {g}_{j} = k^\prime \right)} \nonumber\\
    &= \frac{O_p((o+m)B_{kk^\prime})}{\Omega_p((o+m)^{-1} (o+m))} = O_p((o+m)B_{kk^\prime}),\label{eqn:Bstar}
\end{align}
where the numerator follows from Lemma \ref{lemma:sum-aij}, and the denominator from (\ref{eqn:psi2}) and Lemma \ref{lemma:maxmincomm}. Now, consider the following ratio for any $k, k^\prime \in [K]$ and $q \in [s]$
\begin{align}
    \frac{\hat{B}^\prime_{q, kk^\prime}}{B^{\prime *}_{q, kk^\prime}} = \frac{\sum\limits_{(i,j) \in \hat{\U}_{kk^\prime}^{(q)}} A_{ij} }{\sum\limits_{(i,j) \in \hat{\U}_{kk^\prime}^{(q)}} \hat{\psi}^\prime_{i} \hat{\psi}^\prime_{j}} \times \frac{\sum\limits_{(i,j) \in {\U}_{kk^\prime}^{(q)}} {\psi}^\prime_{q,i} {\psi}^\prime_{q,j}}{\sum\limits_{(i,j) \in {\U}_{kk^\prime}^{(q)}} A_{ij} } = \frac{\hat{\Sigma}}{\hat{\Sigma}_\Psi} \times \frac{\Sigma_\Psi}{\Sigma}, \ \mbox{(say)}.
\end{align}

\begin{figure}[!tb]
    \centering


    \includegraphics[width=9cm,keepaspectratio]{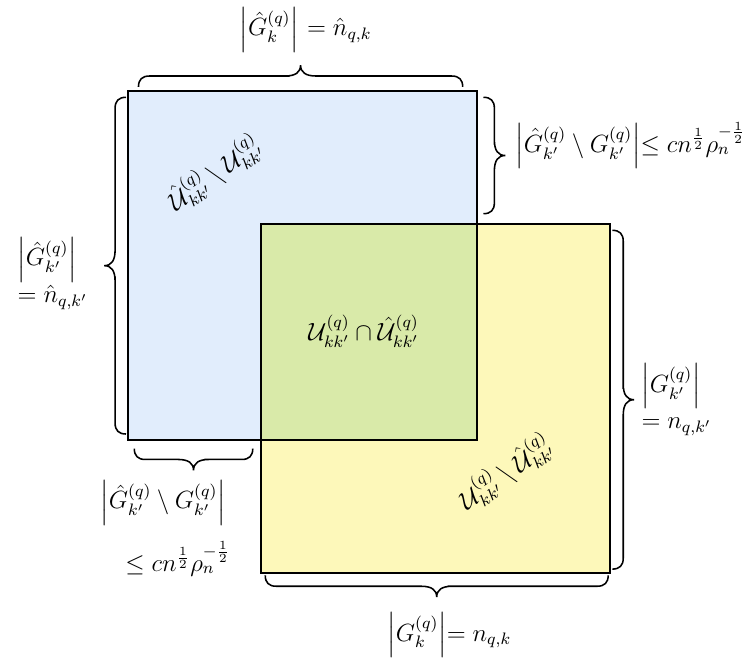} 

    \caption{\normalsize{Intersecting sets of $\U_{kk^\prime}^{(q)}$ and $\hat{\U}_{kk^\prime}^{(q)}$ for DCBM.}}
    \label{fig:DCBM_cardinal}

    {\raggedright
    \noindent\textbf{Alt text:}
    Two offset rectangles represent the estimated and true node pair sets for communities k and k prime in subnetwork q. They coincide in a central region, representing their intersection, while the remaining L-shaped regions represent the two set differences. Braces identify the estimated and true community sizes and show that each marginal mismatch is bounded by c times square root of n, divided by the square root of rho sub n.
    \par}
\end{figure}

Using Theorem \ref{DCBM_bound} and Figure \ref{fig:DCBM_cardinal}, we have $\lvert \hat{\U}_{kk^\prime}^{(q)} \setminus \U_{kk^\prime}^{(q)} \rvert = O_p\left((o+m)n^{1/2}\rho_n^{-1/2}\right)$. 
We also have $\lvert \U_{kk^\prime}^{(q)} \rvert \asymp \lvert \hat{\U}_{kk^\prime}^{(q)} \rvert \asymp (o+m)^2$. Using these bounds and Lemma \ref{lemma:sum-aij}, we get
\begin{align}
    \frac{\hat{\Sigma}}{\Sigma} &= 1 + \frac{\hat{\Sigma} - \Sigma}{\Sigma} \leq 1 + \frac{\lvert \hat{\U}_{kk^\prime}^{(q)} \setminus \U_{kk^\prime}^{(q)} \rvert}{\lvert \U_{kk^\prime}^{(q)} \rvert B_{kk^\prime}}  \nonumber\\
    &\leq 1 + \frac{O_p\left((o+m)n^{1/2} \rho_n^{-1/2}\right)}{c^\prime(o+m)^2 B_{kk^\prime}} \nonumber\\
    &= 1 + O\left( \frac{n^{1/2} \rho_n^{-1/2}}{(o+m)\rho_n} \right) = 1 + O_p\left( n^{-1/2} \rho_n^{-3/2} \right). 
\end{align}
Then under the event of $\mathcal{S}_{\frac{1}{3}q}$,
\begin{align}
    \frac{\Sigma_\Psi}{\hat{\Sigma}_\Psi} = O\left( \frac{\lvert \U_{kk^\prime}^{(q)} \rvert (o+m)^{-1} }{\lvert \hat{\U}_{kk^\prime}^{(q)} \rvert(o+m)^{-1}} \right) = O(1).
\end{align}
Thus, we have
\begin{align}
    {\hat{B}^\prime_{q, kk^\prime}} = \left(1 + O_p\left(n^{-\frac{1}{2}} \rho_n^{-\frac{3}{2}}\right)\right){B^{\prime *}_{q, kk^\prime}}.\label{eqn:b-hat}
\end{align}
Using (\ref{eqn:Bstar}) and (\ref{eqn:b-hat}), we get
\begin{align}
    &\hat{B}^\prime_{q, kk^\prime} = \left(1 + O_p\left(n^{-\frac{1}{2}} \rho_n^{-\frac{3}{2}}\right)\right) (o+m)B_{kk^\prime} \nonumber\\
    \mbox{and }\ &\hat{B}^\prime_{kk^\prime} = \frac{1}{s}\sum\limits_{q \in [s]} \hat{B}^\prime_{q, kk^\prime} = \left(1 + O_p\left(n^{-\frac{1}{2}} \rho_n^{-\frac{3}{2}}\right)\right)(o+m) B_{kk^\prime} \nonumber\\
    \implies &\left\lvert \hat{B}^\prime_{kk^\prime} - (o+m)B_{kk^\prime} \right\rvert = O_p\left( n^{-\frac{1}{2}} \rho_n^{-\frac{3}{2}} (o+m) B_{kk^\prime} \right) = O_p\left( (o+m)n^{-\frac{1}{2}} \rho_n^{-\frac{3}{2}} \rho_n \right) \nonumber\\
    & \hspace{4cm} = O_p\left( (o+m)n^{-\frac{1}{2}} \rho_n^{-\frac{1}{2}} \right).  \label{eqn:psi4}
\end{align}

Combining (\ref{eqn:psi_prime}) and (\ref{eqn:psi4}), for $1 \leq p < q \leq s$ and $i \in S_p \cap S_{\frac13p}, j \in S_q \cap S_{\frac13q}$ such that $(i, j) \in T_{l, l^\prime, l, l^\prime}$, we have
\begin{align}
    \hat{P}_{ij}^{\prime(K)} = \hat{B}^\prime_{ll^\prime} \hat{\psi}^\prime_i \hat{\psi}^\prime_j = \left(1 + O_p\left(n^{-\frac{1}{2}} \rho_n^{-\frac{3}{2}}\right)\right) (1 + {\mathbf{o}}(1)) B_{ll^\prime} \psi_i \psi_j = \left(1 + O_p\left(n^{-\frac{1}{2}} \rho_n^{-\frac{3}{2}}\right)\right) P_{ij}. \label{eqn:V-pdiff}
\end{align}

Let $\S^\prime_{pq} \coloneqq \left(S_p \cap \mathcal{S}_{\frac13p} \right) \times \left(S_q \cap \mathcal{S}_{\frac13q} \right)$. Then (\ref{eqn:V-pdiff}) holds uniformly over all $(i, j) \in \S^\prime_{pq} \cap T_{l, l^\prime, l, l^\prime}$, using a union bound on (\ref{eqn:psi4}) over $K^2$ pairs of communities $(l, l^\prime) \in [K]^2$ and the fact that (\ref{eqn:psi_prime}) always holds for the node pairs $(i, j) \in \S^\prime_{pq}$. Thus,
\begin{align}
    \max\limits_{(i,j) \in \S^\prime_{pq} \cap T_{l,l^\prime,l, l^\prime}} \left\lvert \hat{P}_{ij}^{\prime(K)} - P_{ij} \right\rvert = O_p\left(n^{-\frac{1}{2}} \rho_n^{-\frac{3}{2}}\right)P_{ij} = O_p\left(n^{-\frac{1}{2}} \rho_n^{-\frac{1}{2}}\right) \ \mbox{for any }\ 1 \leq p < q \leq s. \label{eqn:V-pdiff-max}
\end{align}

From (\ref{eqn:math_s_size}), for any $q \in [s]$,
\begin{align}
    \left\lvert \mathcal{S}_{\frac13q} \right\rvert = \lvert S_{0q} \rvert - \left\lvert \mathcal{S}_{\frac13q}^c \right\rvert \geq (o+m) -  C^\prime (o+m)n^{-\frac{1}{3}} \geq (o+m) (1 - C^\prime n^{-\frac{1}{3}}) \geq C_1 n(1 - {\mathbf{o}}(1)).
\end{align}
From Lemma \ref{lemma:maxmincomm}, $\lvert S_q \cap G_l \rvert \geq m\gamma/(1+\gamma)$. Then using a similar argument as in (\ref{eqn:hyper-tail}), for any $q \in [s]$, $\left\lvert S_q \cap \mathcal{S}_{\frac{1}{3}q} \cap G_l \right\rvert = \Omega_p(m)$ as $n \to \infty$. Thus, $c_1 m^2 \leq \lvert \S^\prime_{pq} \cap T_{ll^\prime} \rvert \leq \lvert S_p \times S_q \rvert = m^2 $ for some constant $c_1 > 0$ and sufficiently large $n$ for any $1 \leq p < q \leq s$. Then using Lemma \ref{lemma:sum-aij},
\begin{align}
    \sum\limits_{(i,j) \in \S^\prime_{pq} \cap T_{l,l^\prime,l,l^\prime}} A_{ij} \leq \sum\limits_{(i,j) \in \S^\prime_{pq} \cap T_{ll^\prime}} A_{ij} = O_p(m^2\rho_n). \label{eqn:V-aijsum}
\end{align}

For any $(k, k^\prime), (l, l^\prime) \in [K]^2$ such that $(k, k^\prime) \neq (l, l^\prime)$ and $1 \leq p < q \leq s$, using a similar argument as in Figure \ref{fig:DCBM_cardinal} but for the corresponding subsets of the test set and Theorem \ref{DCBM_bound}, we get
\begin{align}
    \lvert T_{k, k^\prime, l, l^\prime}^{(pq)} \rvert = \lvert \hat{T}_{k, k^\prime}^{(pq)} \cap T_{l, l^\prime}^{(pq)} \rvert \leq \lvert \hat{T}_{l, l^\prime}^{(pq)} \Delta T_{l, l^\prime}^{(pq)} \rvert = O_p\left( m n^{\frac{1}{2}} \rho_n^{-\frac{1}{2}} \right). \label{eqn:VI-tll-size}
\end{align}
Additionally, for any $(k, k^\prime) \neq (l, l^\prime)$, consider 
\begin{align}
    \left\lvert \hat{B}^\prime_{kk^\prime} - (o+m) B_{ll^\prime} \right\rvert &\leq \left\lvert \hat{B}^\prime_{kk^\prime} - (o+m)B_{kk^\prime} \right\rvert + (o+m) B_{kk^\prime} + (o+m)B_{ll^\prime} \nonumber\\ 
    &= O_p\left( (o+m) n^{-\frac{1}{2}} \rho_n^{-\frac{1}{2}} \right) + O\left((o+m)\rho_n\right)  = O_p\left((o+m) \rho_n \right),
\end{align}
as $\rho_n = \omega(n^{-1/3})$ and the first $O_p$ term follows from \eqref{eqn:psi4}. 
For any $(i, j) \in \mathcal{S}^\prime_{pq} \cap T_{k, k^\prime, l, l^\prime}^{(pq)}$ such that $(k, k^\prime) \neq (l, l^\prime)$ and $1 \leq p < q \leq s$,
\begin{align}
    \left\lvert \hat{P}_{ij}^{\prime(K)} - P_{ij} \right\rvert &= \left\lvert \hat{B}^\prime_{kk^\prime} \hat{\psi}^\prime_i \hat{\psi}^\prime_j - B_{ll^\prime} \psi_i \psi_j \right\rvert \nonumber\\
    &\leq \left\lvert \hat{B}^\prime_{kk^\prime} - (o+m)B_{ll^\prime} \right\rvert\hat{\psi}^\prime_i \hat{\psi}^\prime_j + B_{ll^\prime}\left\lvert (o+m)\hat{\psi}^\prime_i \hat{\psi}^\prime_j - \psi_i \psi_j \right\rvert \nonumber\\
    &= O_p\left( (o+m)\rho_n \right) (1 + {\mathbf{o}}(1))(o+m)^{-1} + O(\rho_n) (o+m)(1+{\mathbf{o}}(1)) (o+m)^{-1} \nonumber\\
    &= O_p(\rho_n).
\end{align}
The above bound is uniform over all $(i,j) \in \mathcal{S}^\prime_{pq} \cap T_{k, k^\prime, l, l^\prime}^{(pq)}$ for any $1 \leq p < q \leq s$ using uniform bound over all possible $(k, k^\prime) \neq (l, l^\prime)$ and using the fact that (\ref{eqn:psi_prime}) holds surely for $(i,j) \in \mathcal{S}^\prime_{pq}$. Thus,
\begin{align}
    \max\limits_{(i,j) \in \mathcal{S}^\prime_{pq} \cap T_{k, k^\prime, l, l^\prime}^{(pq)}} \left\lvert \hat{P}_{ij}^{\prime(K)} - P_{ij} \right\rvert = O_p(\rho_n). \label{eqn:VI-pijbound}
\end{align}

Also define $\S^{\prime\prime}_{pq} \coloneqq \left( S_p \setminus \mathcal{S}_{\frac13p} \right) \times \left( S_q \setminus \mathcal{S}_{\frac13q} \right)$. Then from (\ref{eqn:math_s_size}),
\begin{align}
    &\left\lvert S_p \setminus \mathcal{S}_{\frac13p} \right \rvert \leq \left\lvert \mathcal{S}_{\frac13p}^c \right \rvert \leq C^\prime (o+m)n^{-\frac{1}{3}} = O(n^{\frac{2}{3}})\ \mbox{for any}\ p \in [s],\nonumber\\
    \mbox{and}\ &\left\lvert \S^{\prime\prime}_{pq} \right\rvert = O\left(n^{\frac{4}{3}}\right).\label{eqn:VII-spp-size}
\end{align}

Note that $S_p \times S_q = \S^\prime_{pq} \cup \S^{\prime\prime}_{pq}$ for any $p, q \in [s]$. Now, consider the absolute difference between the computed loss and the oracle loss as
\begin{align}
    &\left\lvert L\left(A_{\S^c}, \hat{P}^{\prime(K)}_{\S^c}\right) - L\left(A_{\S^c}\right) \right\rvert\nonumber \\
    = &\left\lvert \sum\limits_{(i,j) \in \S^c} \left( \left(A_{ij} - \hat{P}^{\prime(K)}_{ij}\right)^2 - \left(A_{ij} - P_{ij}\right)^2 \right) \right\rvert \nonumber\\
    \leq& \sum\limits_{(i,j) \in \S^c} \left\lvert A_{ij}^2 -2A_{ij} \hat{P}^{\prime(K)}_{ij} + \hat{P}^{\prime(K)2}_{ij} - A_{ij}^2 + 2A_{ij}P_{ij} - P_{ij}^2 \right\rvert \nonumber\\
    =& \sum\limits_{(i,j) \in \S^c} \left\lvert -2A_{ij}(\hat{P}^{\prime(K)}_{ij} - P_{ij}) + (\hat{P}^{\prime(K)}_{ij} + P_{ij})(\hat{P}^{\prime(K)}_{ij} - P_{ij}) \right\rvert \nonumber\\
    \leq& \sum\limits_{1 \leq p < q \leq s} \sum\limits_{(i,j) \in S_p \times S_q} \left\lvert \hat{P}^{\prime(K)}_{ij} - P_{ij} \right\rvert \left( 2A_{ij} + 2P_{ij} + \left\lvert\hat{P}^{\prime(K)}_{ij} - P_{ij}\right\rvert \right) \nonumber\\
    \leq& \sum\limits_{1 \leq p < q \leq s} \sum\limits_{(i,j) \in \S^\prime_{pq}} \left\lvert \hat{P}^{\prime(K)}_{ij} - P_{ij} \right\rvert \left( 2A_{ij} + 2P_{ij} + \left\lvert\hat{P}^{\prime(K)}_{ij} - P_{ij}\right\rvert \right) \nonumber\\
    &\hspace{0.25cm} + \sum\limits_{1 \leq p < q \leq s} \sum\limits_{(i,j) \in \S^{\prime\prime}_{pq}} \max_{(i,j) \in \S^{\prime\prime}_{pq}}\left\lvert \hat{P}^{\prime(K)}_{ij} - P_{ij} \right\rvert \left( 2A_{ij} + 2P_{ij} + \max_{(i,j) \in \S^{\prime\prime}_{pq}}\left\lvert\hat{P}^{\prime(K)}_{ij} - P_{ij}\right\rvert \right) \nonumber\\
    =& \sum\limits_{1 \leq p < q \leq s} \sum\limits_{(l, l^\prime) \in [K]^2} \sum\limits_{(i,j) \in \S^{\prime}_{pq} \cap T_{ll^\prime}^{(pq)}} \max_{(i,j) \in \S^{\prime}_{pq} \cap T_{ll^\prime}^{(pq)} }\left\lvert \hat{P}^{\prime(K)}_{ij} - P_{ij} \right\rvert \nonumber\\
    & \hspace{6cm} \times \left( 2A_{ij} + 2P_{ij} + \max_{(i,j) \in \S^{\prime}_{pq} \cap T_{ll^\prime}^{(pq)}}\left\lvert\hat{P}^{\prime(K)}_{ij} - P_{ij}\right\rvert \right) \nonumber\\
    &\hspace{0.25cm} + \sum\limits_{1 \leq p < q \leq s} \mathop{\sum\limits_{(k, k^\prime)\neq (l, l^\prime)}}\limits_{\in [K]^2} \sum\limits_{(i,j) \in \S^{\prime}_{pq} \cap T_{k,k^\prime,l,l^\prime}^{(pq)}} \max_{(i,j) \in \S^{\prime}_{pq} \cap T_{k,k^\prime,l,l^\prime}^{(pq)} }\left\lvert \hat{P}^{\prime(K)}_{ij} - P_{ij} \right\rvert \nonumber \\
    & \hspace{7cm} \times \left( 2A_{ij} + 2P_{ij} +\max_{(i,j) \in \S^{\prime}_{pq} \cap T_{k,k^\prime,l,l^\prime}^{(pq)}}\left\lvert\hat{P}^{\prime(K)}_{ij} - P_{ij}\right\rvert \right) \nonumber\\
    &\hspace{0.25cm} + \sum\limits_{1 \leq p < q \leq s} \sum\limits_{(i,j) \in \S^{\prime\prime}_{pq}} \max_{(i,j) \in \S^{\prime\prime}_{pq}}\left\lvert \hat{P}^{\prime(K)}_{ij} - P_{ij} \right\rvert \left( 2A_{ij} + 2P_{ij} + \max_{(i,j) \in \S^{\prime\prime}_{pq}}\left\lvert\hat{P}^{\prime(K)}_{ij} - P_{ij}\right\rvert \right) \nonumber\\
    &= V + VI + VII \ \mbox{(say)}.\label{eqn:dc_lossdiff}
\end{align}

Combining (\ref{eqn:V-pdiff-max}) and (\ref{eqn:V-aijsum}), we get
\begin{align}
    V = \frac{s(s-1)}{2} O_p\left( n^{-\frac{1}{2}} \rho_n^{-\frac{1}{2}} m^2 \rho_n \right) = O_p\left( s(s-1)m^2 n^{-\frac{1}{2}} \rho_n^{\frac{1}{2}} \right).\label{eqn:V-DC}
\end{align}

Combining (\ref{eqn:VI-tll-size}), (\ref{eqn:VI-pijbound}) and the fact that $m \asymp n$, we get
\begin{align}
    VI &\leq \sum\limits_{1\leq p < q \leq s} \max\limits_{(k, k^\prime) \neq (l, l^\prime)}\left\lvert T_{k, k^\prime, l, l^\prime}^{(pq)} \right\rvert O_p(\rho_n) = O_p\left( s(s-1) mn^{\frac{1}{2}} \rho_n^{-\frac{1}{2}} \rho_n \right) \nonumber\\
    &= O_p\left( s(s-1)m^2 n^{-\frac{1}{2}} \rho_n^{\frac{1}{2}} \right).\label{eqn:VI-DC}
\end{align}

From (\ref{eqn:VII-spp-size}) and $\rho_n = \omega(n^{-1/3})$, we get
\begin{align}
    VII \leq \frac{s(s-1)}{2} \max\limits_{1 \leq p < q \leq s}\left\lvert \S^{\prime\prime}_{pq} \right\rvert = O\left(s(s-1) n^\frac{4}{3}\right) = {\mathbf{o}}\left( s(s-1)m^2 n^{-\frac{1}{2}} \rho_n^{\frac{1}{2}} \right).\label{eqn:VII-DC}
\end{align}

Combining (\ref{eqn:V-DC}), (\ref{eqn:VI-DC}) and (\ref{eqn:VII-DC}), we get
\begin{equation}
    \left\lvert L\left(A_{\S^c}, \hat{P}^{\prime(K)}_{\S^c}\right) - L\left(A_{\S^c}\right) \right\rvert = O_p\left( s(s-1)m^2 n^{-\frac{1}{2}} \rho_n^{\frac{1}{2}} \right).
\end{equation}

\subsection{Proof of Theorem \ref{theorem:RDPG_main}}
We show that the computed loss is closely concentrated around the oracle loss for $\d = d$ in Case 1 of this proof and much larger than the oracle loss for $\d = d-1$ in Case 2 of this proof. An induction-like argument follows for the cases $\d \leq d-2$. First, we state some results used in the proof. First, we state some of the results that will be used in the main proof. The following theorem is a slightly modified version of \citet[Theorem~26]{RDPG_survey}.
\begin{theorem}\label{theorem:RDPG_Athreya_bound}
    Let $A \sim RDPG(X)$ with $n\geq 1$ nodes be a sequence of random dot product graphs where $X$ is assumed to be of rank $d$ for all $n$ sufficiently large. Denote by $\hat{X}$ the adjacency spectral embedding of $A$ using Algorithm \ref{algo:ASE}. Then under the conditions \ref{assump:RDPG:p-define} and \ref{assump:RDPG:rho-cond} of Assumption \ref{assump:RDPG}, with probability tending to $1$ as $n \to \infty$, there exists a fixed constant $C>0$ and an orthogonal transformation $W \in \O_d$ such that 
    \begin{align}
        \max_{i \in [n]} \left\lVert \hat{X}_{i\cdot} - W X_{i\cdot} \right\rVert \leq \frac{C\sqrt{d} \log^2{n}}{\sqrt{\delta(P)}},
    \end{align}
    where $\delta(P) \coloneqq \max\limits_{i \in [n]} \sum\limits_{j \in [n]} P_{ij}$.
\end{theorem}
\begin{proof}
    Let $\lambda_d$ be the $d$\textsuperscript{th} largest eigenvalue of $P$. Then \citet[Theorem~26]{RDPG_survey} requires the existence of some constants $a_0 > 0$ and $c_0 > 0$ such that
    \begin{align}
        \delta(P) \geq \log^{4+a_0}{(n)} \ \mbox{ and }\ \lambda_d \geq c_0 \delta(P). \label{eqn:athreya_assump}
    \end{align}
    
    From Condition \ref{assump:RDPG:p-define} of Assumption \ref{assump:RDPG}, there exist a constant $c_1$ such that 
    \begin{align}
        \delta(P) = \rho_n \max_{i \in [n]} \sum_{j \in [n]} {P^0_{n, ij}} \geq c_1 n\rho_n \geq c_1 n \rho_n^2 \geq c_1 \log^{4+a_0}(n). \label{eqn:delta_bound}
    \end{align}
    Also, note that from Condition \ref{assump:RDPG:eigen} of Assumption \ref{assump:RDPG}, we have 
    \begin{align}
        \delta(P) \leq c_2 n\rho_n \ \ \mbox{and}\ \     \lambda_d = \rho_n \lambda^0_d \geq \phi^{-1} n\rho_n \geq \phi^{-1} c_2^{-1} \delta(P). \label{eqn:lambda_d-bound}
    \end{align}
\end{proof}

\begin{theorem}\label{theorem:subset-any-ortho}
    Let $\mathcal{S}$ be a random subset of nodes of size $\eta$ of $A \sim RDPG(X)$ with $n \geq 1$ nodes and sparsity parameter $\rho_n$. If $\hat{X}^{\mathcal{S}}$ is the adjacency spectral embedding of $A_{\mathcal{S}}$, then under Assumption \ref{assump:RDPG} and for $\eta$ such that $\eta/n$ is bounded away from $0$, there exist constants $C$, $C^\prime$ and an orthogonal matrix $W$ such that the following bound holds with probability tending to $1$ as $\eta \to \infty$
    \begin{align}
         \max_{i \in \mathcal{S}} \left\lVert \hat{X}^{\mathcal{S}}_{i\cdot} - W X_{i\cdot} \right\rVert \leq \frac{C\sqrt{d} \log^2{\eta}}{\sqrt{\eta \rho_n}}.
    \end{align}
\end{theorem}

\begin{proof}
    The theorem follows from the fact that the subnetwork with adjacency matrix $A_{\mathcal{S}}$ is still an RDPG and the estimated latent positions from adjacency spectral embedding on $A_{\mathcal{S}}$ satisfy the bound in Theorem \ref{theorem:RDPG_Athreya_bound}, provided the subnetwork satisfies the conditions in (\ref{eqn:athreya_assump}). Consider
    \begin{align}
        \delta(P_{\mathcal{S}}) = \max_{i \in \mathcal{S}} \sum_{j \in \mathcal{S}} P_{ij} \geq c_1 \eta \rho_n \geq c_1 \eta \rho_n^2 \geq c_1 c_0 \frac{\eta}{n} \log^{4+a_0}(n) \geq c^\prime \log^{4+a_0}(\eta),
    \end{align}
    as $\eta/n$ is bounded away from $0$ and for sufficiently large $\eta$. Let $\lambda_d(M)$ of a matrix $M$ denote the $d$-th largest eigenvalue of $M$. 
    From sublinear time algorithm result \citep[Theorem~1]{sublinear-time}, bounding eigenvalues of random principal submatrix, for any $\epsilon, \delta > 0 > 0$ and $\eta \geq \frac{c \log{(1/\epsilon \delta)} \log^3{n}}{\epsilon^3 \delta}$, the following holds with probability $\geq 1 - \delta$
    \begin{align}
        \lambda_d(P_{\mathcal{S}}) &\geq \frac{\eta}{n} \lambda_d(P) - \epsilon n\lVert P \rVert_{\infty}\\
        &\geq \frac{\eta}{n} \lambda_d(P) - c_2 \epsilon \eta \rho_n\\
        &\geq \phi^{-1} \eta \rho_n - c_2 \epsilon \eta \rho_n. \label{eqn:lambda_d}
    \end{align}
    We also have
    \begin{align}
        \delta(P_{\mathcal{S}}) = \max_{i \in \mathcal{S}} \sum_{j \in \mathcal{S}} P_{ij} \leq c_2 \eta \rho_n.\label{eqn:delta_p}
    \end{align}

    Take $\epsilon = \frac{\phi^{-1}}{2c_2}$. Then combining (\ref{eqn:lambda_d}) and (\ref{eqn:delta_p})
    \begin{align}
        \lambda_d(P_{\mathcal{S}}) \geq \frac{1}{2}\phi^{-1} c_2^{-1}\delta(P_{\mathcal{S}}).
    \end{align}
    
    Taking $\delta = n^{-1/2}$, the above holds with probability $ \geq 1 - \frac{1}{\sqrt{n}}$ for $\eta \geq \frac{c \sqrt{n} \log{(\sqrt{n}/\epsilon)} \log^3{n}}{\epsilon^3} \approx c^{''} \sqrt{n} \log^4{n}$, which is satisfied for all sufficiently large $n$ from theorem assumption.
    
    Thus, the assumptions of \citet[Theorem ~ 26]{RDPG_survey} are satisfied for adjacency spectral embedding on $A_{\mathcal{S}}$ and the bound in this theorem follows from Theorem \ref{theorem:RDPG_Athreya_bound}.
\end{proof}

\begin{theorem}\label{theorem:theo-stitch-mat}
    Let $A_{S_{01}}, \ldots, A_{S_{0s}}$ be the subadjacency matrices used in \code{NETCROP} with $s$ subgraphs and overlap size $o$. Assume that $\hat{X}_q$ is the estimated latent positions using adjacency spectral embedding on $A_{S_{0q}}$ for $q \in [s]$. Then under Assumption \ref{assump:RDPG}, there exist orthogonal matrices $W_1, \ldots, W_s\in \O_d$ and a constant $C > 0$ such that the following holds with probability going to $1$
    \begin{align}
        \max_{q \in [s]} \max_{i \in S_{0q}} \left\lVert W_q \hat{X}_{q,i\cdot} - X_{i \cdot} \right\rVert \leq \frac{C \sqrt{d} \log^2(o+m)}{\sqrt{(o+m) \rho_n}}.
    \end{align}
\end{theorem}

\begin{proof}
    From Theorem \ref{theorem:subset-any-ortho} for the $q$th subnetwork of \code{NETCROP}, there exist $H_q \in \O_d$ and a constant $C_q > 0$ such that 
    \begin{align}
        \max_{i \in S_{0q}} \left\lVert \hat{X}_{q,i\cdot} - H_q X_{i \cdot} \right\rVert \leq \frac{C_q \sqrt{d} \log^2(o+m)}{\sqrt{(o+m) \rho_n}}, \ q \in [s].
    \end{align}

    Set $W_1 = H_1^\top, \ldots, W_s = H_s^\top$. Then $W_q \in \O_d$ for any $q \in [s]$. This implies
    \begin{align}
        \max_{q \in [s]} \max_{i \in S_{0q}} \left\lVert W_q \hat{X}_{q,i\cdot} - X_{i \cdot} \right\rVert = \max_{q \in [s]} \max_{i \in S_{0q}} \left\lVert \hat{X}_{q,i\cdot} - H_q X_{i \cdot} \right\rVert \leq \frac{C \sqrt{d} \log^2(o+m)}{\sqrt{(o+m) \rho_n}},
    \end{align}
    where $C = max\{C_1, \ldots, C_s\}$. The above bound is uniform as $s$ is fixed.
\end{proof}

The following lemma states that the $d$-th signal in the eigen decomposition of $P$ is positive on a large subset of the test set. Define ${P^{0(\d)}_n} \coloneqq {U_{n, \cdot[\d]}} \Lambda^0_{[\d]} {U_{n, \cdot[\d]}^\top}$, the $\d$-truncated version of $P$ for any $\d < d$.
\begin{lemma}\label{lemma:rdpg_php}
    Let $A \sim RDPG(X)$ with $n$ nodes and latent space dimension $d$. Let {$P^{0(\d)}_n$} be the $\d$ rank truncated version of {$P^0_n$}. Then there exist constants $\phi^\prime, \kappa >0$ such that the following holds with very high probability.
    \begin{align*}
        \left\lvert\left\{(i,j)\in \S^c:{\left\lvert P_{n, ij}^0 - P_{n, ij}^{0(\d)} \right\rvert} \geq \phi^\prime \right\}\right\rvert \geq  \kappa s(s-1)m^2,
    \end{align*}
    for any $\d < d$, where $S^c$ is the test set of \code{NETCROP}.
\end{lemma}
\begin{proof}
By definition, for any $i, j \in [n]$, we have
\begin{align}
    {\left\lvert P_{n, ij}^0 - P_{n, ij}^{0(\d)}\right\rvert} &= \left\lvert \sum_{k = \d+1}^d \lambda_k^0 {U_{n, ik}U_{n, jk}} \right\rvert \nonumber\\
    &\leq \phi n \left\lvert \sum_{k = \d+1}^d {U_{n, ik}U_{n, jk}} \right\rvert \nonumber\\
    &\leq \phi n \left( \sum_{k=\d+1}^d {U_{n, ik}^2} \sum_{k=\d+1}^d {U_{n, jk}^2} \right)^{\frac{1}{2}} \nonumber\\
    &\leq \phi n  {\lVert U_{n, i\cdot} \rVert \lvert U_{n, j\cdot} \rVert}  \nonumber\\
    &\leq \phi n \frac{da}{n} = \phi d a. \label{eqn:rd-lowdiff}
\end{align}

From the same assumption, we also have 
\begin{align}
    \sum\limits_{i,j \in [n]} {\left\lvert P_{n, ij}^0 - P_{n, ij}^{0(\d)} \right\rvert^2 = \left\lVert P^0_n - P^{0(\d)}_n \right\rVert_F^2} = \sum\limits_{k = \d + 1}^d (\lambda_{k}^0)^2 \geq \phi^{-2}(d - \d) n^2.\label{eqn:rd-highdiff}
\end{align}
Thus, for (\ref{eqn:rd-lowdiff}) and (\ref{eqn:rd-highdiff}) to simultaneously occur, there must exist constants $\phi^\prime, \kappa^\prime > 0$ and a subset of node pairs $\Delta_{\d} \subset [n] \times [n]$ such that
\begin{align}
    \left\lvert \Delta_{\d} \right\rvert \coloneqq \left\lvert \left\{ (i,j) : i\leq j, {\left\lvert P_{n, ij}^0 - P_{n, ij}^{0(\d)} \right\rvert} \geq \phi^\prime \right\} \right\rvert \geq \kappa^\prime n^2. \label{eqn:delta_set}
\end{align}
For each node $i \in [n]$, define $T_i \coloneqq \{j \in [n] : 
(i,j) \in \Delta_{\d} \}$. Then $\Delta_{\d} = \bigcup\limits_{i \in [n]} \{i\} \times T_i$. For (\ref{eqn:delta_set}) to hold, there must exist a subset of nodes $T \subset [n]$ satisfying $\lvert T \rvert \geq \kappa_1 n$ for some $\kappa_1 > 0$ such that $\lvert T_i \rvert \geq \kappa_2 n$ for all $i \in T$.

Using a similar hypergeometric tail bound as in (\ref{eqn:hyper-tail}), replacing $G_l$ by $T$ and $T_i$ and for any non-overlap part $S_p$, $p \in [s]$, we have
\begin{align}
    \lvert T \cap S_p \rvert \geq \frac{\kappa_1 m}{2} \ \ \mbox{and} \ \ \lvert T_i \cap S_p \rvert \geq \frac{\kappa_2 m}{2} \ \ \mbox{for all}\ \ i \in A,
\end{align}
both with probability going to $1$ as $m \to \infty$. For any $1 \leq p < q \leq s$, consider the size of the following subset of nodes in the test set
\begin{align}
    \left\lvert \bigcup\limits_{i \in T \cap S_p} \{i\} \times (T_i \cap S_q) \right\rvert &= \sum\limits_{i \in T \cap S_p} \left\lvert T_i \cap S_q \right\rvert\nonumber\\
    &\geq \lvert T \cap S_p \rvert \frac{\kappa_2 m}{2} \ \ \mbox{with high probability}\nonumber\\
    &\geq \frac{\kappa_1 \kappa_2 m^2}{4} \ \ \mbox{with high probability}.
\end{align}

As the test set $\S^c$ is given by the disjoint union $\S^c = \bigcup\limits_{1 \leq p < q \leq s} (S_p \times S_q)$, the subset of the test set in the lemma is given by
\begin{align}
    T_{\S^c} \coloneqq \bigcup\limits_{1 \leq p < q \leq s}\bigcup\limits_{i \in T \cap S_p} \{i\} \times (T_i \cap S_q).
\end{align}
It satisfies $\lvert T_{\S^c} \rvert \geq \binom{s}{2}\frac{\kappa_1 \kappa_2 m^2}{4}$ with probability going to $1$ as $m \to \infty$ and for every $(i,j) \in T_{\S^c}$, ${\lvert P_{n, ij}^0 - P_{n, ij}^{0(\d)} \rvert} \geq \phi^\prime$. 
\end{proof}

\paragraph*{Proof of Theorem \ref{theorem:rdpg_lossdiff}}

\paragraph*{Case 1: ($\d = d$)}
Before deriving the upper bound for the difference between the computed loss and the oracle loss, we bound the norm of the $i$th latent position, which will be used in multiple places in this proof.
    \begin{align}
    \lVert X_{i\cdot} \rVert = \left( \sum_{k=1}^d X_{ik}^2 \right)^\frac{1}{2} = \sqrt{\rho_n} \left( \sum_{k=1}^d \lambda_k^0 {U_{n, ik}^2} \right)^\frac{1}{2} \leq \sqrt{\rho_n \lambda_1^0 \frac{d}{n}} \left(\max_{i} \frac{n}{d}\lVert {U_{n, i\cdot}} \rVert^2 \right)^\frac{1}{2} \leq  \sqrt{\phi a \rho_n d}. \label{eqn:norm-x_i}
\end{align}

Consider the absolute difference between the computed loss and the oracle loss
\begin{align}
    &\left\lvert L(A, \hat{P}^{(d)}) - L(A) \right\rvert
    \leq& \sum\limits_{(i,j) \in \S^c} \max_{(i,j) \in \S^c}\left\lvert \hat{P}^{(d)}_{ij} - P_{ij} \right\rvert \left( 2A_{ij} + 2P_{ij} + \max_{(i,j) \in \S^c}\left\lvert\hat{P}^{(d)}_{ij} - P_{ij}\right\rvert \right).\label{eqn:rdpg_lossdiff}
\end{align}

We bound $\left\lvert \hat{P}^{(d)}_{ij} - P_{ij} \right\rvert$ on the test set $\S^c$. Let $i \in S_p$ and $j \in S_q$ for some $1 \leq p < q \leq s$. Then from Theorem \ref{theorem:theo-stitch-mat}, there exist $W_p, W_q \in \O_d$ such that
\begin{align}
     \left\lvert \hat{P}^{(d)}_{ij} - P_{ij} \right\rvert &= \left\lvert \left(W_p\hat{X}_{p,i\cdot}\right)^\top \left(W_q\hat{X}_{q,j\cdot}\right) - X_{i \cdot}^\top X_{j \cdot} \right\rvert\nonumber\\
    &\leq \left\lvert \left(W_p\hat{X}_{p,i\cdot} - X_{i \cdot}\right)^\top \left(W_q\hat{X}_{q,j\cdot} -  X_{j \cdot}\right) \right\rvert \nonumber\\
    &\hspace{1cm} + 
    \left\lvert X_{i \cdot}^\top \left(W_q\hat{X}_{q,j\cdot} - X_{j \cdot}\right) \right\rvert + \left\lvert \left(W_p\hat{X}_{p,i\cdot} - X_{i \cdot}\right)^\top X_{j \cdot} \right\rvert\nonumber\\
    &\leq \left\lVert W_p \hat{X}_{p,i\cdot} - X_{i \cdot}\right\rVert \left\lVert W_q \hat{X}_{q,j\cdot} - X_{j \cdot} \right\rVert \nonumber \\
    &\hspace{1cm} + 
    \left\lVert X_{i\cdot} \right\rVert \left\lVert W_q \hat{X}_{q,j\cdot} - X_{j \cdot} \right\rVert + \left\lVert W_p \hat{X}_{p,i\cdot} - X_{i \cdot} \right\rVert \left\lVert X_{j\cdot} \right\rVert \nonumber \\
    &= O_p\left(\sqrt{\rho_n} \frac{\log^2{(o+m)}}{\sqrt{(o+m) \rho_n}} \right) \ \mbox{(From Theorem \ref{theorem:theo-stitch-mat} and \eqref{eqn:norm-x_i})} \nonumber\\
    &= O_p\left(\frac{\log^2{(o+m)}}{\sqrt{(o+m)}}\right).\label{eqn:rdpg-p-d}
\end{align}

Since $1 \leq p < q \leq s$ are arbitrary and the bound in Theorem \ref{theorem:theo-stitch-mat} is uniform over all subsets $S_{0q}$, we have
\begin{align}\label{eqn:rdpg_prob_bound}
    \max_{(i,j) \in \S^c} \left\lvert\hat{P}^{(d)}_{ij} - P_{ij}\right\rvert = O_p\left( \frac{\log^2(o+m)}{\sqrt{o+m}} \right).
\end{align}

Continuing from (\ref{eqn:rdpg_lossdiff}), we have
\begin{align}
    \left\lvert L(A, \hat{P}^{(d)}) - L(A) \right\rvert
    &= O_p \left( \lvert \S^c \rvert\rho_n \frac{\log^2(o+m)}{ \sqrt{o+m}} \right) = O_p\left(s(s-1)m^2 \rho_n \frac{\log^2(o+m)}{ \sqrt{o+m}} \right).
\end{align}

\paragraph*{Case 2: ($\d = d-1$)}
Recall that $P^{(\d)}$ is the rank $\d$ truncated version of $P$, defined as
\begin{align*}
    P^{(\d)} = {U_{n, \cdot [\d]}} \Lambda_{[\d]} {U_{n, \cdot [\d]}^\top} = X_{\cdot [\d]} X_{\cdot [\d]}^\top,
\end{align*}
where $P = {U_n} \Lambda {U_n^\top}$ is the eigen decomposition of $P$ and $\Lambda = \rho_n \Lambda^{0}$. Then we have
\begin{align}
    \left\lvert \hat{P}_{ij}^{(d-1)} - P_{ij} \right\rvert &= \left\lvert \hat{P}_{ij}^{(d-1)} - P_{ij}^{(d-1)} + P_{ij}^{(d-1)} - P_{ij} \right\rvert \nonumber\\
    &\geq \left\lvert P_{ij} - P_{ij}^{(d-1)} \right\rvert - \left\lvert \hat{P}_{ij}^{(d-1)} - P_{ij}^{(d-1)} \right\rvert \nonumber\\
    &= I - II, \ \mbox{(say)}. 
\end{align}

From Lemma \ref{lemma:rdpg_php}, there exists $\Delta_{d-1}$
\begin{align}
    I = \left\lvert P_{ij} - P_{ij}^{(d-1)} \right\rvert = \rho_n \left\lvert {P_{n, ij}^0} - {P_{n, ij}^{0(d-1)}} \right\rvert \geq \phi^\prime \rho_n \ \mbox{ for all } (i,j) \in \Delta_{d-1} \cap S^c. \label{eq:rdpg-t1}
\end{align}

Next, we consider
\begin{align}
    \left\lvert \hat{P}_{ij}^{(d-1)} - P_{ij}^{(d-1)} \right\rvert &= \left\lvert \left(W_p \hat{X}_{p, i \cdot}\right)_{[d-1]}^\top \left(W_q \hat{X}_{q, j \cdot}\right)_{[d-1]} - X_{i [d-1]}^\top X_{j [d-1]}  \right\rvert \nonumber\\
    &\leq \left\lVert \left(W_p \hat{X}_{p, i \cdot}\right)_{[d-1]} - X_{i [d-1]} \right\rVert \left\lVert \left(W_q \hat{X}_{q, j \cdot}\right)_{[d-1]} - X_{j [d-1]} \right\rVert \nonumber\\
    & \hspace{1cm} + \left\lVert X_{i[d-1]} \right\rVert \left\lVert \left(W_q \hat{X}_{q, j \cdot}\right)_{[d-1]} - X_{j [d-1]} \right\rVert \nonumber\\
    &\hspace{1cm} + \left\lVert \left(W_p \hat{X}_{p, i \cdot}\right)_{[d-1]} - X_{i [d-1]} \right\rVert \left\lVert X_{j[d-1]} \right\rVert \nonumber\\
    &\leq \left\lVert W_p \hat{X}_{p, i \cdot} - X_{i \cdot} \right\rVert \left\lVert W_q \hat{X}_{q, j \cdot} - X_{j \cdot} \right\rVert \nonumber\\
    & \hspace{1cm} + \left\lVert X_{i \cdot} \right\rVert \left\lVert W_q \hat{X}_{q, j \cdot} - X_{j \cdot} \right\rVert + \left\lVert W_p \hat{X}_{p, i \cdot} - X_{i \cdot} \right\rVert \left\lVert X_{j \cdot} \right\rVert \nonumber\\
    &= O_p\left(\frac{\log^2{(o+m)}}{\sqrt{(o+m)}}\right). \ \mbox{(from \eqref{eqn:rdpg-p-d})} \label{eq:rdpg-t2}
\end{align}

Thus,
\begin{align}
    \min_{(i,j) \in \Delta_{d-1} \cap \S^c} \left\lvert \hat{P}^{(d-1)}_{ij} - P_{ij} \right\rvert = \Omega_p\left( \rho_n \right) - O_p\left( \frac{\log^2(o+m)}{\sqrt{o+m}} \right) = \Omega_p\left( \rho_n \right),
\end{align}
as $\rho_n = \Omega\left( \frac{\log^{2 + a_0/2}(n)}{\sqrt{n}} \right)$ and $o+m \asymp n$.

Using Bernstein's inequality, the following holds with very high probability, where the expectation is over the test entries
\begin{align}
    L(A, \hat{P}^{(d-1)}) - L(A) &\geq \frac{1}{2} E_{\S^c}\left[ L\left(A, \hat{P}^{(d-1)}\right) - L(A) \right] \nonumber\\
    &= \frac{1}{2} \sum_{(i,j) \in \S^c} E_{\S^c}\left[ \left(A_{ij} - \hat{P}_{ij}^{(d-1)}\right)^2 - (A_{ij} - P_{ij})^2 \right] \nonumber\\
    &= \frac{1}{2} \sum_{(i,j) \in \S^c} E_{\S^c}\left[ A_{ij}^2 - 2 A_{ij}\hat{P}_{ij}^{(d-1)} + \left(\hat{P}_{ij}^{(d-1)}\right)^2 - A_{ij}^2 + 2 A_{ij} P_{ij} - P_{ij}^2 \right] \nonumber\\
    &= \frac{1}{2} \sum_{(i,j) \in \S^c} \left[ -2 P_{ij} \hat{P}_{ij}^{(d-1)} + \left(\hat{P}_{ij}^{(d-1)}\right)^2 + 2 P_{ij}^2 - P_{ij}^2 \right] \nonumber\\
    &= \frac{1}{2} \sum\limits_{(i,j) \in \S^c} \left( \hat{P}^{(d-1)}_{ij} - P_{ij} \right)^2 \nonumber\\
    &\geq \frac{1}{2} \sum\limits_{(i,j) \in \Delta_{d-1} \cap \S^c} \left( \hat{P}^{(d-1)}_{ij} - P_{ij} \right)^2 \nonumber\\
    &= \Omega_p\left( s(s-1)m^2 \rho_n^2 \right).
\end{align}

\paragraph*{Case 3: $(\d \leq d-2)$}
For $\d \leq d-2$, a similar argument can be used as in the $\d = d-1$ case, the only difference being the bounds in \eqref{eq:rdpg-t1}, \eqref{eq:rdpg-t2} and the subsequent terms differing by a factor of $d$. Since $d$ is fixed in terms of $n$, the asymptotic bounds for this case are the same as $\d = d-1$ case.

    \section{Additional Numerical Results for Small Networks}\label{app:numerical:small}
In this section, we compare the performance of \code{NETCROP} on small networks with the competing algorithms \code{NCV} and \code{ECV}. We use the problem of selecting the number of communities and the presence of degree correction in blockmodels for this comparison. We generated $100$ networks from two true models --- SBM and DCBM, both with $n \in \{200, 300, 400, 500\}$ nodes and $K = 3$ communities. We followed a similar network generation mechanism as in the SBM and DCBM examples for large networks with the out-in ratio $\beta = 0.3$ and $\alpha$ chosen in a way so that the expected average degree is approximately $n/4$. We applied \code{NETCROP} with $p_{test} = 10 \%$ that resulted in the number of subgraphs $s = 3$ for all examples in this section and $o = 113, 168, 223,$ and $278$ for $n = 200, 300, 400,$ and $500$, respectively. We applied \code{NETCROP} with $R \in \{1, 3, 5\}$ repetitions. We also applied \code{NCV} with $3$ folds and \code{ECV} with $3$ folds and $10\%$ held-out set. We applied the stabilized versions of \code{NCV} and \code{ECV} with $R = 20$ repetitions, as recommended by their authors. We used a candidate set of $\{SBM, DCBM\} \times [5]$ for all applications. Figure \ref{fig:small-plot} contains model selection accuracy, logarithm of mean runtime in seconds and RAM usage in Mebibytes (MiB), measured with R package peakRAM \citep{quinn2017peakRAM}, averaged across $100$ simulations, plotted against the network size for each variant of the algorithms. The top row of the figure corresponds to the SBM case and the bottom row corresponds to the DCBM case. All algorithms for these simulations were run serially with $1$ processor and $8$ gigabytes of RAM.

\begin{figure}[!tb]
    \centering
    \includegraphics[width=\linewidth]{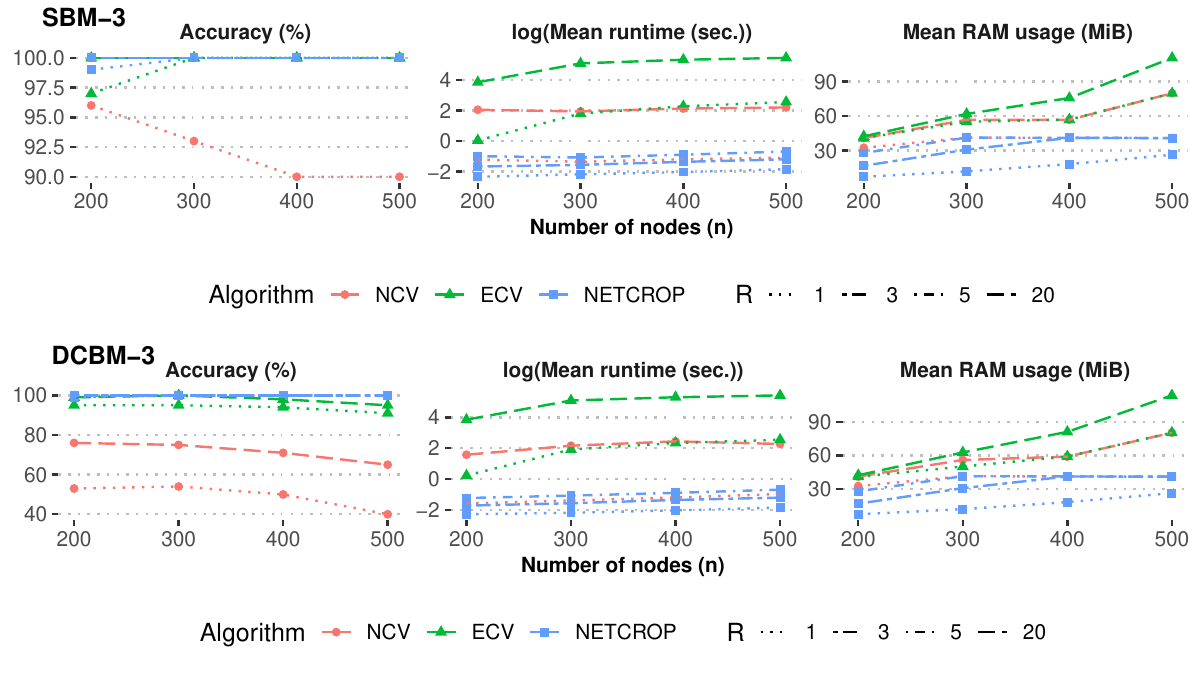} 

    \caption{\normalsize{Accuracy, logarithm of mean runtime in seconds and mean RAM usage in Mebibytes (MiB) of \code{NETCROP}, \code{NCV} and \code{ECV} against the network size for small networks. The networks in the top and bottom rows were generated from SBM and DCBM with $K = 3$ communities, respectively. \code{NETCROP} was applied with $R \in \{1, 3, 5\}$ repetitions, and \code{NCV} and \code{ECV} with $R = 1$ and $R = 20$ (stabilized).}}
    \label{fig:small-plot}

    {\raggedright
    \noindent\textbf{Alt text:}
    Six line charts compare accuracy, log mean runtime, and mean RAM usage as network size increases from 200 to 500 nodes for SBM and DCBM networks. NETCROP maintains nearly 100 percent accuracy in both settings and consistently has the lowest runtime and memory use. ECV is similarly accurate but generally requires the most runtime and memory, especially with 20 repetitions. NCV has significantly lower accuracy, especially for DCBM, while its computational costs lie between those of NETCROP and ECV. Runtime and memory generally increase with network size.
    \par}
\end{figure}

From Figure \ref{fig:small-plot}, we observe that \code{NETCROP} with $R \in \{1, 3, 5\}$ repetitions has near $100 \%$ accuracies in all cases for all network sizes. The accuracies of \code{ECV} and stabilized \code{ECV} are close to $100 \%$ for the SBM case for $n \geq 300$, while for DCBM they are slightly less than $100 \%$ with slight improvement with stabilization. 
For the SBM examples, \code{NCV} has accuracies in the $90 - 96\%$ range that goes up to $100\%$ with stabilization. However, for the DCBM cases, both \code{NCV} and its stabilized version had significantly lower accuracies (approximately $40 - 80 \%$).
\code{NETCROP} used significantly less computational resources in terms of runtime and RAM usage in comparison with \code{NCV} and \code{ECV}. Although \code{NCV} with $R = 1$ used similar resources as \code{NETCROP} with $R = 5$, \code{NCV} had poorer accuracies. \code{NCV} with stabilization and both variants of \code{ECV} were several orders of magnitude slower while using more RAM than all variants of \code{NETCROP}. The gains in computation time and RAM usage for \code{NETCROP} are observed to increase with the network size. Overall, \code{NETCROP} outperformed its competing algorithms while being several times faster and more efficient with the resources even for small networks.
\end{appendices}

\end{document}